\documentclass[11pt]{article}
\usepackage[left=1in,top=1in,right=1in,bottom=1in,head=.1in,nofoot]{geometry}

\usepackage{setspace,url,bm,amsmath} %

\usepackage[small]{titlesec} %
\titlelabel{\thetitle.\quad} %

\titleformat*{\section}{\bf\large\center}

\usepackage{graphicx} %
\usepackage{bbm}
\usepackage{latexsym}
\usepackage{caption}
\usepackage[margin=20pt]{subcaption}
\usepackage{hyperref}
\usepackage{booktabs}
\usepackage{multirow}

\usepackage{enumitem}

\usepackage[table]{xcolor}

\newcommand{\GG}[1]{}

\usepackage{amsthm}
\usepackage{amssymb}
\usepackage{amsmath}
\usepackage{color}

\usepackage{comment}
\theoremstyle{definition}

\newtheorem*{theorem*}{Theorem}
\newtheorem{theorem}{Theorem}
\newtheorem*{rmk*}{Remark}

\newtheorem{proposition}{Proposition}
\newtheorem{lemma}{Lemma}

\newtheorem{condition}{Condition}

\newtheorem{remark}{Remark}
\newtheorem{corollary}{Corollary}
\newtheorem*{corollary*}{Corollary}

\def\deri{\textup{d}}

\def\cC{\mathcal{C}}
\def\pr{\mathbb{P}}

\usepackage{natbib} %
\bibpunct{(}{)}{;}{a}{}{,} %

\usepackage{etoolbox} %
\apptocmd{\sloppy}{\hbadness 10000\relax}{}{} %

\usepackage{color}
\usepackage{listings}

\DeclareMathOperator*{\argmin}{arg\,min}

\def\ind{\begin{picture}(9,8)
         \put(0,0){\line(1,0){9}}
         \put(3,0){\line(0,1){8}}
         \put(6,0){\line(0,1){8}}
         \end{picture}
        }

\def\Pr{\mathbb{P}}

\def\Var{\text{Var}}
\def\Cov{\text{Cov}}

\def\I{\mathbbm{1}}
\def\E{\mathbb{E}}
\def\cE{\mathcal{E}}
\def\R{\mathbb{R}}
\newcommand{\one}{\mathbbm{1}}

\def\bs{\boldsymbol}

\def\limsup{\overline{\lim}}
\def\liminf{\underline{\lim}}

\def\rev{\color{black}}

\RequirePackage[normalem]{ulem}

\usepackage{soul}

\usepackage{subcaption}
\usepackage[textsize=tiny, textwidth = 2cm, shadow]{todonotes}

\usepackage{diagbox}
\usepackage{authblk}

\makeatletter
\let\@fnsymbol\@alph
\makeatother

\begin{document}

\onehalfspacing

\allowdisplaybreaks

\title{\bf 
\Large
Asymptotic Theory of the Best-Choice Rerandomization using the Mahalanobis Distance}
\author{
	Yuhao Wang \textsuperscript{a} \ 
    and Xinran Li \textsuperscript{b,*,\textdagger}	 
}

\date{}
\maketitle

\vspace{-1.7em}

\begin{abstract}
\singlespacing
    Rerandomization, a design that utilizes pretreatment covariates and improves their balance between different treatment groups, has received attention recently in both theory and practice.  
    From a survey by \citet{Bruhn:2009}, there are at least two types of rerandomization that are used in practice: the first rerandomizes the treatment assignment until covariate imbalance is below a prespecified threshold; the second randomizes the treatment assignment multiple times and chooses the one with the best covariate balance. 
    In this paper we will consider the second type of rerandomization, namely the best-choice rerandomization, 
    whose theory and inference are still lacking in the literature. 
    In particular, we will focus on the best-choice rerandomization that uses the Mahalanobis distance to measure covariate imbalance, which is one of the most commonly used imbalance measure for multivariate covariates and is invariant to affine transformations of covariates. 
    We will study the large-sample repeatedly sampling properties of the best-choice rerandomization, allowing both the number of covariates and the number of tried complete randomizations to increase with the sample size. 
    We show that the asymptotic distribution of the difference-in-means estimator is more concentrated around the true average treatment effect under rerandomization than under the complete randomization, 
    and propose large-sample accurate confidence intervals for rerandomization that are shorter than that for the completely randomized experiment. 
    We further demonstrate that, with moderate number of covariates and with the number of tried randomizations increasing polynomially with the sample size, 
    the best-choice rerandomization can achieve the ideally optimal precision that one can expect even with perfectly balanced covariates. 
    The developed theory and methods for rerandomization are also illustrated using real field experiments. 
\end{abstract}


{\let\thefootnote\relax\footnotetext{
\textsuperscript{a} Institute for Interdisciplinary Information Sciences, Tsinghua University, and Shanghai Qi Zhi Institute, China.}}

{\let\thefootnote\relax\footnotetext{
\textsuperscript{b} Department of Statistics, University of Chicago, United States of America.}}

{\let\thefootnote\relax\footnotetext{
* Corresponding author.
}}

{\let\thefootnote\relax\footnotetext{
\textsuperscript{\color{white} 1} 
E-mail address: \href{mailto:yuhaow@tsinghua.edu.cn}{yuhaow@tsinghua.edu.cn} (Y. Wang), \href{mailto:xinranli@uchicago.edu}{xinranli@uchicago.edu} (X. Li).}}

{\let\thefootnote\relax\footnotetext{
\textsuperscript{\textdagger} 
Li was partially funded by the U.S. National Science Foundation (grant 2400961). 
}}

\newpage

\onehalfspacing

\section{Introduction}

\citet{fisher1925statistical} advocated randomization in experimental design since it can eliminate bias and permit valid test of significance \citep{Hall2007}. 
Since then, randomized experiments have become the gold standard
for studying causal effects in many research areas\footnote{\rev Note that this ``gold standard'' is not without critics; see, e.g., \citet{deaton2018understanding} and  references therein for more related discussion.}, 
such as randomized clinical trials in medical research \citep{rosenberger2015randomization}, randomized field experiments in social sciences \citep{Greenbook}, and online experiments in technology companies \citep{Bojinov2022Online}. 
The completely randomized experiment (CRE){\rev, along with its stratified counterpart,} has become one of the most popular designs due to its simplicity in both implementation and analysis. 
In addition,
the CRE can balance all potential confounding factors, no matter observed or unobserved on average, 
and can justify simple and intuitive comparison between different treatment groups. 
For example, 
the difference between outcome means in two treatment groups, often called the difference-in-means estimator, is unbiased for the true average treatment effect under the CRE \citep{Neyman:1923}. 
However, as commented by \citet{fisher1926}, 
most experimenters carrying out random assignments of plots will be shocked to find out how far from equally the plots distribute themselves. 
More recently,  \citet{morgan2012rerandomization} commented that, with $10$ mutually independent covariates and at $5\%$ significance level, the usual covariate balance test will be significant for at least one covariate with probability about $40\%$. 
Note that the covariate balance test has become a common practice when reporting randomized experiments nowadays. 
When chance imbalances are observed, researchers may worry about the results from the experiment, since the difference between the treatment groups in comparison may be to due to the difference in pretreatment covariates. 
Technically speaking, this is related to the variability of the treatment effect estimation, and, as
discussed shortly, 
we can reduce the variability of the treatment effect estimator or equivalently enhance its precision by improving the balance of pretreatment covariates.

The classical solution to avoiding chance imbalance of pretreatment covariates is blocking {\rev or stratification} \citep{fisher1926, BHH2005}.
Through a survey of leading researchers carrying out randomized experiments in developing countries, 
\citet{Bruhn:2009} discovered several rerandomization methods that are used in practice to improve covariate balance but are not well discussed in print.  
Rerandomization turns out to provide a general solution to the covariate balance issue, which can easily accommodate many covariates of various types. 
Although its idea has existed for a long time in the literature tracing back to Fisher \cite[][Page 88]{savage:1962}, \citet{student:1938}, \citet{cox:1982} and etc.,
the rerandomization design is formally proposed recently by \citet{morgan2012rerandomization}, who also adopted and advocated the Fisher randomization test to analyze such a design. 
As discussed in \citet{Bruhn:2009}, 
there are at least two types of rerandomization: 
the first specifies a certain covariate balance criterion and keeps drawing treatment assignments until getting an acceptable one, 
and the second draws, say, $1000$, randomizations and chooses the one with the best covariance balance based on a certain covariate imbalance measure. 
Both of them are intuitive designs and have been commonly used in practice, but their analysis is not straightforward compared to the classical and well-studied CRE. 
Recently \citet{LDR18} studied the large-sample theory for the first type of rerandomization, revealing a general non-Gaussian asymptotic distribution for the usual difference-in-means estimator; 
see, e.g., \citet{LD20reg, LDR20factorial,YQL2021, ZD2021rerand, WWL2021, lu2022design, CF22pivot, branson2022power, wang2022rerandomization} for related extensions.

In this paper we will focus on the second type of rerandomization, which randomizes the treatment assignment multiple times and chooses the one with the best covariate balance.
To distinguish it from the first type, 
we will call it the \textit{best-choice rerandomization}. 
The best-choice rerandomization has received less attention in theory, despite its popularity in practice. 
Our goal is to address this theoretical gap by developing the large-sample theory and inference for the best-choice rerandomization.  
Specifically, 
we will consider the best-choice rerandomization design that draws $T\ge 1$ complete randomizations and chooses the one with the smallest covariate imbalance measured by the Mahalanobis distance, which is one of the most popular imbalance measure for multivariate covariates. 
A general procedure of a best-choice rerandomization is illustrated using the diagram in Figure \ref{fig:diagram}, in parallel with \citet[][Figure 1]{morgan2012rerandomization} for the first type of rerandomization. 
Specifically, 
we first randomly and independently draw treatment assignments $T$ times, 
then calculate the covariate balance for each of these assignments based on some prespecified measure, 
and finally choose the one with the best balance and use that to conduct the actual experiment. 

\begin{figure}
    \centering
    \includegraphics[width=0.8\textwidth]{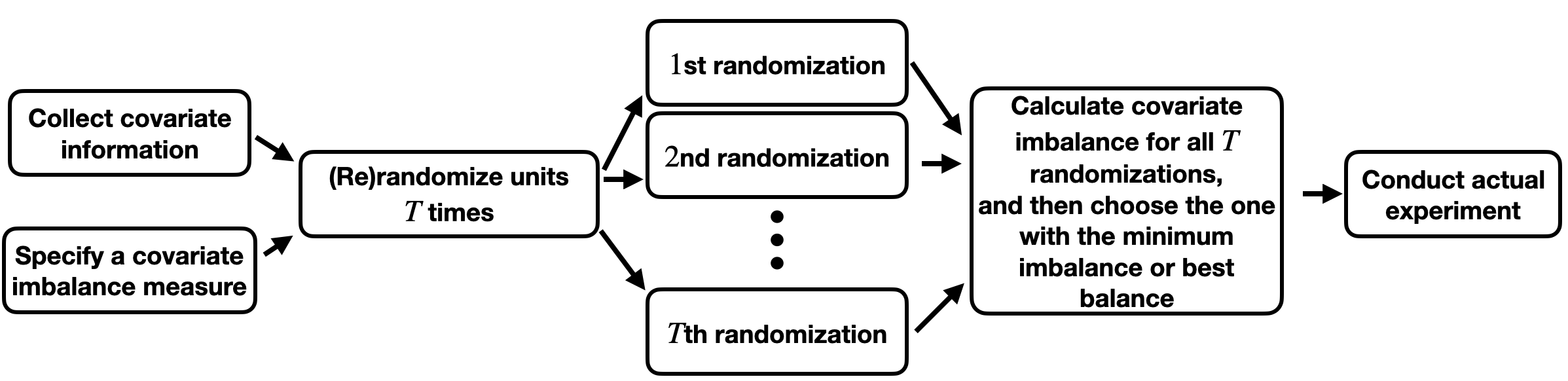}
    \caption{A general procedure of the best-choice rerandomization design.}
    \label{fig:diagram}
\end{figure}

{\rev Different from the first type of rerandomization that discards assignments with bad covariate balance, the best-choice rerandomization specifies the number of tried randomizations 
rather than a covariate balance criterion, which, if too stringent, may result in no acceptable assignments.
We conduct a survey of recent empirical studies and observe a trend toward better-documented rerandomization schemes following the seminal works of \citet{Bruhn:2009} and \citet{morgan2012rerandomization}.
Examples that explicitly used a best-choice rerandomization scheme include \citet{Brune2021}, \citet{Lowe2021}, \citet{Beaman2023}, \citet{Resnjanskij2024}, \citet{McKenzie2019}, \citet{Lee2021}, \citet{LEE2022276}, \citet{CRONIN2024102877}, \citet{BOYD2023102462}, \citet{Grimm2016}, \citet{Brade2023}, \citet{Yang2024} and \citet{Whitem3797}, ranging from social to biomedical sciences.}
{\rev Despite its wide use in practice,}
it is not clear from the existing literature that how a proper statistical inference can be conducted for the best-choice rerandomization. 
Note that, following \citet{morgan2012rerandomization}, we can still use Fisher randomization test, but it will work only for sharp null hypotheses that generally requires constant-effect-type assumptions or more broadly bounded null hypotheses that typically focus on the extreme individual effect \citep{CDLM21quantile}. 
In this paper, we will instead focus on \citet{Neyman:1923}'s 
design-based
large-sample repeated sampling inference for the average treatment effect, allowing unknown individual effect heterogeneity, 
and demonstrate the advantage of rerandomization over complete randomization. 
The {\rev design-based inference} has recently gained attention in causal inference, 
in particular because it requires no distributional assumptions on potential outcomes and guarantees the inference validity using the physical randomization as the ``reasoned basis'' \citep{Fisher:1935}; 
see, e.g., \citet{fcltxlpd2016} and \citet{Abadie2020} for recent reviews.

Another question that will receive special attention in our paper is the choice of $T$, the number of tried complete randomizations. 
Intuitively, larger $T$ can provide greater covariate balance and seems an attractive option for practitioners. 
However, when $T$ is overly large and in particular is infinite in the extreme case, all possible treatment assignments will be enumerated and the best-choice rerandomization will essentially choose the one with the best balance from all possible assignments.
When some covariates are continuous, this will generally lead to an almost deterministic design where there is no randomness in the treatment assignment. 
This apparently violates Fisher's principle of experimental design. 
A natural question to ask is then: 
how large can and should $T$ be so that 
(i) there is still sufficient randomness in the treatment assignment for robust causal inference
and 
(ii) rerandomization can achieve an ``optimal'' efficiency for treatment effect estimation? 
To the best of our knowledge, 
the choice of $T$ has been theoretically investigated only recently by \citet{Banerjee20}, 
from an ambiguity-averse decision-making perspective. 
Specifically, the authors considered an $\varepsilon$-contamination-type model \citep{Huber1964}, which essentially allows model or prior misspecification, to facilitate the discussion on the trade-off between subjective expected performance and robust performance guarantees. 
They found that the loss in robustness due to rerandomization is 
of order $O(\sqrt{\log (T)/n})$, with $T$ denoting the number of tried complete randomizations and $n$ denoting the number of experimental units, 
and suggested choosing {\rev $T$} less than the sample size $n$, ensuring the loss is on the order of $O(\sqrt{\log (n)/n})$. 
We will also study the same issue on the choice of $T$, but from a different perspective. 
In particular, we will focus on the feasibility of a large-sample 
{\rev design-based}
robust inference for treatment effects. 
In addition, we will also investigate the role of the number of covariates $K$ in rerandomization.

{\rev 
The study on the use of covariates to improve efficiency of randomized experiments dates back at least to \citet{fisher1926}, who proposed blocking or stratification as one of his principles of experimental design. 
Here we give a brief review focusing more on the recent progress. 
First, stratification can also be viewed as a special case of rerandomization with covariates being the stratum indicators \citep{morgan2012rerandomization}. 
With a fixed number of strata and large stratum sizes, ensuring equal proportions of treated units across strata can improve, or at least preserve, the precision of the difference-in-means estimator asymptotically compared to the CRE \citep[see, e.g.,][Chapter 5.3.3]{ding2023first}.
Second, \citet{GLSR2004}, \citet{Bai2022} and \citet{cytrynbaum2021optimal} have suggested finely stratified designs, such as matched pairs, which can be optimal in the sense of minimizing the mean squared error of the difference-in-means estimator; see also \citet{F18pairadj}, \citet{Bai22InfPair}, \citet{bai2023covariate}, \citet{bai2023efficiency}, \citet{cytrynbaum2024covariate}, \citet{bai2024primer} and references therein. 
Third, \citet{Harshaw24} recently proposed the Gram–Schmidt walk design utilizing tools from algorithmic discrepancy.  
Interestingly, with appropriately chosen design parameter (analogously to $T$ in our best-choice rerandomization), 
the Gram-Schmidt walk design and rerandomization achieves the same asymptotic efficiency; see also \citet[][Section 9]{Harshaw24} for more detailed comparison among rerandomization, Gram-Schmidt walk design and the paired randomization. 
Fourth, \citet{chattopadhyay2022balanced} has recently extended the finite selection model \citep{morris1979finite} into a general experimental design tool, where each treatment group, in a randomly determined order, sequentially selects units to optimize a certain assignment criterion.
Lastly, \citet{WWL2021}, \citet{Krieger2023} and \citet{cytrynbaum2024finely} have proposed combination of stratification (including pair matching) and rerandomization; 
that is we first perform stratified randomization and then rerandomizes based on some covariate balance criterion. 
In this paper, we focus on complete randomization in the first step. It will be interesting to extend it to stratified randomization with best-choice rerandomization, as such a scheme has already been implemented in some of the empirical papers we mention before; we leave this for future investigation. 
}

The paper proceeds as follows. 
Section \ref{sec:framework} introduces the framework and notation. 
Section \ref{sec:asym_bcr} studies the asymptotic properties of the best-choice rerandomization. 
Section \ref{sec:optimal} investigates whether the best-choice rerandomization can achieve its ideally optimal precision that one can expect even with perfectly balanced covariates. 
Section \ref{sec:inf} proposes large-sample valid inference for the best-choice rerandomization. 
{\rev Section \ref{sec:regadj} studies regression adjustment under the best-choice rerandomization.}
Section \ref{sec:numeric} conducts simulations to illustrate our theory, 
and Section \ref{sec:discuss} concludes with a short discussion.

\section{Framework and Notation}\label{sec:framework}

\subsection{Potential outcomes, covariates and treatment assignments}

Consider an experiment with $n$ units, where $n_1$ of them will receive some active treatment and the remaining $n_0 = n - n_1$ will receive control. 
We invoke the potential outcome framework to define treatment effects \citep{Neyman:1923, Rubin:1974}. 
For each unit $1\le i\le n$, let $Y_i(1)$ and $Y_i(0)$ denote the treatment and control potential outcomes,
and 
$\tau_i = Y_i(1) - Y_i(0)$ be the corresponding individual treatment effect. 
We are interested in inferring the average treatment effect $\tau = n^{-1} \sum_{i=1}^n \tau_i = \bar{Y}(1) - \bar{Y}(0)$, 
where $\bar{Y}(1) = n^{-1} \sum_{i=1}^n Y_i(1)$ and $\bar{Y}(0) = n^{-1} \sum_{i=1}^n Y_i(0)$ denote the average treatment and control potential outcomes, respectively. 
The fundamental difficulty of causal inference is that we can observe at most one potential outcome for each unit and thus half of the potential outcomes will be missing. 
Specifically, 
for each unit $i$, let $Z_i\in \{0,1\}$ be the treatment assignment indicator, where $Z_i=1$ if the unit receives treatment and $0$ otherwise.
The observed outcome for each unit $i$ is then $Y_i = Z_i Y_i(1)  + (1-Z_i)Y_i(0)$, one of the two potential outcomes.

Throughout the paper, we will conduct the 
{\rev design-based inference}\footnote{\rev This is also often called the finite population inference or randomization-based inference, which uses the randomization of treatment assignment as the ``reasoned basis'' \citep{Fisher:1935}. 
We mainly use the terminology of ``design-based inference'' to better distinguish it from the permutation inference driven by random sampling of units from some population; see \citet{Ernst2004} and \citet{Goeman2021} for more related discussion. 
}
\citep{Neyman:1923, fcltxlpd2016}, 
where all the potential outcomes (as well as the pretreatment covariates introduced shortly) for the $n$ experimental units are viewed as fixed constants or equivalently being conditioned on. 
The 
{\rev design-based}
inference has the advantage of avoiding any model or distributional assumptions on the potential outcomes and covariates (as well as their dependence structure)\footnote{{\rev 
It is worth pointing out that our design-based inference focuses on treatment effects for units in an experiment. 
Generalize the results of an experiment to other or larger populations requires additional assumptions, such as
the representativeness of the experimental units for the population of interest; see, e.g., \citet{Rubin:1974}, \citet{YQL2021}, and references therein.
}}. 
The randomness in the observed data comes solely from the random treatment assignment.  Therefore, the distribution of the treatment assignment vector $\bs{Z} = (Z_1, Z_2, \ldots, Z_n)^\top$, also called the treatment assignment mechanism \citep{Rubin:1978}, governs the data generating process and is crucial for statistical inference. 
In a randomized experiment, the experimenter can generate the treatment assignment vector from a carefully prespecified or designed distribution, based on which units will be allocated into treatment and control groups.

The completely randomized experiment (CRE) is one of the most commonly used treatment assignment mechanism, 
under which the treatment assignment vector $\bs{Z}$ takes a particular value $\bs{z} = (z_1, z_2, \ldots, z_n)^\top \in \{0,1\}^n$ with probability $\binom{n}{n_1}^{-1}$ if $\sum_{i=1}^n z_i = n_1$ and zero otherwise.

\subsection{Covariate imbalance and rerandomization}\label{sec:imbalance}

Let $\bs{x}_i\in \mathbb{R}^K$ denote the available pretreatment covariate vector for each unit $i$, 
$\bar{\bs{x}} = n^{-1} \sum_{i=1}^n \bs{x}_i$ denote the average covariate vector for all units, 
and $\bs{S}^2_{\bs{x}} = (n-1)^{-1} \sum_{i=1}^n (\bs{x}_i - \bar{\bs{x}}) (\bs{x}_i - \bar{\bs{x}})^\top$ denote the finite population covariance matrix of covariates. 
We further introduce 
\begin{align}\label{eq:diff_cov}
    \hat{\bs{\tau}}_{\bs{x}} = \bar{\bs{x}}_1 - \bar{\bs{x}}_0
    = \frac{1}{n_1} \sum_{i=1}^n Z_i \bs{x}_i - \frac{1}{n_0} \sum_{i=1}^n (1-Z_i) \bs{x}_i
\end{align}
to denote the difference-in-means of covariates, 
where $\bar{\bs{x}}_1$ and $\bar{\bs{x}}_0$ denote the average covariates in treated and control groups.
Denote the covariance matrix of $\hat{\bs{\tau}}_{\bs{x}}$ under the CRE by 
$\bs{V}_{\bs{xx}} = \Cov(\hat{\bs{\tau}}_{\bs{x}}) = n/(n_1n_0)\cdot \bs{S}^2_{\bs{x}}$.

In practice, it is often a routine to check the imbalance of the pretreatment covariates when conducting randomized experiments. 
In this paper we will focus on
the Mahalanobis distance imbalance measure, which is one of the most commonly used imbalance measure for multivariate covariates, enjoys the affine invariant property, and has the following form: 
\begin{align}\label{eq:M_dist}
    M
    & = \hat{\bs{\tau}}_{\bs{x}}^\top \bs{V}_{\bs{xx}}^{-1} \hat{\bs{\tau}}_{\bs{x}}
    = 
    \frac{n_1n_0}{n}(\bar{\bs{x}}_1 - \bar{\bs{x}}_0)^\top (\bs{S}_{\bs{x}}^2 )^{-1} (\bar{\bs{x}}_1 - \bar{\bs{x}}_0). 
\end{align}
When the covariates, especially those likely to have strong associations with the potential outcomes, are imbalanced, 
we may worry about the results from the experiment. 
In particular, we may worry that the difference in outcomes between treated and control groups is due to the difference in baseline covariates, instead of the treatment effects. 
Moreover, as discussed earlier,  
covariate imbalance is not rare even under the intuitive and commonly used CRE \citep{morgan2012rerandomization}.
Therefore, a design that can mitigate or avoid unlucky and bad chance covariate imbalance will be highly desirable.

Rerandomization is a general design that can improve the balance of pretreatment covariates, 
by checking covariate balance prior to conducting the actual experiment. 
This is feasible, since the covariate balance depends only on the treatment assignment and pretreatment covariates, without involving any post-treatment variables. 
Throughout the paper, we will focus on the best-choice rerandomization using the Mahalanobis distance. 
Specifically, we first completely randomize the units or equivalently draw treatment assignments from the CRE $T$ times, where $T\ge 1$ is a prespecified integer, 
then calculate the Mahalanobis distance in \eqref{eq:M_dist} for each of these $T$ complete randomizations, 
and finally choose the treatment assignment with the minimum Mahalanobis distance to conduct the actual experiment; see also Figure \ref{fig:diagram} for a general best-choice rerandomization.
In this paper we aim to develop the large-sample theory and inference for the best-choice rerandomization under the {\rev design-based inference} framework.

\subsection{Difference-in-means of the outcome and covariates under the CRE}\label{sec:diff_out_cov_cre}

Throughout the paper we will focus on inference of the average treatment effect $\tau$ under the best-choice rerandomization. 
Moreover, we will focus on the intuitive difference-in-means estimator:
\begin{align}\label{eq:diff_est}
    \hat{\tau} & = \frac{1}{n_1} \sum_{i=1}^n Z_i Y_i - \frac{1}{n_0} \sum_{i=1}^n (1-Z_i) Y_i,  
\end{align}
which is the difference between the average observed outcomes in treated and control groups. 
As discussed shortly, the joint distribution of the difference-in-means of the outcome and covariates in \eqref{eq:diff_est} and \eqref{eq:diff_cov} under the CRE plays an important role in studying the property of the best-choice rerandomization. 
Below we discuss its first two moments, i.e., mean and covariance matrix. 

Recall that $\bs{S}^2_{\bs{xx}}$ denotes the finite population covariance matrix of the covariates. 
For $z=0,1$, 
let $S^2_z = (n-1)^{-1} \sum_{i=1}^n \{Y_i(z) - \bar{Y}(z)\}^2$ be the finite population variance of potential outcomes, 
and 
$\bs{S}_{z\bs{x}} = \bs{S}_{\bs{x}z}^\top = (n-1)^{-1} \sum_{i=1}^n \{Y_i(z) - \bar{Y}(z)\} (\bs{x}_i - \bar{\bs{x}})^\top$ be the finite population covariance between potential outcomes and covariates. 
Define analogously $S^2_{\tau} = (n-1)^{-1} \sum_{i=1}^n (\tau_i - \tau)^2$ as the finite population variance of individual effects and $S_{\tau\bs{x}} = S_{\bs{x}\tau}^\top = (n-1)^{-1} \sum_{i=1}^n (\tau_i - \tau)(\bs{x}_i - \bar{\bs{x}})^\top$ as the finite population covariance between individual effects and covariates. 
From \citet{LDR18}, under the CRE, the difference-in-means of the outcome and covariates
$(\hat{\tau}, \hat{\bs{\tau}}_{\bs{X}}^\top)^\top$
has mean $(\tau, \bs{0}^\top)^\top$, indicating that the difference-in-means estimator is unbiased for the true average treatment effect and the covariates are balanced on average between the two treatment groups, 
and covariance matrix
\begin{align}\label{eq:V}
    \bs{V} 
    \equiv
    \begin{pmatrix}
    V_{\tau\tau} & \bs{V}_{\tau \bs{x}}\\
    \bs{V}_{\bs{x}\tau} & \bs{V}_{\bs{xx}}
    \end{pmatrix}
    = 
    \begin{pmatrix}
    n_1^{-1} S_1^2 + n_0^{-1} S_0^2 - n^{-1} S^2_{\tau} & n_1^{-1} \bs{S}_{1\bs{x}} + n_0^{-1} \bs{S}_{0\bs{x}}\\
    n_1^{-1} \bs{S}_{\bs{x}1} + n_0^{-1} \bs{S}_{\bs{x}0} & n/(n_1 n_0) \cdot \bs{S}^2_{\bs{x}}
    \end{pmatrix}. 
\end{align}

Below we further introduce an important measure for the association between potential outcomes and covariates, which will play an important role in studying the asymptotic properties of the best-choice rerandomization. 
Specifically, we consider the squared multiple correlation between the difference-in-means of the outcome and covariates under the CRE as an $R^2$-type measure for the association between potential outcomes and covariates: 
\begin{align}\label{eq:R2}
    R^2 = \textup{Corr}^2(\hat{\tau}, \hat{\bs{\tau}}_{\bs{x}})
    = 
    \frac{
    \bs{V}_{\tau \bs{x}} \bs{V}_{\bs{xx}}^{-1} \bs{V}_{\bs{x}\tau} 
    }{
    V_{\tau\tau}
    }
    = 
    \frac{
    n_1^{-1} S_{1\mid \bs{x}}^2 + n_0^{-1} S_{0\mid \bs{x}}^2 - n^{-1} S^2_{\tau\mid \bs{x}}
    }{
    n_1^{-1} S_1^2 + n_0^{-1} S_0^2 - n^{-1} S^2_{\tau}
    }, 
\end{align}
where the equivalent forms follow from \citet{LDR18}. 
In \eqref{eq:R2}, 
$S_{z\mid \bs{x}} = \bs{S}_{z\bs{x}} (\bs{S}_{\bs{x}}^2)^{-1} \bs{S}_{\bs{x}z}$ denotes the finite population variance of the linear projections of potential outcomes on covariates, for $z=0,1$, 
and $S_{\tau\mid \bs{x}} = \bs{S}_{\tau\bs{x}} (\bs{S}_{\bs{x}}^2)^{-1} \bs{S}_{\bs{x}\tau}$ analogously denotes the finite population variance of the linear projections of individual effects on covariates. 
When treatment effects are additive, in the sense that $\tau_i$ is constant across all $i$, $R^2$ reduces to $S_{0\mid \bs{x}}^2/S^2_0$, the squared multiple correlation between control potential outcomes and covariates (i.e., 
the proportion of variability in the control potential outcomes that can be linearly explained by the covariates). 

\subsection{Finite population asymptotics and Berry--Esseen-type bounds}\label{sec:fp_asym}

Because the exact distribution of the difference-in-means estimator is generally intractable under the best-choice rerandomization,   
we will invoke large-sample approximations. 
Specifically, we will conduct the finite population asymptotics that embeds the finite population of size $n$ into a sequence of finite populations with increasing sizes; see \citet{fcltxlpd2016} for a review with an emphasize on applications to causal inference. 
Importantly, as pointed out by \citet{Neyman:1923} in his seminal paper, 
under the CRE and when the sample size is large, 
the distribution of 
the difference-in-means of the outcome in \eqref{eq:diff_est} (and analogously of covariates in \eqref{eq:diff_cov}) can be well approximated by a Gaussian distribution; see, for example, \citet{hajek1960limiting} for a rigorous proof and \citet{fcltxlpd2016} for extension to vector outcomes with multivariate Gaussian approximation. 

Furthermore, in our large-sample analysis for the best-choice rerandomization, we will allow both the number of tried complete randomizations $T$ and the number of covariates $K$ to vary (say, increase) with the sample size. 
Specifically, we will view $T$ and $K$ as $T_n$ and $K_n$ in the remainder of the paper; for descriptive convenience, we will keep such dependence on the sample size $n$ implicit. 
In order to deal with the sample size dependent $T$ and $K$, we need a more delicate characterization of the multivariate Gaussian approximation under the CRE. 
In particular, we will consider the following Berry--Esseen-type bound for the Gaussian approximation of the joint distribution of the difference-in-means of the outcome and covariates under the CRE. {\rev Define} 
\begin{align}\label{eq:Delta}
    \Delta_n \equiv \sup_{Q \in \mathcal{C}_{K + 1}} \left|\Pr\left(\bs{V}^{-1/2} \left(
    \begin{matrix}
    \hat{\tau} - \tau \\
    \hat{\bs{\tau}}_{\bs{X}}
    \end{matrix}
    \right) \in Q\right)
     - \Pr(\bs{\varepsilon} \in Q)\right|, 
\end{align}
where $\mathcal{C}_{K + 1}$ denotes the collection of all measurable convex sets in $\mathbb{R}^{K+1}$, 
$\bs{\varepsilon}\sim \mathcal{N}(\bs{0}, \bs{I}_{K+1})$ is a $K+1$ dimensional standard Gaussian random vector, 
and $\bs{V}$ is defined as in \eqref{eq:V}. 
Based on \citet{R15}'s conjecture, 
there exists an absolute constant $C$ such that 
$\Delta_n\le C\gamma_n$ 
with  
\begin{align}\label{eq:gamma}
    \gamma_n & = 
    \frac{(K+1)^{1/4}}{\sqrt{nr_1r_0}} 
    \frac{1}{n} \sum_{i=1}^n \| \bs{S}_{\bs{u}}^{-1} (\bs{u}_i - \bar{\bs{u}}) \|_2^3, 
\end{align}
where $\bs{u}_i \equiv (r_0 Y_i(1) + r_1 Y_i(0), \bs{X}_i^\top)^\top$, 
$\bar{\bs{u}}$ and $\bs{S}_{\bs{u}}^2$ denote the finite population mean and covariance of the $\bs{u}_i$'s, 
and $\bs{S}_{\bs{u}}^{-1}$ denotes the inverse of the positive semidefinite square root of $\bs{S}^2_{\bs{u}}$. 
\citet{wang2022rerandomization} recently proved that 
$\Delta_n \le 174 \gamma_n + 7 \gamma_n^{1/3}$; 
see also \citet[][Theorem 2]{wang2022rerandomization},  \citet{shi2022berry} and \citet{ShiLi2024} for other forms of Berry-Esseen-type bounds on $\Delta_n$. 
We will then assume the following regularity condition along the sequence of finite populations, which can guarantee the Gaussian approximation for the difference-in-means of the outcome and covariates (or equivalently that $\Delta_n$ converges to zero as $n\rightarrow \infty$). 

\begin{condition}\label{cond:gamma}
    As the size of the finite population $n\rightarrow\infty$, $\gamma_n$ in \eqref{eq:gamma} converges to zero.  
\end{condition}

{\rev Condition \ref{cond:gamma} is the same as~\citet[Condition~1]{wang2022rerandomization}.}
It implicitly requires that the potential outcomes and covariates are not too heavy-tailed, 
and that the number of covariates does not increase too fast with the sample size. 
{\rev 
When the number of covariates $K$ is bounded, the proportions of treated and control units are bounded away from zero, the potential outcomes and covariates are bounded, and the minimum eigenvalue of $\bs{S}_{\bs u}^2$ is bounded away from zero (which intuitively requires that the potential outcomes and covariates are not too colinear), then $\gamma_n$ is on the order of $n^{-1/2}$, under which Condition \ref{cond:gamma} must hold. 
When $K$ can diverge with $n$,}
Condition \ref{cond:gamma} 
implies 
that $K = o(n^{2/7})$ \citep{wang2022rerandomization}. 
{\rev In addition, as discussed in \citet{wang2022rerandomization} and also later in Section \ref{sec:opt_bcr_subsec}, when experimental units are random samples from a superpopulation, under some regularity conditions on the superpopulation, 
$\gamma_n$ will converge to zero in probability as long as $K$ is sufficiently smaller than $n$. 
We also refer readers to \citet{wang2022rerandomization} for more detailed discussion about this regularity condition.}

We then impose the following condition that the number of tried complete randomizations does not increase too fast with the sample size. This will be discussed and emphasized in detail later. 

\begin{condition}\label{cond:iterations}
    As $n \to \infty$, $T \Delta_n  \to 0$, or equivalently $T = o(\Delta_n^{-1})$. 
\end{condition}

From the discussion before, 
a sufficient condition for Condition \ref{cond:iterations} is that $T \gamma_n^{1/3} \rightarrow 0$ as $n\rightarrow 0$, or a weaker form of $T \gamma_n \rightarrow 0$ if the conjecture in \citet{R15} holds. 
{\rev Note that, if Condition \ref{cond:gamma} holds, then Condition \ref{cond:iterations} must hold for any fixed $T$ that do not vary with the sample size, say, $T=1000$. 
See also the discussion in Section \ref{sec:opt_bcr_subsec} regarding the acceptable rates of $T$ under various rates for $K$ and consequently $\gamma_n$.} 

\section{Asymptotic theory for the best-choice rerandomization}\label{sec:asym_bcr}

\subsection{The best-choice rerandomization using the Mahalanobis distance}
To formally introduce the best-choice rerandomization design, we first introduce several notations. 
Let $\bs{Z}_{[1]}, \bs{Z}_{[2]}, \ldots,$ and $\bs{Z}_{[T]}$ denote $T$ mutually independent treatment assignment vectors from the CRE with $n_1$ and $n_0$ units receiving treatment and control, respectively. For each $1\le t\le T$, 
let 
$\hat{\bs{\tau}}_{[t]\bs{x}}$ be the difference-in-means of covariates as in \eqref{eq:diff_cov} under the treatment assignment $\bs{Z}_{[t]}$, 
and $M_{[t]} \equiv \hat{\bs{\tau}}_{[t]\bs{x}}^\top \bs{V}_{\bs{xx}}^{-1} \hat{\bs{\tau}}_{[t]\bs{x}}$ be the corresponding Mahalanobis distance for covariate imbalance as in \eqref{eq:M_dist}. 

With a slight abuse of notation, 
we use $M_{(1)} = \min_{1\le t\le T} M_{[t]}$ to denote the minimum 
Mahalnobis distance, with the subscript $(1)$ representing the index in $\{1,2,\ldots, T\}$ that achieves this minimum. If there are multiple treatment assignments achieving the minimum at the same time, we will then randomly choose one from them. 
Consequently, $\bs{Z}_{(1)}$ will be the treatment assignment with the minimum covariate imbalance (measured by the Mahalanobis distance) among all the $T$ complete randomizations. 
Under the best-choice rerandomization, 
as illustrated in Figure \ref{fig:diagram}, 
we will use the ``best'' assignment $\bs{Z}_{(1)}$ to conduct the actual experiment (or more precisely to conduct the actual treatment allocation). 
We emphasize that the best-choice rerandomization depends on the number $T$ of tried complete randomizations; 
for descriptive convenience, we will make such dependence implicit, unless otherwise stated. 

\subsection{Difference-in-means estimator under the best-choice rerandomization}\label{sec:dim_bcr}

We consider the intuitive difference-in-means estimator in \eqref{eq:diff_est} to estimate the average treatment effect $\tau$ under the best-choice rerandomization. 
Specifically, 
recalling that $\bs{Z}_{(1)} = (Z_{(1)1}, \ldots, Z_{(1)n})^\top$ is the treatment assignment actually implemented under the best-choice rerandomization, 
we will denote the corresponding difference-in-means estimator by $\hat{\tau}_{(1)} = n_1^{-1} \sum_{i=1}^n Z_{(1)i} Y_i - n_0^{-1} \sum_{i=1}^n (1-Z_{(1)i}) Y_i$, 
where we use the subscript $(1)$ to emphasize that it is the estimator under the treatment assignment $\bs{Z}_{(1)}$ with the minimum covariate imbalance. 
Below we will study the asymptotic distribution of $\hat{\tau}_{(1)}$ under the best-choice rerandomization.

By the construction of the best-choice rerandomization design, the distribution of $\hat{\tau}_{(1)}$ relies on the joint distribution of the differences in means of the outcome and covariates for the $T$ mutually independent complete randomizations. 
From Section \ref{sec:fp_asym}, 
under certain regularity conditions, 
these differences in means are approximately Gaussian distributed. 
Thus, intuitively, we can approximate the distribution of $\hat{\tau}_{(1)}$ by the corresponding part implied by the multivariate Gaussian approximations. 
As demonstrated below, such an intuition can be made rigorous under Conditions \ref{cond:gamma} and \ref{cond:iterations}. 

Let 
$(\tilde{\tau}_{[t]}, \tilde{\bs{\tau}}_{\bs{x}[t]}^\top)^\top$, $1\le t \le T$, 
be $T$ mutually independent Gaussian random vectors with mean zero and covariance matrix $\bs{V}$ in \eqref{eq:V}, which can be viewed as Gaussian approximations for the differences in means of the outcome and covariates from the $T$ mutually independent complete randomizations. 
Define $\tilde{M}_{[t]} \equiv  \tilde{\bs{\tau}}_{\bs{x}[t]}^\top \bs{V}_{\bs{xx}}^{-1} \tilde{\bs{\tau}}_{\bs{x}[t]}$ analogously as in  \eqref{eq:M_dist} for $1\le t \le T$, and let $\tilde{M}_{(1)} = \min_{1\le t \le T}\tilde{M}_{[t]}$ be the minimum among the  $\tilde{M}_{[t]}$'s. 
With a slight abuse of notation, we use 
the subscript $(1)$ to denote the index in $\{1,2,\ldots, T\}$ achieving this minimum; 
when there are multiple indices (i.e., ties) achieving the minimum at the same time, we randomly choose one from them. 
Consequently, $\tilde{\tau}_{(1)}$ is one of the $\tilde{\tau}_{[t]}$'s that corresponds to the minimum value of the $\tilde{M}_{[t]}$'s. 
By construction of the best-choice rerandomization, $\tilde{\tau}_{(1)}$ 
corresponds to
$\hat{\tau}_{(1)}$ under the Gaussian approximation. 
The theorem below characterizes the difference between the distributions of $\hat{\tau}_{(1)}$ and $\tilde{\tau}_{(1)}$. 

\begin{theorem}\label{thm:asymp}
    Under the best-choice rerandomization using the Mahalanobis distance, 
    \begin{align}\label{eq:bound_diff_tilde}
    \sup_{c\in \mathbb{R}} \Big| \Pr\big\{ V_{\tau\tau}^{-1/2} (\hat{\tau}_{(1)} - \tau) \leq c \big\}   - 
    \Pr
    \big(V_{\tau\tau}^{-1/2} \tilde{\tau}_{(1)} \leq c\big) \Big|
    \le 2 T \Delta_n. 
    \end{align}
    If Conditions \ref{cond:gamma} and \ref{cond:iterations} hold, then the supremum in \eqref{eq:bound_diff_tilde} converges to zero. 
\end{theorem}

Theorem \ref{thm:asymp} justifies the asymptotic approximation for the difference-in-means estimator under the best-choice rerandomization. 
Below we simplify the distribution of $\tilde{\tau}_{(1)}$. 
Let 
$\bs{D}_t=(D_{t1}, \ldots, D_{tK})^\top$, for $1\le t \le T$, be independent and identically distributed (i.i.d.) $K$-dimensional standard Gaussian random vectors, i.e., $\bs{D}_1, \ldots, \bs{D}_T  \stackrel{\textup{i.i.d.}}{\sim}  \mathcal{N}(\bs{0}, \bs{I}_K)$. 
We further define the following constrained Gaussian random variable: 
\begin{align}\label{eq:L_KT}
    L_{K, T} \sim D_{11} \mid \|\bs{D}_1\|_2^2 \le \min_{1 \le t \le T} \|\bs{D}_t\|_2^2.
\end{align}
Recall the squared multiple correlation $R^2$ in \eqref{eq:R2}. 
\begin{theorem}\label{thm:asym_equiv}
    The asymptotic distribution in \eqref{eq:bound_diff_tilde}
    for the standardized difference-in-means estimator $V_{\tau\tau}^{-1/2} (\hat{\tau}_{(1)} - \tau)$ under the best-choice rerandomization
    has the following equivalent form:
    \begin{align}\label{eq:asym_equiv}
        V_{\tau\tau}^{-1/2} \tilde{\tau}_{(1)} \sim \sqrt{1-R^2} \ \varepsilon_0 + \sqrt{R^2} \ L_{K, T},
    \end{align}
    where $\varepsilon_0\sim \mathcal{N}(0,1)$, 
    $L_{K,T}$ follows the distribution in \eqref{eq:L_KT}, 
    and they are mutually independent. 
\end{theorem}

\begin{remark}\label{rmk:first_type}
    For the first type of rerandomization using the Mahalanobis distance, 
    \citet{wang2022rerandomization} showed that, under Condition \ref{cond:gamma}, 
    the supremum distance between the distribution functions of the standardized difference-in-means estimator under rerandomization and the corresponding constrained-Gaussian approximation as in \eqref{eq:bound_diff_tilde} is of order $O(\Delta_n/p)$, with $p$ being the approximate acceptance probability under the given imbalance threshold\footnote{In~\citet{wang2022rerandomization}, the approximate acceptance probability $p$ is defined as $p \equiv \pr(\chi_K^2 \le a)$, where $\chi_K^2$ is the chi-squared random variable with degrees of freedom $K$ and $a$ is the given imbalance threshold. This is because under Condition \ref{cond:gamma}, the distribution of the Mahalanobis distance is approximately $\chi_K^2$, so that $\pr(M \le a) \approx \pr(\chi_K^2 \le a)$.}.
    Thus, the first and second types of rerandomization share similar approximation error (at least in terms of the derived upper bounds) when $1/p$ and $T$ are of the same order. 
    This is not surprising from their implementation. 
    Under the first type of rerandomization, in expectation, we will draw about $1/p$ assignments to get an acceptable one; 
    whereas under the second type, we will deterministically draw $T$ assignments to get an acceptable one, which is the one with the best balance. 
    Nevertheless, the technical derivation for these error bounds is considerably different for these two types of rerandomization. 
\end{remark}

\subsection{Representation for the asymptotic distribution under rerandomization}

From Theorems \ref{thm:asymp} and \ref{thm:asym_equiv}, the asymptotic distribution of the difference-in-means estimator under the best-choice rerandomization can be approximated by the distribution in \eqref{eq:asym_equiv}, which involves the constrained Gaussian random variable $L_{K,T}$ in \eqref{eq:L_KT}. 
Below we will give a representation of $L_{K,T}$, which can facilitate its simulation. 

Let $U_K$ be the first coordinate of a $K$-dimensional random vector uniformly distributed on the $(K-1)$-dimensional unit sphere, 
$S$ be a random sign with probability $1/2$ being $1$ and $-1$, $\beta_K\sim \textup{Beta}(1/2, (K-1)/2)$ be a Beta random variable that degenerates to $1$ when $K=1$.  
Let $\chi^2_{K[1]}, \chi^2_{K[2]}, \ldots,$ and $\chi^2_{K[T]}$ be i.i.d.\ chi-squared random variables with degrees of freedom $K$, 
and $\chi^2_{K(1)} = \min_{1\le t \le T} \chi^2_{K[t]}$ be the minimum of these $T$ i.i.d.\ chi-squared random variables. 
Define further the following constrained chi-squared random variable: 
\begin{align}\label{eq:chi2_KT}
    \chi^2_{K,T} 
    \ \sim \ 
    \chi^2_{K[1]} \mid \chi^2_{K[1]} \le \min_{1\le t\le T} \chi^2_{K[t]} 
    \ \sim \ \chi^2_{K(1)} 
    \ \sim \ 
    F_{K}^{-1} (\textup{Beta}(1,T)),
\end{align}
where $F_K^{-1}$ denotes the quantile function for the chi-squared distribution with degrees of freedom $K$, 
and $\textup{Beta}(1,T)$ denotes a Beta random variable with parameters $1$ and $T$; 
see the supplementary material for a proof of the equivalence in \eqref{eq:chi2_KT}. 

\begin{proposition}\label{prop:L_KT}
The constrained Gaussian random variable in \eqref{eq:L_KT} has the following representations:
\begin{align}\label{eq:L_KT_rep}
    L_{K, T} \sim \bs{c}^\top \bs{D}_1 \mid \|\bs{D}_1\|_2^2 \le \min_{1 \le t \le T} \|\bs{D}_t\|_2^2
    \sim \chi_{K,T} U_K
    \sim 
    \chi_{K,T} S \sqrt{\beta_K},
\end{align}
where $\bs{c}$ can be any constant unit vector in $\mathbb{R}^K$, $\bs{D}_1, \ldots, \bs{D}_T \overset{\textup{i.i.d.}}{\sim} \mathcal{N}(\bs{0}, \bs{I}_K)$, 
$\chi_{K,T}$ is the square root of $\chi^2_{K,T}$ in \eqref{eq:chi2_KT}, 
$\chi_{K,T}\ind U_K$, 
and $(\chi_{K,T}, S, \beta_K)$ are mutually independent. 
\end{proposition}

The representation in Proposition \ref{prop:L_KT} is analogous to that 
in \citet{LDR18} for the first type of rerandomization using the Mahalanobis distance. Both of them have similar forms, except that our representation in \eqref{eq:L_KT_rep} involves the order statistic of chi-squared random variables while that in \citet{LDR18} involves truncated chi-squared random variable. 
This is not surprising given the implementation of the design: 
the best-choice rerandomization chooses the best one among multiple randomizations, while the first-type rerandomization chooses only those assignments with covariate imbalance below a certain threshold. 

More importantly, 
from Proposition \ref{prop:L_KT} and \eqref{eq:chi2_KT}, we can easily simulate the constrained Gaussian random variable $L_{K,T}$ using the multiplication of the three random variables in \eqref{eq:L_KT_rep}, which can be more efficient than using the form in \eqref{eq:L_KT}. 
Consequently, we can also efficiently simulate from the asymptotic distribution of $\hat{\tau}_{(1)}$ in Theorem \ref{thm:asym_equiv}. 
This can be useful when conducting inference for the average treatment effect under the best-choice rerandomization, as discussed in Section \ref{sec:inf}.

\subsection{Improvement from the best-choice rerandomization}\label{sec:improve}

In this subsection we will compare the asymptotic properties of the classical CRE and the best-choice rerandomization. 
Note that the CRE can be viewed as a special case of the best-choice rerandomization with $T=1$, 
under which $L_{K,T}$ reduces to a standard Gaussian random variable and the asymptotic distribution of the standardized difference-in-means estimator $V_{\tau\tau}^{-1/2}(\hat{\tau}-\tau)$ reduces to a standard Gaussian distribution. 
Below we essentially compare the asymptotic distribution in Theorem \ref{thm:asym_equiv} to the standard Gaussian distribution $\mathcal{N}(0,1)$.

First, both the standard and constrained Gaussian random variables $\varepsilon_0$ and $L_{K,T}$ are symmetric and unimodal around zero. 
These properties will also be maintained under scaling and convolution. 
We can then immediately derive the following corollary, 
which implies that the difference-in-means estimator is asymptotically unbiased under both the CRE and the best-choice rerandomization. 

\begin{corollary}\label{cor:sum}
The asymptotic distribution for the standardized difference-in-means estimator in \eqref{eq:asym_equiv} is symmetric and unimodal around zero. 
\end{corollary}

Second, we compare the asymptotic variance of the difference-in-means estimator under the two designs. 
Let $v_{K,T} = \Var(L_{K,T})$ denote the variance of the constrained Gaussian random variable in \eqref{eq:L_KT} and \eqref{eq:L_KT_rep}. 
Note that $v_{K,T} = K^{-1} \E(\chi^2_{K(1)})$ as implied by \eqref{eq:chi2_KT} and \eqref{eq:L_KT_rep}. 
We may use expressions from \citet{Nadarajah08} for moments of chi-squared order statistics. 
However, these expressions involves Lauricella functions that are not available in standard software. 
For simplicity, 
we will mainly consider Monte Carlo approximation for $v_{K,T}$.

\begin{corollary}\label{cor:reduce_var}
Under the best-choice rerandomization, 
the asymptotic variance of the standardized difference-in-means estimator is smaller than or equal to that under the CRE. Specifically, the percentage reduction in asymptotic variance is $(1-v_{K,T})R^2$, which is nonnegative and nondecreasing in both $R^2$ and $T$. 
\end{corollary}

Intuitively, the covariates can be viewed as potential outcomes that are unaffected by the treatment. 
Thus, by the same logic, the covariates 
will be more balanced (or more precisely have smaller asymptotic variances) under the best-choice rerandomization than under the CRE. 
Moreover, the percentage reduction in asymptotic variance of any linear combination of covariates is $1-v_{K,T}$, enjoying the ``equal percent variance reducing'' property \citep{morgan2012rerandomization}. 

Third, we compare the asymptotic quantile ranges of the difference-in-means estimator under the two designs, 
because the asymptotic distribution under the best-choice rerandomization is generally non-Gaussian and its variability cannot be fully characterized by the variance. 
Moreover, we will focus on the symmetric quantile range, which will be shortest at any given coverage level due to the unimodality in Corollary \ref{cor:sum} \citep{CB2002}.
This is also related to the two-sided confidence intervals discussed later in Section \ref{sec:inf}.
For any $\alpha\in (0,1)$, 
let $z_{\alpha}$ be the $\alpha$th quantile of the standard Gaussian distribution, and 
$\nu_{\alpha, K, T}(R^2)$ be the $\alpha$th quantile of the distribution in \eqref{eq:asym_equiv}. 

\begin{corollary}\label{cor:reduce_qr}
    Under the best-choice rerandomization, 
    for any $\alpha\in (0,1)$, 
    the asymptotic $1-\alpha$ symmetric quantile range is narrower than or equal to that under the CRE. 
    Specifically, the percentage reduction in length of the  asymptotic $1-\alpha$ symmetric quantile range is $1-\nu_{1-\alpha/2, K, T}(R^2)/z_{1-\alpha/2}$, which is nonnegative and nondecreasing in both $R^2$ and $T$. 
\end{corollary}

From Corollaries \ref{cor:reduce_var} and \ref{cor:reduce_qr}, the best-choice rerandomization improves the estimation precision compared to the usual CRE. 
Moreover, the gain from rerandomization increases with the squared multiple correlation $R^2$ in \eqref{eq:R2}, 
which characterizes the strength of the association between the potential outcomes and covariates. 
This is intuitive. 
When the covariates have stronger association with the potential outcomes (or equivalently can explain more variability in the potential outcomes), 
the best-choice rerandomization can provide more precision improvement by balancing these covariates. 

Corollaries \ref{cor:reduce_var} and \ref{cor:reduce_qr} also imply that the gain from rerandomization increases with the number $T$ of tried complete randomizations. 
However, this does not mean that we should use as many complete randomizations as possible for the best-choice rerandomization. 
The is because Condition \ref{cond:iterations} requires that $T$ cannot be too large. 
If $T$ is too large, the asymptotic approximation in Theorem \ref{thm:asymp} may fail, which will further invalidate the results in Corollaries \ref{cor:reduce_var} and \ref{cor:reduce_qr}. 
We will focus on this issue regarding the choice of $T$ in the following section.

\section{Optimal best-choice rerandomization}\label{sec:optimal}

Condition \ref{cond:iterations} and Corollaries \ref{cor:reduce_var} and \ref{cor:reduce_qr} show the trade-off when choosing the number $T$ of tried complete randomizations for the best-choice rerandomization. 
On the one hand, we want $T$ to be small so that the regularity condition is more likely to hold and the asymptotic approximation can be more accurate. 
On the other hand, we want $T$ to be large so that we can gain more improvement in precision from the best-choice rerandomization. 
In particular, 
the asymptotic distribution in \eqref{eq:asym_equiv} becomes most concentrated around zero when $T=\infty$, under which it reduces to the Gaussian distribution $\mathcal{N}(0, 1-R^2)$. 
This is the ideally optimal precision that we can expect from rerandomization, since $1-R^2$ comes from the variability in potential outcomes that cannot be explained by the covariates. 
These then naturally lead to the following question: 
Can we increase $T$ at a proper rate of the sample size $n$ so that the asymptotic theory for the best-choice rerandomization still holds and it achieves the ideally optimal precision? 

Below we first study the asymptotic properties of the constrained Gaussian random variable $L_{K,T}$ when both $T$ and $K$ vary and possibly diverge to infinity. 
We then study the optimal best-choice rerandomization that can achieve the ideally optimal precision. 
We finally discuss some practical guidance for the choice of $T$ as well as the number of covariates $K$.

\subsection{Asymptotic properties of the constrained Gaussian random variable}

We study the asymptotic behavior of the constrained Gaussian random variable $L_{K,T}$ in \eqref{eq:L_KT} and \eqref{eq:L_KT_rep} along a sequence of varying $(K, T)$'s. 
In particular, we will allow both $K$ and $T$ to diverge to infinity along the sequence, 
and consider sufficient and necessary conditions for the constrained Gaussian random variable to be asymptotically ignorable. 
Note that the $L_{K,T}$'s for any set of $(K,T)$'s are uniformly integrable; see the supplementary material for details.
This then implies that $L_{K,T} = o_{\Pr}(1)$ if and only if its variance $v_{K,T} = o(1)$. 
Moreover, 
from Corollary \ref{cor:reduce_var}, 
$v_{K,T}$ is also closely related to the precision gain from the best-choice rerandomization. 
Therefore, in the following, we will consider mainly the asymptotic behavior of $v_{K,T}$, 
which turns out to depend critically on the ratio between $\log(T)$ and $K$. 
We summarize the results in the following theorem. 
We use $\limsup$ and $\liminf$ to denote limit superior and limit inferior, respectively. 

\begin{theorem}\label{thm:v_kt}
	Along any given sequence of $(K,T)$'s,  
	\begin{itemize}
		\item[(i)] if $\log(T) / K \to \infty$, then $v_{K, T} \to 0$;
		\item[(ii)] if $\limsup_{n \to \infty} \log(T) / K < \infty$, then $\liminf_{n\rightarrow \infty} v_{K, T} > 0$;
		\item[(iii)] if $\liminf_{n \to \infty} \log(T) / K > 0$, then $\limsup_{n\rightarrow \infty} v_{K, T} < 1$;
		\item[(iv)] if $\log(T) / K \to 0$, then $v_{K, T} \to 1$. 
	\end{itemize}
\end{theorem}

Theorem \ref{thm:v_kt} has several implications regarding the impact of the relative magnitude of the number of tried complete randomizations $T$ and the number of covariates $K$. 
First, if $T$ grows at a super-exponential rate of $K$, in the sense that $T = \exp(cK)$ for $c\to \infty$, then the constrained Gaussian random variable $L_{K,T}$ becomes asymptotically negligible. 
This indicates that the best-choice rerandomization obtains its ideally optimal efficiency, under which the covariates are also asymptotically exactly balanced. 
We emphasize that this, however, does not mean we should use as large $T$ as possible, 
because the asymptotic theory in Section \ref{sec:asym_bcr} may fail when $T$ is too large; 
see the next subsection for more detailed discussion regarding this issue. 

Second, if $T$ grows at a sub-exponential rate of $K$, in the sense that $T = \exp(cK)$ for $c\to 0$, then the variance of the constrained Gaussian random variable becomes asymptotically the same as that of the unconstrained standard Gaussian random variable. 
From Corollary \ref{cor:reduce_var}, 
the best-choice rerandomization then provides no gain on the precision of the treatment effect estimation. 
This reminds us that we should not use too many covariates and should try an appropriate number of complete randomizations for the best-choice rerandomization.

Third, if $T$ grows at an exponential rate of $K$, in the sense that $T = \exp(cK)$ for $c$ bounded away from zero and infinity,  
then the variance of the constrained Gaussian random variable will be bounded strictly between $0$ and $1$. 
In this case, the best-choice rerandomization still provides precision gain compared to the complete randomization, although there is a gap from the ideally optimal one. 

\begin{remark}\label{rmk:LKT_LKa}
    The asymptotic behavior of the constrained Gaussian random variable $L_{K, T}$ is similar to the truncated variable $L_{K,a}$ studied in \citet[][Theorem 4]{wang2022rerandomization} with $a$ being the $1/T$th quantile of the chi-squared distribution with degrees of freedom $K$. 
    This is not surprising due to similar reasons as in Remark \ref{rmk:first_type}. 
    By their representation in Proposition \ref{prop:L_KT} and \citet[][Proposition 2]{LDR18}, 
    the difference in $L_{K, T}$ and $L_{K,a}$ comes mainly from the component of the constrained chi-squared random variable. 
    Specifically, $L_{K, T}$ involves the minimum order statistic from $T$ i.i.d.\  $\chi^2_K$ random variables, 
    whereas $L_{K,a}$ involves the $\chi^2_K$ random variable given that it is bounded by its $1/T$th quantile. 
    Intuitively, both of these constrained random variables are from the smallest $1/T$ proportion of $\chi^2_K$ distribution. 
    This intuition may help explain their similar asymptotic behavior. 
    However, an obvious difference between them is that $L_{K,a}$ always has a bounded support, while $L_{K,T}$ can take value on the whole real line. Moreover, the proof of Theorem~\ref{thm:v_kt} relies on the characterization of the order statistics of multiple chi-squared random variables, which is different from its analogue in \citet{wang2022rerandomization} that focuses on analyzing a single truncated $\chi_K^2$ random variable.
\end{remark}

\subsection{Optimal best-choice rerandomization with diverging number of tries}\label{sec:opt_bcr_subsec}

We now consider the question at the beginning of this section: Can we let $T$ increase at a proper rate of the sample size so that the best-choice rerandomization can achieve the ideally optimal precision asymptotically? 
From Theorems \ref{thm:asymp}, \ref{thm:asym_equiv} and \ref{thm:v_kt}, 
such an optimal rerandomization exists if we can find $T$ such that Condition \ref{cond:iterations} holds (i.e., $T \Delta_n  \to 0$) and the condition in Theorem \ref{thm:v_kt}(i) holds (i.e., $\log (T)/K \to \infty$), 
where the former guarantees the asymptotic approximation and the latter guarantees the optimal precision. 
We summarize the results below. 

\begin{condition}\label{cond:T_opt}
    As $n\to\infty$, $\log (T)/K \to \infty$. 
\end{condition}

\begin{theorem}\label{thm:opt_bcr}
Under the best-choice rerandomization using the Mahalanobis distance, 
if Conditions \ref{cond:gamma}, \ref{cond:iterations} and \ref{cond:T_opt} hold, 
and $\limsup_{n\rightarrow\infty} R^2 < 1$, 
then, as $n\rightarrow\infty$, 
\begin{align}\label{eq:asym_Gauss}
    \sup_{c\in \mathbb{R}} \Big| \Pr\big\{ V_{\tau\tau}^{-1/2} (\hat{\tau}_{(1)} - \tau) \leq c \big\}   - 
    \Pr
    \big( \sqrt{1-R^2}\ \varepsilon_0 \leq c\big) \Big|
    \to 0, 
\end{align}
recalling that $\varepsilon_0\sim \mathcal{N}(0,1)$, 
and 
$R^2$ is defined in \eqref{eq:R2}.  
\end{theorem}

In Theorem \ref{thm:opt_bcr}, 
we additionally assume that $R^2$ is bounded away from $1$, 
which is reasonable since in practice we generally do not expect the covariates to perfectly explain all the variability in potential outcomes. 
Importantly, from Theorem \ref{thm:opt_bcr}, 
the difference-in-means estimator under the best-choice rerandomization becomes asymptotically Gaussian distributed, 
and 
achieves the ideally optimal precision with remaining variation due solely to variability in potential outcomes that cannot be linearly explained by the covariates.
In addition, it has the same asymptotic distribution as the linearly regression-adjusted estimator under the CRE \citep{lin2013, fcltxlpd2016}. 
Therefore, the best-choice rerandomization is essentially a dual of covariate adjustment, where the former is at the design stage while the latter is at the analysis stage. 
Moreover, rerandomization has the advantage of being blind to outcomes and can thus avoid data snooping, and the difference-in-means estimator is a more intuitive and transparent estimator for the average treatment effect \citep{lin2013, Rosenbaum:2010, Cox2007, Freedman2008}.

{\rev 
\begin{remark}
    For coarsely stratified experiments with a fixed number of strata and large stratum sizes, 
    if the proportions of treated units are equal across strata, then 
    the asymptotic distribution of the difference-in-means estimator has the same form as that in \eqref{eq:asym_Gauss} for rerandomization 
    \citep[see, e.g.,][Appendix A4]{YQL2021}, 
    with covariates defined as the stratum indicators. 
    Although the technical derivations for these asymptotic approximations are rather different \citep[see, e.g.,][]{ShiLi2024}, 
    similarity in their forms is quite intuitive, since, in such a coarsely stratified experiment, the covariates are exactly balanced between treated and control units. 
    For finely stratified experiments such as matched pairs, under design-based asymptotic inference, stratification can either improve or hurt the precision of the difference-in-means estimator \citep{Imai08}. 
    Recently, \citet{Bai2022}, \citet{bai2023efficiency} and \citet{cytrynbaum2021optimal} showed that, under superpopulation inference with i.i.d.~sampling of units, 
    the difference-in-means estimator can have improved efficiency through fine stratification, and moreover it can achieve the nonparametric efficiency, in the sense that fine stratification can control covariates nonparametrically and nonlinearly. 
    This is in contrast to the linear adjustment of covariates under rerandomization. 
    Below we briefly discuss two extensions of rerandomization that can better accommodate nonlinear relation between potential outcomes and covariates. 
    First, we can include more transformations and interactions of basic covariates into rerandomization,      
    noting that our theory allows the number of covariates to increase with the sample size and can thus permit sieve-type methods \citep{grenander1981abstract}.
    Second, rerandomization can also be combined with coarse or fine stratification, which has been explored recently by \citet{WWL2021}, \citet{Krieger2023} and \citet{cytrynbaum2024finely}. 
\end{remark}}

From Theorem \ref{thm:opt_bcr} and the discussion before, 
a proper choice of $T$ such that Conditions \ref{cond:iterations} and \ref{cond:T_opt} hold is crucial for designing the optimal best-choice rerandomization. 
On the one hand, $T$ should be small in the sense that $T= o(\Delta_n^{-1})$ to ensure the asymptotic approximation; 
on the other hand, $T$ should be large in the sense that $T= \exp(cK)$ with $c\rightarrow  \infty$ to ensure the optimal efficiency. 
Below we investigate under what conditions such a choice of $T$ exists. 
We summarize the results in the following theorem. 

\begin{theorem}\label{thm:opt_rerand}
    Under the best-choice rerandomization, 
    assume Condition \ref{cond:gamma} and $\limsup_{n\to \infty} R^2 < 1$.  
    \begin{itemize}
		\item[(i)] If and only if $\log(\Delta_n^{-1}) / K \to \infty$, then there exists a sequence of $T$ such that Conditions \ref{cond:iterations} and \ref{cond:T_opt} hold, under which the best-choice rerandomization achieves its optimal efficiency with the asymptotic Gaussian approximation in \eqref{eq:asym_Gauss}. 
		\item[(ii)] If $\limsup_{n \to \infty} \log(\Delta_n^{-1}) / K < \infty$, then for any sequence of $T_n$ such that Condition \ref{cond:iterations} and consequently the asymptotic approximation in \eqref{eq:bound_diff_tilde} hold, $\liminf_{n\rightarrow \infty} v_{K, T} > 0$;
		\item[(iii)] If $\liminf_{n \to \infty} \log(\Delta_n^{-1}) / K > 0$, then there exists a sequence of $T$ such that Conditions \ref{cond:iterations} and consequently the asymptotic approximation in \eqref{eq:bound_diff_tilde} hold, under which $\limsup_{n\rightarrow \infty} v_{K, T} < 1$;
		\item[(iv)] If $\log(\Delta_n^{-1}) / K \to 0$, then for any sequence of $T$ such that Condition \ref{cond:iterations} and consequently the asymptotic approximation in \eqref{eq:bound_diff_tilde} hold, 
        $v_{K, T} \to 1$ as $n\to \infty$, under which the best-choice rerandomization loses efficiency gain compared to the CRE. 
	\end{itemize}
\end{theorem}

From Theorem \ref{thm:opt_rerand}, 
under our asymptotic theory, 
the feasibility of the optimal best-choice rerandomization depends crucially on whether 
the ratio between $\log(\Delta_n^{-1})$ and $K$ can diverge to infinity. This is similar to that for the first-type of rerandomization studied in \citet{wang2022rerandomization}; see Remark \ref{rmk:optimal}. 
From \citet[][Theorem 2]{wang2022rerandomization}, 
it is not difficult to see that a sufficient condition for $\log(\Delta_n^{-1}) \gg K$ is $\log(\gamma_n^{-1}) \gg K$. 
To get more intuition, 
similar to \citet{wang2022rerandomization}, 
we consider the asymptotic rate of $\gamma_n$ assuming that the $n$ experimental units are i.i.d.\ samples from a certain superpopulation. 
Specifically, 
we invoke the following regularity condition for the sequence of superpopulations that generate the finite populations. 
Recall that $\bs{u}_i \equiv (r_0 Y_i(1) + r_1 Y_i(0), \bs{X}_i^\top)^\top$ for $1\le i \le n$. 
\begin{condition}\label{cond:superpopulation}
    Each finite population consists of i.i.d.\ $\bs{u}_i$'s from some superpopulation distribution, with $\bs{\xi}_i \equiv \Cov(\bs{u}_i)^{-1/2} (\bs{u}_i - \E \bs{\mu}_i)$ being the corresponding standardized vector. 
    Moreover, 
    for some constant $\delta>2$ and all $n$, 
    the inner product of the standardized vector and any constant unit vector has its $\delta$-th absolute moment uniformly bounded by an absolute constant, i.e., 
    $\sup_{\bs{v}\in \mathbb{R}^{K+1}: \bs{v}^\top \bs{v} = 1}\E|\bs{v}^\top \bs{\xi}_i|^\delta = O(1)$. 
\end{condition}

From \citet[][Proposition F.1]{LeiD2020highdim}, a sufficient condition for Condition \ref{cond:superpopulation} is that all coordinates of $\bs{\xi}_i$ are mutually independent and their $\delta$-th absolute moment is uniformly bounded. 
From \citet{wang2022rerandomization}, under Condition \ref{cond:superpopulation}, we have 
\begin{align}\label{eq:bound_gamma_n}
    2^{-3/2} \frac{1}{\sqrt{r_1r_0}} 
    \frac{(K+1)^{7/4}}{n^{1/2}}
    \le 
    \gamma_n 
    = 
    O_{\Pr} \left( \frac{1}{\sqrt{r_1r_0}} 
    \frac{(K+1)^{7/4}}{n^{1/2-1/\delta}} \right). 
\end{align}
In \eqref{eq:bound_gamma_n}, the lower bound of $\gamma_n$ always hold regardless of Condition \ref{cond:superpopulation}, 
and it implies that the upper bound in \eqref{eq:bound_gamma_n} is precise up to an $n^{1/\delta}$ factor. 
Suppose that both $r_1$ and $r_0$ are strictly bounded away from zero and one, which is reasonable in practice so that there is nonnegligible proportions of units receiving both treatment arms. 
Below we will consider three cases for the rate of $K$ such that Condition \ref{cond:gamma} holds with high probability, and the resulting ``largest'' choice of $T$ 
(ignoring subpolynomial terms) such that Condition \ref{cond:iterations}, as well as the asymptotic approximation, for rerandomization holds with high probability. 
Let $\kappa=1/3$ (or $1$ if the conjecture in \citet{R15} holds). 
\begin{enumerate}[label=(\roman*)]
    \item[(i)] [$K = o(\log n)$.] In this case, 
    $\gamma_n = o_{\Pr}(n^{-(1/2-1/\delta)} (\log n)^{7/4})$,
    and we can choose $T \asymp n^{\beta}$ with $0 < \beta <\kappa/2 - \kappa/\delta$. 
    Consequently, $v_{K,T}=o(1)$ as implied by Theorem \ref{thm:v_kt}(i), and the best-choice rerandomization achieves the ideal optimal efficiency. 

    \item[(ii)] [$K \asymp \log n$.] In this case, $\gamma_n = O_{\Pr}(n^{-(1/2-1/\delta)} (\log n)^{7/4})$, and we can choose $T \asymp n^{\beta}$ with $0 < \beta < \kappa/2 - \kappa/\delta$. 
    Consequently, $0 < \liminf_{n\rightarrow \infty} v_{K,T} \le \limsup_{n\rightarrow \infty}v_{K,T} < 1$ as implied by Theorem \ref{thm:v_kt}(ii) and (iii), and the best-choice rerandomization has nonnegligible gain over the complete randomization, although there is still a gap from the ideally optimal precision. 

    \item[(iii)] [$K \asymp n^{\zeta}$ with $\zeta \in (0, 2/7-4/(7\delta))$.] 
    In this case, $\gamma_n = O_{\Pr}(n^{-(1/2-1/\delta-7\zeta/4)})$, and we can choose $T \asymp n^{\beta}$ with $0 < \beta < \kappa/2 - \kappa/\delta - 7\zeta \kappa/4 $. 
    Consequently, $v_{K,T} = 1 - o(1)$ as implied by Theorem \ref{thm:v_kt}(iv), and the best-choice rerandomization provides no gain over the 
    CRE. 
\end{enumerate}

The above theoretical results suggest that we should use at most $O(\log n)$ number of covariates in rerandomization. 
Importantly, we should use those covariates relevant for the potential outcomes (measured by the corresponding $R^2$ as in \eqref{eq:R2}).  
In addition, we can try multiple complete randomizations with the number of tries {\rev $T$} increasing polynomially with the sample size, in order to maximize the efficiency gain while maintaining the robustness; {\rev see also the next subsection for more discussion on the choice of $T$.}

\begin{remark}\label{rmk:optimal}
The conditions for achieving the optimal precision under both types of rerandomization are actually equivalent when $p=1/T$, where $p$ denotes the approximate acceptance probability for the first type and $T$ denote the number of tried complete randomizations for the second type. 
This is actually a direct consequence of the similarity discussed in Remarks \ref{rmk:first_type} and \ref{rmk:LKT_LKa}. 
\end{remark}

\subsection{\rev Guidance on the choice of $T$}\label{sec:prac}

{\rev 
From Section \ref{sec:opt_bcr_subsec}, when we increase the number of tried randomizations $T$ properly with the sample size, the best-choice rerandomization can achieve the ideal precision that one can expect even with perfectly balanced covariates.
Nevertheless, we emphasize that we generally do not want $T$ to be too large at any given finite sample size. 
Note that setting $T = \infty$ leads to an ``optimal'' design that uses only assignments minimizing the covariate imbalance and is thus almost deterministic in the presence of continuous covariates, which has been recommended by, e.g., \citet{Kasy2016} from a decision-theoretic perspective. 
Below we review and discuss arguments against the choice of a too large $T$, which also shed light on the choice of $T$ in practice. 

First, \citet{Banerjee20} recently studied the choice of $T$ in the best-choice rerandomization from an
ambiguity averse decision making perspective, 
and showed that rerandomization can involve a robustness loss of order $\log(T)/n$. In particular, they show that rerandomization with $K$ increasing exponentially with the sample size can lead to nonvanishing robustness loss. 
They therefore recommend $K\le n$, ensuring that the loss is at most of order $\sqrt{\log(N)/N}$. 
See also \citet[][Section VB]{Banerjee20} for a more detailed comparison of deterministic and randomized designs, including contexts in which each may be preferable.

Second, in our design-based inference, we require $T$ to not increase too fast with the sample size, in order for our asymptotic theory to work. 
In particular, our theory generally requires $T$ to increase at most polynomially with the sample size, which is more stringent than the requirement in \citet{Banerjee20} for a vanishing robustness loss. 
We do want to emphasize that our requirement on the growing rate of $T$ is only sufficient for the asymptotic design-based theory of rerandomization. 
Whether it is possible to relax such a requirement 
is still an open question. 
Nevertheless, our theory still shows that, when $T$ increases properly with the sample size, the best-choice rerandomization can still achieve the ideal precision we expect even with  perfectly balanced covariates as when $T=\infty$, while maintaining its design-based robustness property. 
Relatedly, \citet{morgan2012rerandomization} has proposed Fisher randomization test for sharp null hypotheses as an alternative design-based approach for rerandomization. 
If $T$ is too large or $T=\infty$, the number of possible randomizations may be too limited, rendering the randomization test powerless.

Third, we can follow \citet{wang2022rerandomization} to conduct the worst-case analysis for the best-choice rerandomization at any finite sample size. 
Specifically, consider any design $\mathcal{D}$ that assigns $r_1$ proportion of units into treatment and $r_0 = 1-r_1$ proportion of units into control. 
As shown in \citet[][Proposition A1]{wang2022rerandomization}, 
the worst-case bias and root mean squared error of the difference-in-means estimator $\hat{\tau}$ under the design $\mathcal{D}$, standardized by the corresponding root mean squared error (or equivalently standard deviation) under the CRE, has the following equivalent forms: 
\begin{align}
		\max_{\bs{y} \ne \bs{0}}V_{\tau\tau}^{-1/2} \left| \E_{\mathcal{D}}(\hat{\tau} - \tau) \right|
		& = 
		\sqrt{\frac{n-1}{n r_1 r_0}} \cdot \left\| \bs{\pi} - r_1 \bs{1}_n \right\|_2 \ge 0, 
        \nonumber
		\\
		\label{eq:max_mse}
		\max_{\bs{y} \ne \bs{0}}V_{\tau\tau}^{-1/2}  \sqrt{\E_{\mathcal{D}}\{(\hat{\tau} - \tau)^2\} }
		& = 
		\sqrt{\frac{n-1}{n r_1 r_0}} \cdot \lambda_{\max}^{1/2}\left(  \bs{\Omega} + (\bs{\pi}-r_1 \bs{1}_n) (\bs{\pi}-r_1 \bs{1}_n)^\top \right) \ge 1, 
	\end{align}
where $V_{\tau\tau}$ is defined as in \eqref{eq:V}, $\bs{y} = (y_1, \ldots, y_n)^\top$ with $y_i = r_0 \{ Y_i(1) - \bar{Y}(1) \} + r_1 \{ Y_i(0) - \bar{Y}(0) \}$ being a weighted average of unit $i$'s centered potential outcomes, 
$\bs{\pi} \equiv \E_{\mathcal{D}} (\bs{Z})$ and $\bs{\Omega} \equiv \Cov_{\mathcal{D}}(\bs{Z})$ denote the mean and covariance matrix of the treatment assignment vector under the design $\mathcal{D}$, and $\lambda_{\max}(\cdot)$ denotes the largest eigenvalue of a matrix.  
These then enable us to perform finite-sample diagnosis for various designs, including the best-choice rerandomization with various $T$ and covariates; see the simulations in Section \ref{sec:numeric}. 
In particular, from the simulation in Section \ref{sec:star_main} with details in the supplementary material, the worst-case mean squared error for rerandomization tends to increase with both number of tries $T$ and number of covariates $K$. As a side note, trimming the extreme values of covariates turns out to be quite effective for reducing the worst-case mean squared error.
See Appendices A1 and A2 for more discussion and numerical illustration on practical guidance for designing rerandomization, including the choice of both $T$ and $K$.

Fourth, in general, letting $T=\infty$ or finding the assignments with the minimum covariate imbalance is computationally challenging. 
The theoretical results discussed above suggest that a proper choice of $T$ not only reduces the computational burden, but also ensures the robustness and efficiency of the design. 
In addition, \citet{Kasy2016} has also suggested the best-choice rerandomization with a finite $T$, say $T=500$, as a practical way to implement the design that tries to minimize the covariate imbalance. 

}

\section{Statistical inference under the best-choice rerandomization}\label{sec:inf}

We now consider the statistical inference for the average treatment effect $\tau$ under the best-choice rerandomization. 
From the asymptotic approximation in Theorem \ref{thm:asymp} and \eqref{eq:asym_equiv}, 
it suffices to estimate $V_{\tau\tau}$ and $R^2$. 
By their definition in \eqref{eq:V} and \eqref{eq:R2}, we essentially need to estimate the finite population variances of potential outcomes and their linear projections on covariates. 
We estimate these quantities by their sample analogues. 
Specifically, for $z = 0, 1$, let $s^2_{z}$ be
the sample variance of the observed outcomes for units receiving treatment arm $z$, 
and $s_{z\bs{x}} = s_{\bs{x}z}^\top$ be the corresponding sample covariance between observed outcomes and covariates.
We further define 
$s^2_{z\mid \bs{x}} = s_{z\bs{x}} (\bs{S}_{\bs{x}}^2)^{-1} s_{\bs{x}, z}$
as a sample analogue of $S^2_{z\mid \bs{x}}$, 
and 
$s^2_{\tau \mid \bs{x}} = (s_{1\bs{x}} - s_{0\bs{x}}) (\bs{S}_{\bs{x}}^2)^{-1} (s_{\bs{x}1} - s_{\bs{x}0})$ as a sample analogue of $S^2_{\tau \mid \bs{x}}$. 
Let $s^2_{z\setminus \bs{x}} = s^2_z - s^2_{z\mid \bs{x}}$ for $z=0,1$. 
We can then estimate $V_{\tau\tau}$ and $R^2$ by: 
\begin{align}\label{eq:VR2_hat}
	\hat{V}_{\tau\tau} = n_1^{-1} s_{1}^2 + n_0^{-1} s_{0}^2 - n^{-1} s_{\tau \mid \bs{x}}^2, 
	\quad 
	\hat{R}^2 =  
 1 - 
	\hat{V}_{\tau\tau}^{-1} \big(n_1^{-1} s_{1 \setminus \bs{x}}^2 + n_0^{-1} s_{0 \setminus \bs{x}}^2 \big),
\end{align}
By the asymptotic approximation in Theorem \ref{thm:asymp} and \eqref{eq:asym_equiv}, 
for any $\alpha\in (0, 1)$, 
we then propose the following $1-\alpha$ confidence interval for the average treatment effect $\tau$:
\begin{align}\label{eq:C_hat}
\hat{\mathcal{C}}_{\alpha} =  \big[\hat\tau_{(1)} - \hat{V}_{\tau\tau}^{1/2} \cdot \nu_{1-\alpha / 2, K, T} (\hat{R}^2), \ \ \hat\tau_{(1)} + \hat{V}_{\tau\tau}^{1/2} \cdot \nu_{1-\alpha/2, K, T} (\hat{R}^2)\big].
\end{align}
Importantly, the quantile $\nu_{1-\alpha / 2, K, T} (\cdot)$ can be efficiently estimated by the Monte Carlo method using the representation in Proposition \ref{prop:L_KT}.

As demonstrated in the theorem below, 
under certain regularity conditions, we can derive the probability limits of the estimators in \eqref{eq:VR2_hat} and prove the asymptotic validity of the confidence interval in \eqref{eq:C_hat}. 
Let $S^2_{\tau\setminus \bs{x}} = S^2_{\tau} - S^2_{\tau\mid \bs{x}}$ denote the finite population variance of the residuals from the linear projection of individual effects on covariates. 
We then invoke the following 
condition.
\begin{condition}\label{cond:infer}
	As $n\rightarrow \infty$, 
	\begin{align}\label{eq:infer_cond}
		\frac{\max_{z\in \{0,1\}}\max_{1\le i \le n}\{Y_i(z) - \bar{Y}(z)\}^2}{r_0 S^2_{1\setminus \bs{x}} + r_1 S^2_{0\setminus \bs{x}}}
		\cdot 
		\frac{\max\{K, 1\}}{r_1r_0}
		\cdot 
		\sqrt{ \frac{\max\{1, \log K, \log T\} }{n} } 
		\rightarrow 0. 
	\end{align}
\end{condition}

\begin{theorem}\label{thm:inf}
	Under the best-choice rerandomization and Conditions~\ref{cond:gamma},~\ref{cond:iterations} and~\ref{cond:infer}, 
	\begin{itemize}
		\item[(i)] the estimators in \eqref{eq:VR2_hat} have the following probablility limits:
		\begin{align*}
        \max\big\{ 
        |\hat{V}_{\tau\tau}(1-\hat{R}^2) - 
        V_{\tau\tau}(1-R^2) - S^2_{\tau \setminus \bs{x}}/n|, 
        \ 
        |\hat{V}_{\tau\tau} \hat{R}^2 - 
        V_{\tau\tau} R^2|
        \big\}
        = 
        o_{\pr}\big( V_{\tau\tau}(1-R^2) + S^2_{\tau\setminus \bs{x}} /n \big); 
		\end{align*}
		\item[(ii)] for any $\alpha\in (0,1)$, 
		the $1-\alpha$ confidence interval in \eqref{eq:C_hat} is asymptotically conservative,  
		in the sense that
		$
		\liminf_{n\rightarrow \infty}
		\pr
		(
		\tau \in \hat{\mathcal{C}}_{\alpha})
		\ge 
		1 - \alpha; 
		$
		\item[(iii)] if further $S^2_{\tau \setminus \bs{x}} = n V_{\tau\tau}(1-R^2)\cdot o(1),$ the $1-\alpha$ confidence interval in \eqref{eq:C_hat} becomes asymptotically exact, 
		in the sense that 
		$
		\lim_{n\rightarrow \infty}
		\pr
		(
		\tau \in \hat{\mathcal{C}}_{\alpha})
		= 
		1 - \alpha. 
		$
	\end{itemize}
\end{theorem}

Below we give several remarks regarding Theorem \ref{thm:inf}. 
First, by the same logic as Corollary \ref{cor:reduce_qr}, the confidence interval $\hat{\mathcal{C}}_{\alpha}$ is always shorter than or equal to \citet{Neyman:1923}'s one for the CRE, while still being asymptotically conservative. This demonstrates the gain in inference from rerandomization.

Second, in addition to Conditions \ref{cond:gamma} and \ref{cond:iterations}, 
the large-sample inference in 
Theorem \ref{thm:inf} 
requires additionally Condition \ref{cond:infer}. 
From \citet[][Corollary 2(iii)]{wang2022rerandomization}, 
this additional condition can be guaranteed by moderate conditions on the moments of the potential outcomes. 

Third, Theorem \ref{thm:inf}(i) implies that the estimators in \eqref{eq:VR2_hat} are consistent for their population analogues. 
Intuitively, $\hat{V}_{\tau\tau}\hat{R}^2$ is consistent for $V_{\tau\tau}R^2$, while $\hat{V}_{\tau\tau}$ itself is only conservative for $V$, due to the individual treatment effect heterogeneity $S^2_{\tau\setminus\bs{x}}$ that cannot be linearly explained by the covariates. 
This is a feature of the finite population inference \citep{Neyman:1923}. 

Fourth, Theorem \ref{thm:inf}(ii) shows the asymptotic validity of the confidence interval in \eqref{eq:C_hat}. 
The confidence intervals are generally conservative, due to the conservativeness in estimating $V_{\tau\tau}$ as discussed before. 
From Theorem \ref{thm:inf}(iii), 
when the effect heterogeneity $S^2_{\tau \setminus \bs{x}}$ is asymptotically negligible, the confidence intervals become asymptotically exact. 

Fifth, $s^2_{1\setminus \bs{x}}$ and $s^2_{0\setminus \bs{x}}$ are almost equivalent to the sample variances of the residuals from the linear regression of observed outcomes on covariates in treatment and control groups, respectively. Similar to \citet{LeiD2020highdim} and \citet{wang2022rerandomization}, we can consider rescaling the residuals as in usual regression analysis \citep{mackinnon2013thirty} to improve its finite sample performance; see the simulation studies in Section \ref{sec:star_main}.

If additionally Condition \ref{cond:T_opt} holds and $\limsup_{n\rightarrow\infty} R^2 < 1$, then, from Theorem \ref{thm:opt_bcr}, the difference-in-means estimator under 
rerandomization will become asymptotically Gaussian distributed. 
In this case, 
we can also
use the Wald-type confidence intervals. 
Let $z_{\alpha}$ denote the $\alpha$th quantile of the standard Gaussian distribution. 

\begin{theorem}\label{thm:inf_gaussian}
    Under the best-choice rerandomization, 
    if Conditions \ref{cond:gamma}, \ref{cond:iterations}, \ref{cond:T_opt} and \ref{cond:infer} hold, 
    and $\limsup_{n\rightarrow\infty} R^2 < 1$, 
    then 
    Theorem~\ref{thm:inf} still holds with $\hat{\mathcal{C}}_{\alpha}$ replaced by 
    the Wald-type confidence interval
    $\tilde{\cC}_{\alpha} = 
    \hat{\tau}_{(1)} \pm \hat{V}_{\tau\tau}^{1/2} (1 - \hat{R}^2)^{1/2} \cdot z_{1 - \alpha / 2}.
    $ 
\end{theorem}

In practice, we generally suggest using the confidence interval $\hat{\mathcal{C}}_{\alpha}$ in \eqref{eq:C_hat}, since it requires weaker regularity conditions.
In addition, when Condition \ref{cond:T_opt} holds and $\limsup_{n\rightarrow\infty} R^2 < 1$, the constrained Gaussian random variable $L_{K,T}$ is $o_{\Pr}(1)$, and, intuitively, the confidence interval $\hat{\mathcal{C}}_{\alpha}$ in \eqref{eq:C_hat} will be 
close 
to the Wald-type confidence interval $\tilde{\cC}_{\alpha}$.

{\rev 
\section{Regression adjustment under the best-choice rerandomization}\label{sec:regadj}

In this section, we study 
regression adjustment under the best-choice rerandomization. Suppose we observe covariate vector $\bs{w}_i \in \R^J$, for $1\le i \le n$, when analyzing the experiment, where $J$ detnotes the dimension of the covariate vector.
Similar to Sections \ref{sec:diff_out_cov_cre} and \ref{sec:inf}, 
let $\bs{S}_{\bs{w}}^2$ and $\bs{S}_{\bs{w}z}$
denote the finite population covariances for covariates and potential outcomes, and $\bs{s}_{\bs{w}z}$ denote the sample covariance between covariates and observed outcomes for units under treatment arm $z$, for $z=0,1$.
We consider the following linearly regression-adjusted estimator: 
\begin{equation}\label{eq:taureg}
\hat{\tau}_{(1)}(\hat{\bs{\beta}}_1, \hat{\bs{\beta}}_0) \equiv \frac{1}{n_1} \sum_{i=1}^n Z_{(1)i} \{Y_i - \hat{\bs\beta}_1^\top (\bs{w}_i - \bar{\bs{w}})\} - \frac{1}{n_0} \sum_{i=1}^n (1 - Z_{(1)i}) \{Y_i - \hat{\bs\beta}_0^\top (\bs{w}_i - \bar{\bs{w}})\},
\end{equation}
In \eqref{eq:taureg}, the subscripts $(1)$ emphasizes that the best-choice rerandomization uses the assignment with minimum covariate imbalance, and $\hat{\bs{\beta}}_z = (\bs{S}_{\bs{w}}^2)^{-1} \bs{s}_{\bs{w}z}$ is an estimated adjustment coefficient targeting for
\[
\tilde{\bs{\beta}}_z \equiv \argmin_{\bs{\beta}} \sum_{i=1}^n \{Y_i(z) - \bar{Y}(z) - \bs{\beta}^\top (\bs{w}_i - \bar{\bs{w}})\}^2 = (\bs{S}_{\bs{w}}^2)^{-1} \bs{S}_{\bs{w}z}, 
\quad (z=0,1)
\]
which is the preferred adjustment coefficient and enjoys certain optimality as suggested by~\citet{lin2013}, \citet{LD20reg} and \citet{wang2022rerandomization}. 

To facilitate the discussion on the asymptotic properties of the regression-adjusted estimator and the corresponding regularity conditions, we introduce some additional notation. 
Similar to ~\citet{wang2022rerandomization}, we introduce the residual potential outcomes $\tilde{e}_i(z) \equiv Y_i(z) - \bar{Y}(z) - \tilde{\bs{\beta}}_z^\top (\bs{w}_i - \bar{\bs{w}})$ for $1\le i \le n$ and $z=0,1$, and define $\tilde{\Delta}_n, \tilde{\gamma}_n$ and $\tilde{R}^2$ the same as $\Delta_n, \gamma_n$ and $R^2$ in \eqref{eq:R2}--\eqref{eq:gamma}, but with $Y_i(z)$ replaced by $\tilde{e}_i(z)$. 
Define further $\rho^2$ the same as $R^2$, but with $\bs{x}_i$ replaced by $\bs{w}_i$. 
We invoke the following regularity conditions.

\begin{condition}\label{cond:regrem_Delta}
    Conditions \ref{cond:gamma} and \ref{cond:iterations} hold with 
     $\gamma_n$ and $\Delta_n$ replaced by $\tilde{\gamma}_n$ and $\tilde{\Delta}_n$. 
\end{condition}

\begin{condition}\label{cond:regrem}
	As $n \to \infty$, 
	\begin{align*}
	    \frac{\max_{z \in \{0,1\}} \max_{1 \leq i \leq n} |Y_i(z) - \bar{Y}(z)|}{\sqrt{V_{\tau\tau} (1-\rho^2) \{ 1-\tilde{R}^2\}}}  \cdot J \cdot \frac{\max\{1, \ \log J, \  \log T\}}{n r_1^2 r_0^2} \to 0.
	\end{align*}
\end{condition}

\begin{theorem}\label{thm:regrem}
	Under the best choice rerandomization and Conditions \ref{cond:regrem_Delta} and \ref{cond:regrem}, as $n \to \infty$,
		\begin{align*}%
		\sup_{c\in \mathbb{R}} \left| \Pr \left\{ 
		\frac{\hat{\tau}_{(1)}(\hat{\bs{\beta}}_1, \hat{\bs{\beta}}_0) - \tau}{
		\sqrt{V_{\tau\tau}(1-\rho^2)}}
		\le c \right\} 
		- \Pr\left\{ \sqrt{1-\tilde{R}^2}\ \varepsilon_0  + \sqrt{\tilde{R}^2} \ L_{K, T} \le c \right\}
		\right| 
		\rightarrow 0; 
	\end{align*}
	If further 
    Condition \ref{cond:T_opt} holds and  $\limsup_{n \to \infty} \tilde{R}^2 < 1$, then 
	\begin{align*}
		\sup_{c\in \mathbb{R}} \left| \Pr \left\{ 
		\frac{\hat{\tau}_{(1)}(\hat{\bs{\beta}}_1, \hat{\bs{\beta}}_0) - \tau}{
		\sqrt{V_{\tau\tau}(1-\rho^2)}}
		\le c \right\} 
		- \Pr\left\{ \sqrt{1-\tilde{R}^2}\ \varepsilon_0  \le c \right\}
		\right| 
		\rightarrow 0. 
	\end{align*}
\end{theorem}

Theorem \ref{thm:regrem} allows diverging numbers of covariates in both design and analysis. Importantly, when the covariates in analysis contain those in design, in the sense that  $\bs{x}_i$'s correspond to subvectors of $ \bs{w}_i$'s, $\tilde{R}^2$ reduces to zero, and the regression-adjusted estimator will always have asymptotic Gaussian
distribution. 
That is, $V_{\tau\tau}^{-1/2} \{ \hat{\tau}_{(1)}(\hat{\bs{\beta}}_1, \hat{\bs{\beta}}_0) - \tau \}$ is asymptotically Gaussian with mean zero and variance $1-\rho^2$. 
When $\bs{w}_i = \bs{x}_i$, this is the same as the difference-in-means estimator under the optimal rerandomization. 
However, when we can observe more covariate information in analysis, regression adjustment can provide substantial gain by adjusting the imbalance of these additional covariates. 
In addition, we can estimate $\rho^2$ and $\tilde{R}^2$ in a similar way as that in \eqref{eq:VR2_hat}, which can further lead to confidence intervals for the average treatment effect based on the regression-adjusted estimator.

}

\section{Numerical illustration}\label{sec:numeric}

\subsection{The Student Achievement and Retention Project}\label{sec:star_main}

To illustrate the performance of the best-choice rerandomization 
using various 
numbers of covariates and tried complete randomizations, 
we conduct a simulation using the Student Achievement and Retention Project \citep{angrist2009incentives} 
similarly 
as in \citet{wang2022rerandomization}. 
This also facilitates the comparison between the best-choice rerandomization and the first type of rerandomization 
studied there. 
To avoid the paper being too lengthy, 
we relegate the simulation details and results to the supplementary material.

\subsection{Mobile Banking in Bangladesh}\label{sec:mobile_bank}

We now illustrate the gain of the best-choice rerandomization using an actual rerandomized experiment recently conducted by 
\citet{Lee2021}, which aims to study the effect of mobile technology on reducing inequality through the modern ways of money transfer. 
Specifically, the experiment involves rural household-urban migrant pairs at two connected sites: Gaibandha district in Northwest Bangladesh and Dhaka District at the capital of Bangladesh, where each pair can be viewed as an experimental unit. 
\citet{Lee2021} randomized the pairs into treatment and control following the rerandomization procedure as described in \citet{Bruhn:2009}, 
and the treated pairs would receive training about how to sign up for and use the mobile banking service provided by bKash, which is the largest provider of such services in Bangladesh. 
The actual rerandomization 
uses imbalance measure other than the Mahalanobis distance. 
Nevertheless, we can use the data from this experiment to illustrate the potential gain of rerandomization over the complete randomization, and to compare rerandomization with other designs.  
Similar to \citet{Bai2022}, we include six pretreatment covariates for migrants: household size, age, gender, whether completed primary school, and total remittance in the past seven months and expenditure in the past month (in $1000$ taka) from a baseline survey right before the intervention, 
where missing values are imputed in the same way as in \citet{Bai2022}. 
We focus on two post-treatment outcomes: total remittance in the past seven months and expenditure in the past month from an endline survey one year after the intervention. 
To make the simulation realistic, we compute the average treatment effect estimate using difference-in-means, 
and 
pretend that the treatment effects are constant across units and equal to the average treatment effect estimate. 
In this way, we are able to impute all the missing potential outcomes from the observed ones.

\begin{table}[ht]
\centering
\caption{Standardized mean squared errors (MSEs) for the difference-in-means of outcomes, pretreatment covariates, and the worst-case scenario, under various designs, including the optimal matched pair design (Pair), the Gram–Schmidt walk design (GSW), the finite selection model (FSM), the best-choice rerandomization with $T=10^3$ (BCR1) and $T=10^4$ (BCR2), and the complete randomization (CRE). 
These MSEs are standardized by the corresponding true MSEs under the CRE.
The first two rows are for the potential outcomes imputed from the observed data, which show the precision for treatment effect estimation. The next six rows are for the six pretreatment covariates, which show the balance of these covariates. 
The last row shows the mean squared errors in the worst-case scenario over all possible configurations of potential outcomes. 
}
\label{tab:table_multi_design}
\begin{tabular}{lcccccc}
\toprule
                        & Pair & GSW & FSM & BCR1 & BCR2 & CRE \\
\midrule
Expenditure     & 0.561 & 0.899 & 0.730 & 0.582 & 0.569 & 1.002 \\
Remittances     & 0.757 & 0.901 & 0.676 & 0.865 & 0.863 & 0.998 \\
\midrule
Baseline expenditure    & 0.079 & 0.537 & 0.005 & 0.057 & 0.026 & 1.001 \\
Baseline remittances    & 0.077 & 0.053 & 0.003 & 0.057 & 0.026 & 1.000 \\
Household size          & 0.069 & 0.863 & 0.006 & 0.057 & 0.026 & 1.005 \\
Age                & 0.076 & 0.946 & 0.006 & 0.057 & 0.026 & 1.000 \\
Female             & 0.018 & 0.918 & 0.006 & 0.057 & 0.026 & 1.003 \\
Completed primary school & 0.020 & 0.981 & 0.004 & 0.057 & 0.026 & 1.000 \\
\midrule
Worst case              & 2.080 & 1.065 & 35.936 & 1.138 & 1.144 & 1.056 \\
\bottomrule
\end{tabular}
\end{table}

{\rev 
We consider the following designs for this experiment: 
(i) the optimal matched pair design \citep{GLSR2004,Bai2022}, (ii) the Gram–Schmidt walk design \citep{Harshaw24},  (iii) the finite selection model \citep{chattopadhyay2022balanced}, (iv) the best-choice rerandomization, and (v) the benchmark complete randomization. We randomly remove one unit so that we can assign exactly half of the units into treatment. 
For the optimal matched pair design, we construct the pairs  by minimizing the sum of Mahalanobis distances between the pairs using  non-bipartite matching \citep{beck2016nbpmatching}. 
For the Gram–Schmidt walk design, we set the parameter $\phi$ for balance–robustness trade-off to be $0.5$. 
For the best-choice rerandomization, we use the Mahalanobis distance to measure covariate imbalance and consider  two choices of $T$, $10^3$ and $10^4$, for the tried randomizations. 
We simulate $10^6$ assignments from each of these designs. 
Table \ref{tab:table_multi_design} summarizes the main simulation results. 
The first two rows show the standardized mean squared errors (MSEs) of the difference-in-means estimator for the two outcomes under these designs. 
The MSEs are standardized by the corresponding true MSE under the CRE; 
obviously, the standardized MSEs for the CRE are $1$ up to some Monte Carlo errors. 
The third to the eighth rows show the standardized MSEs of the difference-in-means of all the six pretreatment covariates with respect to $0$, which measure the covariate balance under these designs. 
The standardized MSEs under the best-choice rerandomization are the same for all covariates, consistent with the ``equal percent variance reducing'' property discussed in Section \ref{sec:improve}. 
The last row shows the worst-case standardized MSEs for these designs over all possible configurations of potential outcomes; see \eqref{eq:max_mse}.
Figure \ref{fig:hist_all} supplements Table \ref{tab:table_multi_design} with the distributions of the difference-in-means estimator for the two outcomes under these designs. 

From Table \ref{tab:table_multi_design}, 
the finite section model has the best covariate balance, followed by the matched pair design and the best-choice rerandomization. The covariate balance for the Gram–Schmidt walk design can be improved by choosing a smaller parameter $\phi$\footnote{We set $\phi=0.5$ here. \citet{Harshaw24} heuristically suggested setting $\phi$ to a value no less than $0.5$.}. 
In terms of the two outcomes imputed from the real data, the matched pair design has the best overall precision, followed by the finite section model and the best-choice rerandomization, which exhibit comparable performance.
Note that rerandomization has notably higher improvement for the expenditure outcome. 
The reason is that the pre-treatment covariates have stronger linear association with the potential expenditures, and the corresponding $R^2$ measure defined as in \eqref{eq:R2} is about $0.45$.  In contrast, the $R^2$ measure is about $0.15$ for the remittance outcome; we could potentially improve the estimation precision for the remittance outcome by including transformations and interactions of the basic covariates in the rerandomization design. 
Lastly, in terms of the worst-case performance over all possible potential outcome configurations, 
the CRE is best due to its robustness and minimax optimal property \citep{wang2022rerandomization}. 
The Gram–Schmidt walk design and the best-choice randomization are also quite robust, with the worst-case standardized MSEs close to $1$, while the matched pair design doubles the worst-case MSE compared to the CRE and thus can be less robust than the former two.
The finite selection model has significantly higher worst-case MSE, which may not be surprising given that it has much more balanced pretreatment covariates. 
From the above, 
the best-choice rerandomization has  performance comparable to other designs. 
Moreover, it can be further incorporated into other designs. 
For example, we can use rerandomization to further improve covariate balance 
for matched pairs 
and consequently the precision of treatment effect estimation \citep{WWL2021, Krieger2023, cytrynbaum2024finely}. }

\begin{figure}[htb]
	\centering
    \begin{subfigure}[htbp]{\textwidth}
        \centering
        \includegraphics[width=0.8\textwidth]{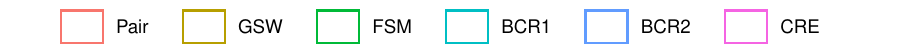}

    \end{subfigure}
    \begin{subfigure}[htbp]{0.5\textwidth}
        \centering
        \includegraphics[width=1\textwidth]{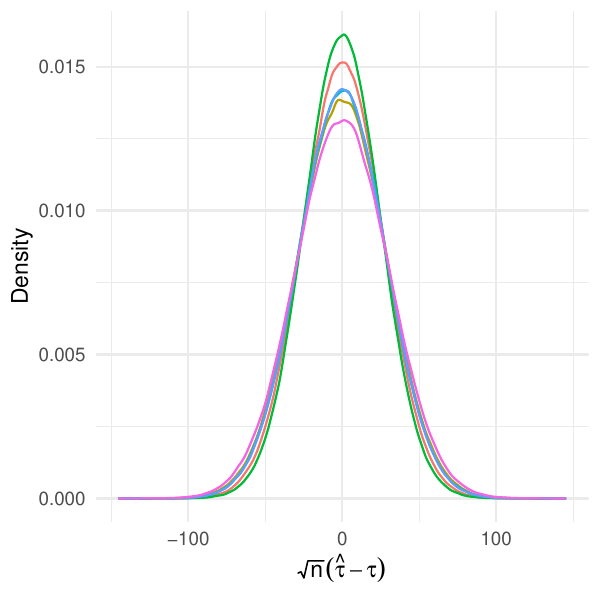}
        \caption{\centering Remittance}
    \end{subfigure}%
    \begin{subfigure}[htbp]{0.5\textwidth}
        \centering
        \includegraphics[width=1\textwidth]{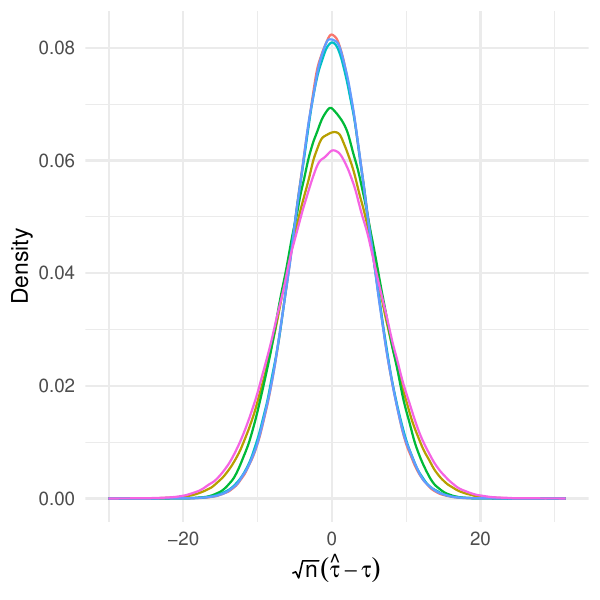}
        \caption{\centering Expenditure}
    \end{subfigure}
    \caption{
    Histograms of the difference-in-means estimator with respect to the two imputed potential outcomes under various designs, including the optimal matched pair design (Pair), the Gram–Schmidt walk design (GSW), the finite selection model (FSM), the best-choice rerandomization with $T=10^3$ (BCR1) and $T=10^4$ (BCR2), and the complete randomization (CRE).     
    } 
    \label{fig:hist_all}
\end{figure}

We then consider the performance of the proposed confidence confidence intervals. 
For the remittance analysis, 
averaging over $10^6$ simulated assignments, 
the coverage probabilities of our $95\%$ confidence interval for the two best-choice rerandomization designs (with $T=10^3$ and $T=10^4$) and \citet{Neyman:1923}'s confidence interval for the CRE are, respectively, $0.948$, $0.948$ and $0.949$. 
Analogously, for the expenditure analysis, the coverage probabilities of our and Neyman's confidence intervals are, respectively, 
$0.947$, $0.946$ and $0.948$. 
All of the coverage probabilities are close to the nominal level. 
This is coherent with our theory since the treatment effects are constant across all units in our simulation, under which these confidence intervals will be asymptotically exact. 
Moreover, our confidence intervals under the best-choice rerandomization is considerably shorter than Neyman's one under the CRE, which demonstrates the gain in inference efficiency from rerandomization. 
For example, the percentage reduction in average length of confidence intervals under the best-choice rerandomization with $T=1000$, compared to that under the CRE, 
is $7.3\%$ for the the remittance analysis 
and $24.1\%$ for the expenditure analysis. 
These essentially lead to $16.4\%$ and $73.7\%$ increase in effective sample size, respectively.

\section{Discussion}\label{sec:discuss}

The best-choice rerandomization design is an intuitive design and has already been widely applied in empirical studies~\citep{Bruhn:2009}. Nevertheless, a theoretical framework characterizing its asymptotic property remains absent. In this paper,
we studied 
the large-sample inference for the best-choice rerandomization using the Mahalanobis distance
and its optimality, allowing sample-size-dependent number of covariates and number of tried complete randomizations. 
We showed that (i) rerandomization can outperform usual complete randomization in terms of both estimation precision and length of confidence intervals, (ii) it should incorporate appropriate number of pretreatment covariates that are relevant for the potential outcomes, 
and (iii) the number of tried complete randomizations should be large but not overly large,  increasing at most at a polynomial rate with respect to the sample size.

In this paper we focus mainly on the best-choice rerandomization based on the completely randomized experiments and using the Mahalanobis distance measure. 
Like the existing literature on the first type of rerandomization, 
it will be interesting to further extend it to other covariate imbalance measure \citep[e.g.,][]{BS2021, ZYR2021, ZD2021rerand, liu2023bayesian} and other randomized experiments \citep[e.g.,][]{BDR2016, LDR20factorial, JS2022, WWL2021, Krieger2023, cytrynbaum2024finely}. It will also be interesting to study 
Fisher randomization test for the average treatment effect under the best-choice rerandomization \citep[][]{ZHAO2021278}. We leave these for future study. 

We also want to point out that 
the main purpose of this paper is not to compare the two types of rerandomization, but rather to provide large-sample inference tools for practitioners who design and analyze experiments using the second type of rerandomization. 
From Remarks \ref{rmk:first_type}--\ref{rmk:optimal} and Section \ref{sec:star_main}, the two types of rerandomization share similarity in both theory and practice. 
However, a comprehensive comparison between them is beyond the scope of this paper and needs further investigation.

\bibliographystyle{plainnat}
\bibliography{reference}

\begin{thebibliography}{87}
\providecommand{\natexlab}[1]{#1}
\providecommand{\url}[1]{\texttt{#1}}
\expandafter\ifx\csname urlstyle\endcsname\relax
  \providecommand{\doi}[1]{doi: #1}\else
  \providecommand{\doi}{doi: \begingroup \urlstyle{rm}\Url}\fi

\bibitem[Abadie et~al.(2020)Abadie, Athey, Imbens, and Wooldridge]{Abadie2020}
A.~Abadie, S.~Athey, G.~W. Imbens, and J.~M. Wooldridge.
\newblock Sampling-based versus design-based uncertainty in regression
  analysis.
\newblock \emph{Econometrica}, 88:\penalty0 265--296, 2020.

\bibitem[Angrist et~al.(2009)Angrist, Lang, and
  Oreopoulos]{angrist2009incentives}
J.~Angrist, D.~Lang, and P.~Oreopoulos.
\newblock Incentives and services for college achievement: Evidence from a
  randomized trial.
\newblock \emph{American Economic Journal: Applied Economics}, 1:\penalty0
  136--163, 2009.

\bibitem[Bai(2022)]{Bai2022}
Y.~Bai.
\newblock Optimality of matched-pair designs in randomized controlled trials.
\newblock \emph{American Economic Review}, 112:\penalty0 3911--40, 2022.

\bibitem[Bai et~al.(2022)Bai, Romano, and Shaikh]{Bai22InfPair}
Y.~Bai, J.~P. Romano, and A.~M. Shaikh.
\newblock Inference in experiments with matched pairs.
\newblock \emph{Journal of the American Statistical Association}, 117:\penalty0
  1726--1737, 2022.

\bibitem[Bai et~al.(2023)Bai, Liu, Shaikh, and
  Tabord-Meehan]{bai2023efficiency}
Y.~Bai, J.~Liu, A.~M. Shaikh, and M.~Tabord-Meehan.
\newblock On the efficiency of finely stratified experiments.
\newblock \emph{arXiv preprint arXiv:2307.15181}, 2023.

\bibitem[Bai et~al.(2024{\natexlab{a}})Bai, Jiang, Romano, Shaikh, and
  Zhang]{bai2023covariate}
Y.~Bai, L.~Jiang, J.~P. Romano, A.~M. Shaikh, and Y.~Zhang.
\newblock Covariate adjustment in experiments with matched pairs.
\newblock \emph{Journal of Econometrics}, accepted, 2024{\natexlab{a}}.

\bibitem[Bai et~al.(2024{\natexlab{b}})Bai, Shaikh, and
  Tabord-Meehan]{bai2024primer}
Y.~Bai, A.~M. Shaikh, and M.~Tabord-Meehan.
\newblock A primer on the analysis of randomized experiments and a survey of
  some recent advances.
\newblock \emph{arXiv preprint arXiv:2405.03910}, 2024{\natexlab{b}}.

\bibitem[Banerjee et~al.(2020)Banerjee, Chassang, Montero, and
  Snowberg]{Banerjee20}
A.~V. Banerjee, S.~Chassang, S.~Montero, and E.~Snowberg.
\newblock A theory of experimenters: Robustness, randomization, and balance.
\newblock \emph{American Economic Review}, 110:\penalty0 1206--30, 2020.

\bibitem[Beaman et~al.(2023)Beaman, Karlan, Thuysbaert, and Udry]{Beaman2023}
L.~Beaman, D.~Karlan, B.~Thuysbaert, and C.~Udry.
\newblock Selection into credit markets: Evidence from agriculture in mali.
\newblock \emph{Econometrica}, 91:\penalty0 1595--1627, 2023.

\bibitem[Beck et~al.(2016)Beck, Lu, Greevy, and Beck]{beck2016nbpmatching}
Cole Beck, Bo~Lu, Robert Greevy, and MC~Beck.
\newblock nbpmatching: Functions for optimal non-bipartite matching.
\newblock \emph{R package version}, 1\penalty0 (1), 2016.

\bibitem[Bojinov and Gupta(2022)]{Bojinov2022Online}
I.~Bojinov and S.~Gupta.
\newblock Online {Experimentation}: Benefits, {Operational} and
  {Methodological} {Challenges}, and {Scaling} {Guide}.
\newblock \emph{Harvard Data Science Review}, 4\penalty0 (3), jul 28 2022.
\newblock https://hdsr.mitpress.mit.edu/pub/aj31wj81.

\bibitem[Box et~al.(2005)Box, Hunter, and Hunter]{BHH2005}
G.~E.~P. Box, J.~S. Hunter, and W.~G. Hunter.
\newblock \emph{Statistics for experimenters: design, innovation, and
  discovery}, volume~2.
\newblock Wiley-Interscience New York, 2005.

\bibitem[Boyd and Díez-Amigo(2023)]{BOYD2023102462}
C.~M. Boyd and S.~Díez-Amigo.
\newblock Effectiveness of free financial education provided by for-profit
  financial institutions: Experimental evidence from rural peru.
\newblock \emph{Economics of Education Review}, 97:\penalty0 102462, 2023.

\bibitem[Brade(2023)]{Brade2023}
R.~Brade.
\newblock Social information and educational investment—nudging remedial math
  course participation.
\newblock \emph{Education Finance and Policy}, 19:\penalty0 106--142, 2023.

\bibitem[Branson and Shao(2021)]{BS2021}
Z.~Branson and S.~Shao.
\newblock Ridge rerandomization: An experimental design strategy in the
  presence of covariate collinearity.
\newblock \emph{Journal of Statistical Planning and Inference}, 211:\penalty0
  287--314, 2021.

\bibitem[Branson et~al.(2016)Branson, Dasgupta, and Rubin]{BDR2016}
Z.~Branson, T.~Dasgupta, and D.~B. Rubin.
\newblock {Improving covariate balance in 2K factorial designs via
  rerandomization with an application to a New York City Department of
  Education High School Study}.
\newblock \emph{The Annals of Applied Statistics}, 10:\penalty0 1958 -- 1976,
  2016.

\bibitem[Branson et~al.(2023)Branson, Li, and Ding]{branson2022power}
Z.~Branson, X.~Li, and P.~Ding.
\newblock Power and sample size calculations for rerandomization.
\newblock \emph{Biometrika}, 111:\penalty0 355--363, 2023.

\bibitem[Bruhn and McKenzie(2009)]{Bruhn:2009}
M.~Bruhn and D.~McKenzie.
\newblock In pursuit of balance: Randomization in practice in development field
  experiments.
\newblock \emph{American Economic Journal: Applied Economics}, 1:\penalty0
  200--232, 2009.

\bibitem[Brune et~al.(2021)Brune, Chyn, and Kerwin]{Brune2021}
L.~Brune, E.~Chyn, and J.~Kerwin.
\newblock Pay me later: Savings constraints and the demand for deferred
  payments.
\newblock \emph{American Economic Review}, 111:\penalty0 2179–2212, 2021.

\bibitem[Casella and Berger(2002)]{CB2002}
G.~Casella and R.~L. Berger.
\newblock \emph{Statistical Inference}.
\newblock Pacific Grove: Duxbury, 2002.

\bibitem[Caughey et~al.(2023)Caughey, Dafoe, Li, and Miratrix]{CDLM21quantile}
D.~Caughey, A.~Dafoe, X.~Li, and L.~Miratrix.
\newblock Randomization inference beyond the sharp null: Bounded null
  hypotheses and quantiles of individual treatment effects.
\newblock \emph{Journal of the Royal Statistical Society, Series B (Statistical
  Methodology)}, in press, 2023.

\bibitem[Chattopadhyay et~al.(2022)Chattopadhyay, Morris, and
  Zubizarreta]{chattopadhyay2022balanced}
A.~Chattopadhyay, C.~N. Morris, and J.~R. Zubizarreta.
\newblock Balanced and robust randomized treatment assignments: The finite
  selection model for the health insurance experiment and beyond.
\newblock \emph{arXiv preprint arXiv:2205.09736}, 2022.

\bibitem[Cohen and Fogarty(2022)]{CF22pivot}
P.~L. Cohen and C.~B. Fogarty.
\newblock Gaussian prepivoting for finite population causal inference.
\newblock \emph{Journal of the Royal Statistical Society: Series B (Statistical
  Methodology)}, 84:\penalty0 295--320, 2022.

\bibitem[Cox(1982)]{cox:1982}
D.~R. Cox.
\newblock Randomization and concomitant variables in the design of experiments.
\newblock In P.~R.~Krishnaiah G.~Kallianpur and J.~K. Ghosh, editors,
  \emph{Statistics and Probability: Essays in Honor of C. R. Rao}, pages
  197--202. North-Holland, Amsterdam, 1982.

\bibitem[Cox(2007)]{Cox2007}
D.~R. Cox.
\newblock {Applied statistics: A review}.
\newblock \emph{The Annals of Applied Statistics}, 1:\penalty0 1 -- 16, 2007.

\bibitem[Cronin and Lieber(2024)]{CRONIN2024102877}
C.~J. Cronin and E.~M.J. Lieber.
\newblock The demand for skills training among medicaid home-based caregivers.
\newblock \emph{Journal of Health Economics}, 95:\penalty0 102877, 2024.

\bibitem[Cytrynbaum(2021)]{cytrynbaum2021optimal}
M.~Cytrynbaum.
\newblock Optimal stratification of survey experiments.
\newblock \emph{arXiv preprint arXiv:2111.08157}, 2021.

\bibitem[Cytrynbaum(2024{\natexlab{a}})]{cytrynbaum2024covariate}
M.~Cytrynbaum.
\newblock Covariate adjustment in stratified experiments.
\newblock \emph{Quantitative Economics}, 15\penalty0 (4):\penalty0 971--998,
  2024{\natexlab{a}}.

\bibitem[Cytrynbaum(2024{\natexlab{b}})]{cytrynbaum2024finely}
M.~Cytrynbaum.
\newblock Finely stratified rerandomization designs.
\newblock \emph{arXiv preprint arXiv:2407.03279}, 2024{\natexlab{b}}.

\bibitem[de~Mel et~al.(2019)de~Mel, McKenzie, and Woodruff]{McKenzie2019}
S.~de~Mel, D.~McKenzie, and C.~Woodruff.
\newblock Labor drops: Experimental evidence on the return to additional labor
  in microenterprises.
\newblock \emph{American Economic Journal: Applied Economics}, 11:\penalty0
  202–35, 2019.

\bibitem[Deaton and Cartwright(2018)]{deaton2018understanding}
A.~Deaton and N.~Cartwright.
\newblock Understanding and misunderstanding randomized controlled trials.
\newblock \emph{Social science \& medicine}, 210:\penalty0 2--21, 2018.

\bibitem[Dharmadhikari and Joag-Dev(1988)]{DJ88}
S.~Dharmadhikari and K.~Joag-Dev.
\newblock \emph{Unimodality, Convexity, and Applications}.
\newblock San Diego, CA: Academic Press, Inc., 1988.

\bibitem[Ding(2023)]{ding2023first}
P~Ding.
\newblock A first course in causal inference.
\newblock \emph{arXiv preprint arXiv:2305.18793}, 2023.

\bibitem[Ernst(2004)]{Ernst2004}
M.~D. Ernst.
\newblock {Permutation Methods: A Basis for Exact Inference}.
\newblock \emph{Statistical Science}, 19:\penalty0 676 -- 685, 2004.

\bibitem[Fisher(1925)]{fisher1925statistical}
R.~A. Fisher.
\newblock \emph{Statistical Methods for Research Workers}.
\newblock Edinburgh by Oliver and Boyd, 1st edition, 1925.

\bibitem[Fisher(1926)]{fisher1926}
R.~A. Fisher.
\newblock The arrangement of field experiments.
\newblock \emph{Journal of the Ministry of Agriculture of Great Britain},
  33:\penalty0 503--513, 1926.

\bibitem[Fisher(1935)]{Fisher:1935}
R.~A. Fisher.
\newblock \emph{The {D}esign of {E}xperiments, 1st Edition}.
\newblock Edinburgh, London: Oliver and Boyd, 1935.

\bibitem[Fogarty(2018)]{F18pairadj}
C.~B. Fogarty.
\newblock {Regression-assisted inference for the average treatment effect in
  paired experiments}.
\newblock \emph{Biometrika}, 105:\penalty0 994--1000, 2018.

\bibitem[Freedman(2008)]{Freedman2008}
D.~A. Freedman.
\newblock On regression adjustments to experimental data.
\newblock \emph{Advances in Applied Mathematics}, 40:\penalty0 180--193, 2008.

\bibitem[Gerber and Green(2012)]{Greenbook}
A.S. Gerber and D.P. Green.
\newblock \emph{Field Experiments: Design, Analysis, and Interpretation}.
\newblock W. W. Norton and Company, New York, 2012.

\bibitem[Greevy et~al.(2004)Greevy, Lu, Silber, and Rosenbaum]{GLSR2004}
R.~Greevy, B.~Lu, J.~H. Silber, and P.~Rosenbaum.
\newblock {Optimal multivariate matching before randomization}.
\newblock \emph{Biostatistics}, 5:\penalty0 263--275, 2004.

\bibitem[Grenander(1981)]{grenander1981abstract}
Ulf Grenander.
\newblock \emph{Abstract inference}.
\newblock Wiley, New York, 1981.

\bibitem[Grimm et~al.(2016)Grimm, Munyehirwe, Peters, and Sievert]{Grimm2016}
M.~Grimm, A.~Munyehirwe, J.~Peters, and M.~Sievert.
\newblock A first step up the energy ladder? low cost solar kits and
  household’s welfare in rural rwanda.
\newblock \emph{The World Bank Economic Review}, 31:\penalty0 631--649, 2016.

\bibitem[H{\'a}jek(1960)]{hajek1960limiting}
J.~H{\'a}jek.
\newblock Limiting distributions in simple random sampling from a finite
  population.
\newblock \emph{Publications of the Mathematics Institute of the Hungarian
  Academy of Science}, 5:\penalty0 361--74, 1960.

\bibitem[Hall(2007)]{Hall2007}
N.~S. Hall.
\newblock R. a. fisher and his advocacy of randomization.
\newblock \emph{Journal of the History of Biology}, 40:\penalty0 295--325,
  2007.

\bibitem[Harshaw et~al.(2024)Harshaw, S{\"a}vje, Spielman, and
  Zhang]{Harshaw24}
C.~Harshaw, F.~S{\"a}vje, D.~A. Spielman, and P.~Zhang.
\newblock Balancing covariates in randomized experiments with the gram--schmidt
  walk design.
\newblock \emph{Journal of the American Statistical Association}, in
  press:\penalty0 1--13, 2024.

\bibitem[Hemerik and Goeman(2021)]{Goeman2021}
J.~Hemerik and J.~J. Goeman.
\newblock Another look at the lady tasting tea and differences between
  permutation tests and randomisation tests.
\newblock \emph{International Statistical Review}, 89:\penalty0 367--381, 2021.

\bibitem[Huber(1964)]{Huber1964}
P.~J. Huber.
\newblock {Robust Estimation of a Location Parameter}.
\newblock \emph{The Annals of Mathematical Statistics}, 35:\penalty0 73 -- 101,
  1964.

\bibitem[Imai(2008)]{Imai08}
K.~Imai.
\newblock Variance identification and efficiency analysis in randomized
  experiments under the matched-pair design.
\newblock \emph{Statistics in Medicine}, 27:\penalty0 4857--4873, 2008.

\bibitem[Johansson and Schultzberg(2022)]{JS2022}
P.~Johansson and M.~Schultzberg.
\newblock Rerandomization: A complement or substitute for stratification in
  randomized experiments?
\newblock \emph{Journal of Statistical Planning and Inference}, 218:\penalty0
  43--58, 2022.

\bibitem[Kasy(2016)]{Kasy2016}
M.~Kasy.
\newblock Why experimenters might not always want to randomize, and what they
  could do instead.
\newblock \emph{Political Analysis}, 24:\penalty0 324–338, 2016.

\bibitem[Krieger et~al.(2023)Krieger, Azriel, and Kapelner]{Krieger2023}
A.~M. Krieger, D.~A. Azriel, and A.~Kapelner.
\newblock Better experimental design by hybridizing binary matching with
  imbalance optimization.
\newblock \emph{Canadian Journal of Statistics}, 51:\penalty0 275--292, 2023.

\bibitem[Lee et~al.(2021)Lee, Morduch, Ravindran, Shonchoy, and Zaman]{Lee2021}
J.~N. Lee, J.~Morduch, S.~Ravindran, A.~Shonchoy, and H.~Zaman.
\newblock Poverty and migration in the digital age: Experimental evidence on
  mobile banking in bangladesh.
\newblock \emph{American Economic Journal: Applied Economics}, 13:\penalty0
  38--71, 2021.

\bibitem[Lee et~al.(2022)Lee, Morduch, Ravindran, and Shonchoy]{LEE2022276}
J.~N. Lee, J.~Morduch, S.~Ravindran, and A.~S. Shonchoy.
\newblock Narrowing the gender gap in mobile banking.
\newblock \emph{Journal of Economic Behavior \& Organization}, 193:\penalty0
  276--293, 2022.

\bibitem[Lei and Ding(2020)]{LeiD2020highdim}
L.~Lei and P.~Ding.
\newblock {Regression adjustment in completely randomized experiments with a
  diverging number of covariates}.
\newblock \emph{Biometrika}, 108:\penalty0 815--828, 2020.

\bibitem[Li and Ding(2017)]{fcltxlpd2016}
X.~Li and P.~Ding.
\newblock General forms of finite population central limit theorems with
  applications to causal inference.
\newblock \emph{Journal of the American Statistical Association}, 112:\penalty0
  1759--1769, 2017.

\bibitem[Li and Ding(2020)]{LD20reg}
X.~Li and P.~Ding.
\newblock Rerandomization and regression adjustment.
\newblock \emph{Journal of the Royal Statistical Society: Series B (Statistical
  Methodology)}, 82:\penalty0 241--268, 2020.

\bibitem[Li et~al.(2018)Li, Ding, and Rubin]{LDR18}
X.~Li, P.~Ding, and D.~B. Rubin.
\newblock Asymptotic theory of rerandomization in treatment{\textendash}control
  experiments.
\newblock \emph{Proceedings of the National Academy of Sciences}, 115:\penalty0
  9157--9162, 2018.

\bibitem[Li et~al.(2020)Li, Ding, and Rubin]{LDR20factorial}
X.~Li, P.~Ding, and D.~B. Rubin.
\newblock {Rerandomization in $2^{K}$ factorial experiments}.
\newblock \emph{The Annals of Statistics}, 48:\penalty0 43 -- 63, 2020.

\bibitem[Lin(2013)]{lin2013}
W.~Lin.
\newblock {Agnostic notes on regression adjustments to experimental data:
  Reexamining Freedman's critique}.
\newblock \emph{The Annals of Applied Statistics}, 7:\penalty0 295--318, 2013.

\bibitem[Liu et~al.(2025)Liu, Han, Rubin, and Deng]{liu2023bayesian}
Z.~Liu, T.~Han, D.~B. Rubin, and K.~Deng.
\newblock {A Bayesian Criterion for Rerandomization}.
\newblock \emph{Journal of the American Statistical Association}, in press,
  2025.

\bibitem[Lowe(2021)]{Lowe2021}
M.~Lowe.
\newblock Types of contact: A field experiment on collaborative and adversarial
  caste integration.
\newblock \emph{American Economic Review}, 111:\penalty0 1807–44, 2021.

\bibitem[Lu et~al.(2022)Lu, Liu, Liu, and Ding]{lu2022design}
X.~Lu, T.~Liu, H.~Liu, and P.~Ding.
\newblock Design-based theory for cluster rerandomization.
\newblock \emph{Biometrika}, 110:\penalty0 467--483, 2022.

\bibitem[MacKinnon(2013)]{mackinnon2013thirty}
J.~G. MacKinnon.
\newblock Thirty years of heteroskedasticity-robust inference.
\newblock In \emph{Recent advances and future directions in causality,
  prediction, and specification analysis: Essays in honor of Halbert L. White
  Jr}, pages 437--461. Springer, Berlin, 2013.

\bibitem[Morgan and Rubin(2012)]{morgan2012rerandomization}
K.~L. Morgan and D.~B. Rubin.
\newblock Rerandomization to improve covariate balance in experiments.
\newblock \emph{The Annals of Statistics}, 40:\penalty0 1263--1282, 2012.

\bibitem[Morris(1979)]{morris1979finite}
C.~Morris.
\newblock A finite selection model for experimental design of the health
  insurance study.
\newblock \emph{Journal of Econometrics}, 11:\penalty0 43--61, 1979.

\bibitem[Nadarajah(2008)]{Nadarajah08}
S.~Nadarajah.
\newblock Explicit expressions for moments of $\chi^2$ order statistics.
\newblock \emph{Bulletin of the Institute of Mathematics Academia Sinica (New
  Series)}, 3:\penalty0 433--444, 2008.

\bibitem[Neyman(1923)]{Neyman:1923}
J.~Neyman.
\newblock {On the application of probability theory to agricultural
  experiments. Essay on principles (with discussion). Section 9 (translated).
  reprinted ed.}
\newblock \emph{Statistical Science}, 5:\penalty0 465--472, 1923.

\bibitem[Rai{\v{c}}(2015)]{R15}
M.~Rai{\v{c}}.
\newblock Multivariate normal approximation: permutation statistics, local
  dependence and beyond.
\newblock 2015.

\bibitem[Resnjanskij et~al.(2024)Resnjanskij, Ruhose, Wiederhold, Woessmann,
  and Wedel]{Resnjanskij2024}
S.~Resnjanskij, J.~Ruhose, S.~Wiederhold, L.~Woessmann, and K.~Wedel.
\newblock Can mentoring alleviate family disadvantage in adolescence? a field
  experiment to improve labor market prospects.
\newblock \emph{Journal of Political Economy}, 132:\penalty0 1013--1062, 2024.

\bibitem[Rosenbaum(2010)]{Rosenbaum:2010}
P.~R. Rosenbaum.
\newblock \emph{Design of Observational Studies.}
\newblock New York: Springer, 2010.

\bibitem[Rosenberger and Lachin(2015)]{rosenberger2015randomization}
W.~F. Rosenberger and J.~M. Lachin.
\newblock \emph{Randomization in Clinical Trials: Theory and Practice}.
\newblock John Wiley \& Sons, 2015.

\bibitem[Rubin(1974)]{Rubin:1974}
D.~B. Rubin.
\newblock Estimating causal effects of treatments in randomized and
  nonrandomized studies.
\newblock \emph{Journal of Educational Psychology}, 66:\penalty0 688--701,
  1974.

\bibitem[Rubin(1978)]{Rubin:1978}
D.~B. Rubin.
\newblock Bayesian inference for causal effects: the role of randomization.
\newblock \emph{The Annals of Statistics}, 6:\penalty0 34--58, 1978.

\bibitem[Savage(1962)]{savage:1962}
L.~J. Savage.
\newblock \emph{The Foundations of Statistical Inference}.
\newblock {Methuen and Co. Led., London}, 1962.

\bibitem[Shi and Ding(2022)]{shi2022berry}
L.~Shi and P.~Ding.
\newblock Berry--esseen bounds for design-based causal inference with possibly
  diverging treatment levels and varying group sizes.
\newblock \emph{arXiv preprint arXiv:2209.12345}, 2022.

\bibitem[Shi and Li(2024)]{ShiLi2024}
L.~Shi and X.~Li.
\newblock Some theoretical foundations for the design and analysis of
  randomized experiments.
\newblock \emph{Journal of Causal Inference}, 12:\penalty0 20230067, 2024.

\bibitem[Student(1938)]{student:1938}
Student.
\newblock Comparison between balanced and random arrangements of field plots.
\newblock \emph{Biometrika}, 29:\penalty0 363--378, 1938.

\bibitem[Wang et~al.(2023)Wang, Wang, and Liu]{WWL2021}
X.~Wang, T.~Wang, and H.~Liu.
\newblock Rerandomization in stratified randomized experiments.
\newblock \emph{Journal of the American Statistical Association}, 118:\penalty0
  1295--1304, 2023.

\bibitem[Wang and Li(2022)]{wang2022rerandomization}
Y.~Wang and X.~Li.
\newblock {Rerandomization with diminishing covariate imbalance and diverging
  number of covariates}.
\newblock \emph{The Annals of Statistics}, 50:\penalty0 3439 -- 3465, 2022.

\bibitem[White et~al.(2020)White, Lowenstein, Srivirojana, Jampaklay, and
  Dow]{Whitem3797}
J.~S. White, C.~Lowenstein, N.~Srivirojana, A.~Jampaklay, and W.~H. Dow.
\newblock Incentive programmes for smoking cessation: cluster randomized trial
  in workplaces in thailand.
\newblock \emph{BMJ}, 371, 2020.

\bibitem[Wintner(1936)]{wintner1936class}
A.~Wintner.
\newblock {On a class of Fourier transforms}.
\newblock \emph{American Journal of Mathematics}, 58:\penalty0 45--90, 1936.

\bibitem[Yang et~al.(2024)Yang, Yu, Liao, Qiao, Fan, Li, Hu, Chen, Ye, Cai, Ma,
  Pang, Huang, Jia, Reinhardt, and Dou]{Yang2024}
S.~Yang, B.~Yu, K.~Liao, X.~Qiao, Y.~Fan, M.~Li, Y.~Hu, J.~Chen, T.~Ye, C.~Cai,
  C.~Ma, T.~Pang, Z.~Huang, P.~Jia, J.~D. Reinhardt, and Q.~Dou.
\newblock Effectiveness of a socioecological model-guided, smart device-based,
  self-management-oriented lifestyle intervention in community residents:
  protocol for a cluster-randomized controlled trial.
\newblock \emph{BMC Public Health}, 24:\penalty0 32, 2024.

\bibitem[Yang et~al.(2023)Yang, Qu, and Li]{YQL2021}
Z.~Yang, T.~Qu, and X.~Li.
\newblock Rejective sampling, rerandomization, and regression adjustment in
  survey experiments.
\newblock \emph{Journal of the American Statistical Association}, 118:\penalty0
  1207--1221, 2023.

\bibitem[Zhang et~al.(2024)Zhang, Yin, and Rubin]{ZYR2021}
H.~Zhang, G.~Yin, and D.~B. Rubin.
\newblock Pca rerandomization.
\newblock \emph{Canadian Journal of Statistics}, 52:\penalty0 5--25, 2024.

\bibitem[Zhao and Ding(2021)]{ZHAO2021278}
A.~Zhao and P.~Ding.
\newblock Covariate-adjusted fisher randomization tests for the average
  treatment effect.
\newblock \emph{Journal of Econometrics}, 225:\penalty0 278--294, 2021.

\bibitem[Zhao and Ding(2024)]{ZD2021rerand}
A.~Zhao and P.~Ding.
\newblock No star is good news: A unified look at rerandomization based on
  p-values from covariate balance tests.
\newblock \emph{Journal of Econometrics}, 241:\penalty0 105724, 2024.

\end{thebibliography}

\newpage

\begin{center}
	\bf \LARGE 
	Supplementary Material 
\end{center}

\setcounter{equation}{0}
\setcounter{section}{0}
\setcounter{figure}{0}
\setcounter{example}{0}
\setcounter{proposition}{0}
\setcounter{corollary}{0}
\setcounter{theorem}{0}
\setcounter{lemma}{0}
\setcounter{table}{0}
\setcounter{condition}{0}

\renewcommand {\theproposition} {A\arabic{proposition}}
\renewcommand {\theexample} {A\arabic{example}}
\renewcommand {\thefigure} {A\arabic{figure}}
\renewcommand {\thetable} {A\arabic{table}}
\renewcommand {\theequation} {A\arabic{equation}}
\renewcommand {\thelemma} {A\arabic{lemma}}
\renewcommand {\thesection} {A\arabic{section}}
\renewcommand {\thetheorem} {A\arabic{theorem}}
\renewcommand {\thecorollary} {A\arabic{corollary}}
\renewcommand {\thecondition} {A\arabic{condition}}
\renewcommand {\theremark} {A\arabic{remark}}

\renewcommand {\thepage} {A\arabic{page}}
\setcounter{page}{1}

\section{Practical guidance for designing the best-choice rerandomization}\label{sec:practical}

In this section, we present the details for Section \ref{sec:prac} regarding the practical guidance for designing the best-choice rerandomization with a given finite set of experimental units. 
In general, the asymptotic gain from the best-choice rerandomization increases with $T$, but the additional gain from increasing $T$ typically decreases with $T$; see, e.g., Figure \ref{fig:vKT} below. 
Thus, we suggest using large but not overly large $T$, say, $T=1000$. 
In addition, we should choose moderate number of covariates $K$, trying to make their association with potential outcomes, measured by $R^2$, large. 
We should avoid excessively large value of $K$ that provides little increment on $R^2$ but substantially increase the variance $v_{K,T}$ of the constrained Gaussian random variable, which will ultimately diminish the gain from rerandomization, as implied by Corollary \ref{cor:reduce_var}. 
Below we provide some useful practical guidance when designing a best-choice rerandomization.

First, we can check the efficiency improvement from a specific choice of $T$ for the best-choice rerandomization. 
From Corollary \ref{cor:reduce_var} and Theorem \ref{thm:opt_bcr}, 
the difference between a best-choice rerandomization with a particular choice of $T$ and the optimal one is 
$R^2 - (1-v_{K,T}) R^2 = v_{K,T} R^2 \le v_{K,T}$. 
Thus, the variance of the constrained Gaussian random variable, $v_{K,T}$, actually gives an upper bound on the gap between a particular rerandomization design and the optimal one. 
In practice, we can choose $T$ such that this gap is reasonably small, say, below $0.05$ or $0.1$. 
Figure \ref{fig:vKT} shows the value of $v_{K,T}$ when $K$ ranges from $1$ to $40$ and $T$ ranges from $10$ to $10^4$. 
Obviously, the value of $v_{K,T}$ increases with $K$ and decreases with $T$. 
Below we investigate the choice of $T$ such that $v_{K,T}$ is bounded by $0.1$. 
When $K=1$, $v_{K, T}$ is about $2.5\%$ when $T=10$; 
when $K=5$, $v_{K, T}$ is about $0.1$ when $T=100$; 
when $K=10$, $v_{K,T}$ is about $0.1$ when $T$ is about $3000$; when $K\ge 20$, we need $T$ to be greater than $10^4$ in order to make $v_{K,T}$ bounded by $0.1$. 
These indicate that in practice we should not include too many covariates into rerandomization, since they would require much larger $T$ in order to make the best-choice rerandomization close to the optimal one; 
intuitively, this will also lead to higher computation cost 
and 
less accurate asymptotic approximation.  

\begin{figure}[htbp]
    \centering
    \includegraphics[width=0.3\textwidth]{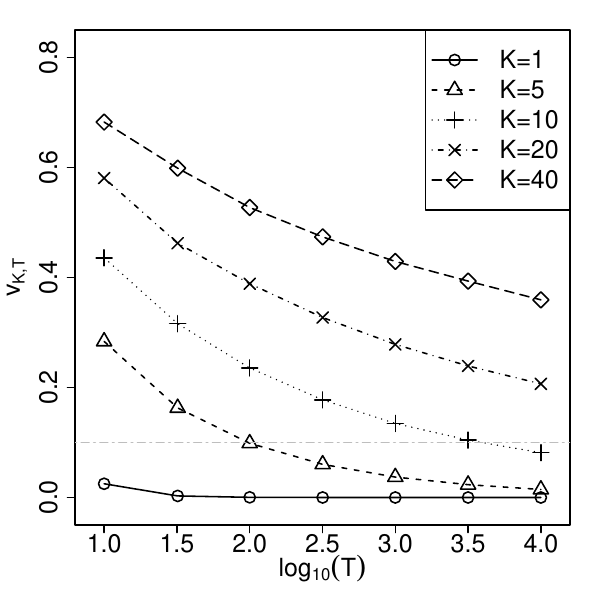}
    \caption{The variance $v_{K,T}$ of the constrained Gaussian random variable for various $(K,T)$.}
    \label{fig:vKT}
\end{figure}

Second, we can evaluate the trade-off between the potential gain and the worst-case loss from the best-choice rerandomization, as suggested in \citet{wang2022rerandomization}. 
As demonstrated before, 
the best-choice rerandomization can asymptotically improve the precision of the difference-in-means estimator compared to the CRE. 
With finite sample size, 
the CRE is actually the minimax optimal in terms of the mean squared error of the difference-in-means estimator, when we considering the worst-case scenario over all possible configurations of the potential outcomes; see the discussion in Section \ref{sec:prac} and \citet[][Proposition A1]{wang2022rerandomization}. 
Note that this does not contradict with our asymptotic theory, since, in the worst case with $R^2=0$, the asymptotic distribution of the difference-in-means estimator under the best-choice rerandomization is the same as that under the CRE. 
These two observations then provide a quantitative way to characterize the trade-off when designing rerandomization. 
Intuitively, 
we can use Corollary \ref{cor:reduce_var} to characterize the potential gain, and the worst-case mean squared error, which can be estimated by the Monte Carlo method, to characterize the worst-case loss. 
We can then consider these trade-offs when comparing multiple rerandomization designs. 
We will discuss this in more detail in Section \ref{sec:star}.

\section{Simulation using the student achievement and retention project}\label{sec:star}

In this section, we provide the details for Section \ref{sec:star_main}. 
The Student Achievement and Retention (STAR) project aims to evaluate the effect of academic services and incentives for college students. 
We focus on the comparison between two treatment arms, 
where the treated students would receive an array of  support services and cash awards for meeting a target GPA and the control students received only standard support services.

We preprocess the data in the same way as \citet{wang2022rerandomization}.  
We remove units with missing covariates, resulting in $n_1 = 118$ treated units and $n_0 = 856$ control units, 
and generate $200$ pretreatment covariates for these $n=974$ units. 
Specifically, the first $5$ covariates are from the STAR project, including high-school GPA, age, gender and indicators for whether lives at home and whether rarely puts off studying for tests, 
and the remaining $195$ covariates are generated independently from the $t$ distribution with degrees of freedom $2$. 
We use the first year GPA from the STAR project as the outcome and let both treatment and control potential outcomes be the same as the observed ones, so that we are able to evaluate the repeated-sampling properties of 
various designs. 
Once generated, the potential outcomes and covariates will be kept fixed during the simulation, mimicking the finite population inference. 
Table \ref{tab:star_simu} shows the empirical bias, root mean squared error and coverage probabilities based on $10^5$ draws from the best-choice rerandomization using the first $K$ covariates and $T$ number of tried complete randomizations, for various choices of $(K,T)$. 
Surprisingly, we find that the best-choice rerandomization performs relatively well even when $K$ and $T$ are large, although the confidence intervals become slightly under-covered when $K$ is large. 
These show the robustness of the best-choice rerandomization design in the presence of high-dimensional covariates and large number of tried complete randomizations. 
Note that these do not contradict with the intuition from our theory, which implies the potential danger of large $K$ and $T$. 
As shown below, for some potential outcome configurations, the performance of rerandomization can deteriorate significantly as $K$ and $T$ increase.

\begin{table}[]
    \centering
    \caption{
    Simulations under the best-choice rerandomization with various choices of $(K,T)$. 
    The first two columns show the values of $K$ and $T$. 
    The 2nd-7th columns show the bias, root mean squared error, and coverage probabilities of confidence intervals using various finite-sample adjustments as discussed in \citet[][Section 6]{wang2022rerandomization} (see also \citet{LeiD2020highdim} and \citet{mackinnon2013thirty}), 
    when the potential outcomes are imputed based on the observed ones. 
    The 8th-13th columns show analogous quantities when the potential outcomes are quantile transformations of the average propensity scores across all the considered rerandomization designs. 
    }
    \label{tab:star_simu}
    \resizebox{0.95\columnwidth}{!}{%
    \begin{tabular}{ccrcccccrccccc}
\toprule
$K$ & $T$ & Bias & RMSE & HC0 & HC1 & HC2 & HC3 & Bias & RMSE & HC0 & HC1 & HC2 & HC3
\\
\midrule
$5$&$10$&$0.012$&$0.889$&$0.950$&$0.955$&$0.955$&$0.957$&$-0.022$&$0.965$&$0.958$&$0.958$&$0.958$&$0.960$\\
$5$&$100$&$0.003$&$0.910$&$0.939$&$0.943$&$0.944$&$0.952$&$0.033$&$0.977$&$0.950$&$0.952$&$0.952$&$0.952$\\
$5$&$1000$&$-0.023$&$0.874$&$0.950$&$0.955$&$0.955$&$0.961$&$0.027$&$1.022$&$0.941$&$0.942$&$0.942$&$0.942$\\
$5$&$10000$&$0.022$&$0.918$&$0.937$&$0.943$&$0.944$&$0.948$&$0.000$&$0.966$&$0.960$&$0.963$&$0.963$&$0.964$\\
\midrule
$10$&$10$&$0.026$&$0.977$&$0.929$&$0.934$&$0.934$&$0.941$&$0.163$&$1.037$&$0.925$&$0.931$&$0.931$&$0.933$\\
$10$&$100$&$0.007$&$0.891$&$0.946$&$0.950$&$0.950$&$0.955$&$0.240$&$1.021$&$0.925$&$0.940$&$0.940$&$0.946$\\
$10$&$1000$&$-0.004$&$0.890$&$0.938$&$0.947$&$0.947$&$0.955$&$0.308$&$0.991$&$0.931$&$0.941$&$0.941$&$0.944$\\
$10$&$10000$&$-0.051$&$0.907$&$0.931$&$0.943$&$0.943$&$0.952$&$0.289$&$1.035$&$0.917$&$0.930$&$0.932$&$0.936$\\
\midrule
$50$&$10$&$0.054$&$1.008$&$0.932$&$0.945$&$0.946$&$0.952$&$0.440$&$1.063$&$0.893$&$0.909$&$0.912$&$0.914$\\
$50$&$100$&$0.006$&$0.983$&$0.921$&$0.940$&$0.941$&$0.947$&$0.704$&$1.183$&$0.837$&$0.869$&$0.870$&$0.874$\\
$50$&$1000$&$-0.039$&$0.890$&$0.943$&$0.961$&$0.962$&$0.969$&$0.911$&$1.334$&$0.762$&$0.812$&$0.815$&$0.822$\\
$50$&$10000$&$0.008$&$0.894$&$0.936$&$0.963$&$0.964$&$0.970$&$1.104$&$1.451$&$0.685$&$0.752$&$0.757$&$0.770$\\
\midrule
$100$&$10$&$0.014$&$0.953$&$0.921$&$0.946$&$0.946$&$0.950$&$0.535$&$1.072$&$0.870$&$0.898$&$0.897$&$0.908$\\
$100$&$100$&$0.063$&$0.944$&$0.918$&$0.951$&$0.951$&$0.954$&$0.978$&$1.342$&$0.722$&$0.784$&$0.789$&$0.806$\\
$100$&$1000$&$-0.011$&$0.965$&$0.892$&$0.942$&$0.944$&$0.952$&$1.126$&$1.488$&$0.614$&$0.704$&$0.710$&$0.728$\\
$100$&$10000$&$0.060$&$0.948$&$0.894$&$0.952$&$0.952$&$0.962$&$1.361$&$1.632$&$0.530$&$0.619$&$0.626$&$0.652$\\
\midrule
$200$&$10$&$0.048$&$0.960$&$0.901$&$0.902$&$0.902$&$0.904$&$0.630$&$1.184$&$0.782$&$0.783$&$0.783$&$0.803$\\
$200$&$100$&$0.021$&$0.973$&$0.889$&$0.889$&$0.889$&$0.896$&$0.937$&$1.320$&$0.714$&$0.720$&$0.721$&$0.749$\\
$200$&$1000$&$0.022$&$0.936$&$0.884$&$0.884$&$0.884$&$0.890$&$1.188$&$1.515$&$0.608$&$0.615$&$0.621$&$0.655$\\
$200$&$10000$&$-0.021$&$0.978$&$0.871$&$0.877$&$0.877$&$0.888$&$1.384$&$1.677$&$0.489$&$0.496$&$0.502$&$0.541$\\
\bottomrule
    \end{tabular}%
    }
\end{table}

We now consider another way of generating potential outcomes. 
We first use Monte Carlo method to  estimate the propensity scores of each unit averaging over the best-choice rerandomization designs under investigation, 
and then take a Gaussian quantile transformation to generate both potential outcomes, where the treatment and control potential outcomes for each unit are the same.
With this potential outcome configuration, 
Table \ref{tab:star_simu} shows analogously the empirical bias, root mean squared error and coverage probabilities under the best-choice rerandomization designs with various choices of $(K,T)$. 
Obviously, as $K$ and $T$ increases, 
both the bias and mean squared error tend to be larger, 
and the coverage probabilities become much smaller than the nominal level. 
Thus, 
the performance of the best-choice rerandomization can be sensitive to the potential outcome configuration when $K$ and $T$ are large. 
Below we present two practical ways to tackle this issue.

First, we can check the worst-case behavior of the best-choice rerandomization using the formula derived in \citet[][Proposition A1]{wang2022rerandomization}. 
As shown in Table \ref{tab:worst_case}, 
both 
the worst-case 
bias and 
root mean squared error increase notably as $K$ and $T$ increases. 
It is worth mentioning that there is considerable Monte Carlo error even when we use $10^5$ simulated assignments to estimate the worst-case root mean squared error; 
for example, the estimated worst-case bias and root mean squared error for the CRE is about $0.1$ and $1.1$, whose true values are known to be $0$ and $1$. 
We can increase the number of simulated assignments to make the Monte Carlo estimation more accurate, but the trend in Table \ref{tab:worst_case} already illustrates the potential drawback of large $K$ and $T$.
As discussed in Section \ref{sec:practical}, 
in practice, we can also combine this with the potential gain that rerandomization can bring as shown in Corollary \ref{cor:reduce_var} to guide our design of rerandomization. 
For example, similar to \citet{wang2022rerandomization}, we can use the geometric mean of the worst-case mean squared error and the ideal-case mean squared error implied by the asymptotic theory as a measure for comparing different best-choice rerandomization designs. 
Note that the asymptotic mean squared error (or equivalently the asymptotic variance since $\hat{\tau}$ is asymptotically unbiased under rerandomization) depends on the association between potential outcomes and covariates, measured by $R^2$ in \eqref{eq:R2}, which needs to be determined by domain knowledge or prior studies. 

\begin{table}[]
    \centering
    \small
    \caption{Worst-case 
    bias (top half) and root mean squared error (bottom half)
    of the best-choice rerandomization with various choices of $(K,T)$.
    The $2$nd-$5$th columns use the original covariates for rerandomization,
    and the $6$th-$9$th columns use the trimmed covariates. 
    }
    \label{tab:worst_case}
    \begin{tabular}{cccccccccccc}
    \toprule
    & & \multicolumn{4}{c}{Original covariates} & & \multicolumn{4}{c}{Trimmed covariates}
    \\
    \diagbox[height=0.7cm]{$K$}{$T$} &&$10$&$100$&$1000$&$10000$&&$10$&$100$&$1000$&$10000$\\
     \midrule
    $5$&&$0.117$&$0.126$&$0.134$&$0.131$&&$0.109$&$0.113$&$0.115$&$0.112$\\
    $10$&&$0.512$&$0.682$&$0.793$&$0.850$&&$0.137$&$0.162$&$0.166$&$0.174$\\
    $50$&&$0.823$&$1.204$&$1.444$&$1.620$&&$0.160$&$0.214$&$0.254$&$0.277$\\
    $100$&&$0.853$&$1.305$&$1.600$&$1.822$&&$0.161$&$0.213$&$0.268$&$0.309$\\
    $200$&&$0.798$&$1.244$&$1.566$&$1.819$&&$0.167$&$0.236$&$0.288$&$0.331$\\
    \midrule
    $5$&&$1.099$&$1.101$&$1.099$&$1.100$&&$1.100$&$1.100$&$1.102$&$1.099$\\
    $10$&&$1.102$&$1.140$&$1.193$&$1.227$&&$1.100$&$1.102$&$1.102$&$1.103$\\
    $50$&&$1.157$&$1.411$&$1.615$&$1.773$&&$1.104$&$1.108$&$1.112$&$1.115$\\
    $100$&&$1.186$&$1.504$&$1.754$&$1.955$&&$1.107$&$1.113$&$1.119$&$1.122$\\
    $200$&&$1.185$&$1.483$&$1.753$&$1.974$&&$1.109$&$1.121$&$1.128$&$1.133$\\
    \bottomrule
    \end{tabular}
\end{table}

Second, we can perform trimming to effectively improve the worst-case performance of the best-choice rerandomization, a strategy that has also been used by \citet{LeiD2020highdim} and \citet{wang2022rerandomization}. 
Note that under our randomization-based inference without any model assumptions, we have the flexibility to conduct arbitrary pre-processing of the covariates.
The intuition of trimming is similar to that discussed in \citet{morgan2012rerandomization} with small sample size. 
When a unit has extreme covariates, it is more likely to be allocated to the group of larger size under rerandomization, which can help balance the covariates between the two groups. 
The extremeness of covariates also appears on the Berry-Esseen bound in \eqref{eq:gamma}, which depends crucially on the leverages scores of the covariate matrix \citep{wang2022rerandomization}. 
Trimming can help us mitigate the extreme covariates, thereby improving the robustness of the best-choice rerandomization. 
For example, we trim each covariate at its $2.5\%$ and $97.5\%$ quantiles, and analogously calculate the worst-case performance as shown in Table \ref{tab:worst_case}. 
Obviously, the performance of the best-choice rerandomization enhances after trimming. 
Indeed,  the propensity scores of units under these rerandomization designs with the original covariates range from $0.00032$ to $0.146$, 
while that with trimmed covariates range from $0.109$ to $0.131$, which are much more stable and concentrate around $n_1/n = 0.121$.

\begin{remark}
We finally comment on the comparison between the best-choice rerandomization and rerandomization based on a prespecified imbalance threshold. 
Comparing Table \ref{tab:worst_case} and Table A1 in \citet{wang2022rerandomization}, 
given the same number of covariates,  
the worst-case mean squared errors under the best-choice rerandomization trying $T$ complete randomizations are comparable to the first-type rerandomization with acceptance probability $1/T$. 
This echos the intuitive and theoretical comparisons in Remarks \ref{rmk:first_type}--\ref{rmk:optimal}. 
However, the best-choice rerandomization can be more stable in practice, because it can always produce ``acceptable'' randomizations. 
Besides, its implementation is intuitive, convenient, and has already been used frequently in practice. 
\end{remark}

\section{Proof for the asymptotic distribution of the difference-in-means estimator under the best-choice rerandomization}

To prove Theorem \ref{thm:asymp}, we need the following three lemmas. 

\begin{lemma}\label{lemma:berry_psi}
Let $(\hat{\tau}, \hat{\bs{\tau}}_{\bs{x}}^\top)^\top$ be the difference-in-means of outcomes and covariates under the CRE, and $(\tilde{\tau}, \tilde{\bs{\tau}}_{\bs{x}}^\top)^\top$ be a Gaussian random vector with mean zero and covariance matrix $\bs{V}$ in \eqref{eq:V}. 
Let $M =  \hat{\bs{\tau}}_{\bs{x}}^\top \bs{V}_{\bs{xx}}^{-1} \hat{\bs{\tau}}_{\bs{x}}$ denote the Mahalanobis distance as in \eqref{eq:M_dist}, 
and define analogously $M =  \tilde{\bs{\tau}}_{\bs{x}}^\top \bs{V}_{\bs{xx}}^{-1} \tilde{\bs{\tau}}_{\bs{x}}$. 
Then for any $c_1, c_2\in \mathbb{R}$, with $\Delta_n$ defined as in \eqref{eq:Delta}, 
\begin{align*}
    \sup_{c_1, c_2 \in \R} \big| \Pr(
    \hat{\tau}-\tau \le c_1, 
    M \le c_2
    )
    - 
    \Pr(
    \tilde{\tau} \le c_1, 
    \tilde{M} \le c_2
    ) \big| & \le \Delta_n, \\
    \sup_{c_1, c_2 \in \R} \big| \Pr(
    \hat{\tau}-\tau \le c_1, 
    M < c_2
    )
    - 
    \Pr(
    \tilde{\tau} \le c_1, 
    \tilde{M} < c_2
    ) \big| & \le \Delta_n, \\
    \sup_{c\in \mathbb{R}} 
    \big| \Pr(
    M \ge c
    )
    - 
    \Pr(
    \tilde{M} \ge c
    ) \big| & \le \Delta_n, \\
    \sup_{c\in \mathbb{R}} 
    \big| \Pr(
    M > c
    )
    - 
    \Pr(
    \tilde{M} > c
    ) \big| & \le \Delta_n. 
\end{align*}
\end{lemma}

\begin{proof}[Proof of Lemma \ref{lemma:berry_psi}]
    Lemma \ref{lemma:berry_psi} follows directly from the definition of $\Delta_n$ and the fact that the sets $\{\bs{w}: \bs{w}^\top \bs{V}_{\bs{xx}}^{-1} \bs{w}\le c\}$ and $\{\bs{w}: \bs{w}^\top \bs{V}_{\bs{xx}}^{-1} \bs{w} < c\}$ are convex for any $c\in \mathbb{R}$. 
\end{proof}

\begin{lemma}\label{lemma:upper_bound}
Under the same setting as Lemma \ref{lemma:berry_psi}, 
let $(A_1, B_1), \ldots, (A_J, B_J)$ be any random vectors that are independent of $(\hat{\tau}, \hat{\bs{\tau}}_{\bs{x}}^\top)^\top$ and $(\tilde{\tau}, \tilde{\bs{\tau}}_{\bs{x}}^\top)^\top$;
Define 
\begin{align*}
    (A_0, B_0) \equiv (\hat{\tau}-\tau, M), 
    \quad 
    (\tilde{A}_0, \tilde{B}_0) \equiv (\tilde{\tau}, \tilde{M}), 
    \quad 
    (\tilde{A}_j, \tilde{B}_j) \equiv
    (A_j, B_j) \text{ for } 1\le j\le J. 
\end{align*}
Then  for any $c\in \mathbb{R}$, 
\begin{align*}
    \Pr\left(\bigcup_{j=0}^{J} \Big\{ A_j \leq c, B_j \leq \min_{k \neq j} B_{k}\Big\}\right)
    -
    \Pr\left(\bigcup_{j=0}^{J} \Big\{ \tilde{A}_j \leq c, \tilde{B}_j \leq \min_{k \neq j} \tilde{B}_{k}\Big\}\right)
    \le 2 \Delta_n. 
\end{align*}
\end{lemma}

\begin{proof}[Proof of Lemma \ref{lemma:upper_bound}]
By the law of iterated expectation, 
\begin{align}\label{eq:prob_union_cond_exp}
& \quad \ 
\Pr\left(\bigcup_{j=0}^{J} \Big\{ A_j \leq c, B_j \leq \min_{k \neq j} B_{k}\Big\}\right) = \E\left[\Pr\left(
\bigcup_{j=0}^{T} \Big\{ A_j \leq c, B_j \leq \min_{k \neq j} B_{k}\Big\}
\mid A_{1:J}, B_{1:J}\right)\right] 
\nonumber
\\
& = \E\left[\Pr\left(\Big\{ A_0 \leq c, B_0 \leq \min_{k>0} B_k\Big\} \bigcup \Big\{\cE_{0} \cap \{B_0 \geq  \min_{k>0} B_k \} \Big\} \mid A_{1:J}, B_{1:J}\right)\right],
\end{align}
where 
\[
\cE_{0} \equiv
\bigcup_{j=1}^{J} \Big\{ A_j \leq c, B_j \leq \min_{k \neq 0, j} B_{k}\Big\}
\]
is an event that becomes deterministic once conditioning on $A_{1:J}\equiv (A_1, \ldots, A_J)$ and $B_{1:J} \equiv (B_1, \ldots, B_J)$. 
By the union bound, 
\begin{align*}
	& \quad \ \Pr\left(\bigcup_{j=0}^{J} \Big\{ A_j \leq c, B_j \leq \min_{k \neq j} B_{k}\Big\}\right) \\
	& \leq \E\left[\Pr\Big(A_0 \leq c, B_0 \leq \min_{k > 0} B_k \mid A_{1:J}, B_{1:J}\Big) + \Pr\Big(B_0 \geq  \min_{k > 0} B_k \mid A_{1:J}, B_{1:J}\Big) \one(\cE_{0})\right]. 
\end{align*}
Note that both $(A_0, B_0)$ and $(\tilde{A}_0, \tilde{B}_0)$ are independent of $(A_{1:J}, B_{1:J})$. 
From Lemma \ref{lemma:berry_psi}, we then have 
\begin{align*}
 & \quad \ \Pr\left(\bigcup_{j=0}^{J} \Big\{ A_j \leq c, B_j \leq \min_{k \neq j} B_{k}\Big\}\right) \\
 & \leq \E\left[\Pr\Big(\tilde{A}_0 \leq c, \tilde{B}_0 \leq \min_{k > 0} B_k \mid A_{1:J}, B_{1:J}\Big) + \Pr\Big(\tilde{B}_0 \geq  \min_{k > 0} B_k \mid A_{1:J}, B_{1:J}\Big) \one(\cE_{0})\right] + 2 \Delta_n\\
 & =
 \E\left[\Pr\left(\Big\{ \tilde{A}_0 \leq c, \tilde{B}_0 \leq \min_{k>0} B_k\Big\} \bigcup \Big\{\cE_{0} \cap \{\tilde{B}_0 \geq  \min_{k>0} B_k \} \Big\} \mid A_{1:J}, B_{1:J}\right)\right] + 2 \Delta_n
 \\
& = \Pr\left(\bigcup_{j=0}^{J} \Big\{ \tilde{A}_j \leq c, \tilde{B}_j \leq \min_{k \neq j} \tilde{B}_{k}\Big\}\right) + 2 \Delta_n,
\end{align*}
where the second last equality holds because $\tilde{B}_0$ is a continuous random variable and consequently the 
measure of the intersection of the two events there is zero, 
and the last equality holds by the same logic as \eqref{eq:prob_union_cond_exp}. 
Therefore, Lemma \ref{lemma:upper_bound} holds. 
\end{proof}

\begin{lemma}\label{lemma:lower_bound}
    Under the same setting as Lemma \ref{lemma:upper_bound}, 
    for any $c\in \mathbb{R}$, 
    \begin{align*}
        \Pr\left(\bigcup_{j=0}^{J} \Big\{ A_j \leq c, B_j < \min_{k \neq j} B_{k}\Big\}\right)
        -
        \Pr\left(\bigcup_{j=0}^{J} \Big\{ \tilde{A}_j \leq c, \tilde{B}_j < \min_{k \neq j} \tilde{B}_{k}\Big\}\right)
        \ge - 2 \Delta_n. 
    \end{align*}
\end{lemma}
\begin{proof}[Proof of Lemma \ref{lemma:lower_bound}]
By the law of iterated expectation, 
\begin{align}\label{eq:prob_union_strict_cond_exp}
    & \quad \ \Pr\left(\bigcup_{j=0}^{J} \Big\{ A_j \leq c, B_j < \min_{k \neq j} B_{k}\Big\}\right)
    \nonumber
    \\
    & = 
    \E\left[\Pr\left(\Big\{ A_0 \leq c, B_0 < \min_{k>0} B_k\Big\} \bigcup \Big\{\cE_{0}' \cap \{ B_0 > \min_{k >0} B_k \} \Big\} \mid A_{1:J}, B_{1:J}\right)\right]
    \nonumber
    \\
    & = 
    \E\left[\Pr\left( A_0 \leq c, B_0 < \min_{k>0} B_k \mid A_{1:J}, B_{1:J}\right)\right]
    + 
    \E\left[
    \I(\cE_{0}')
    \Pr\left( B_0 > \min_{k >0} B_k\mid A_{1:J}, B_{1:J}\right)\right],
\end{align}
where 
\[
\cE_{0}' \equiv  \bigcup_{j=1}^J \Big\{A_j \le c,  B_j < \min_{k \neq 0, j} B_k \Big\}
\]
is an event that becomes deterministic once conditioning on $A_{1:J}\equiv (A_1, \ldots, A_J)$ and $B_{1:J} \equiv (B_1, \ldots, B_J)$, 
and the last equality holds because the two events there are disjoint. 
Note that both $(A_0, B_0)$ and $(\tilde{A}_0, \tilde{B}_0)$ are independent of $(A_{1:J}, B_{1:J})$. 
From Lemma \ref{lemma:berry_psi}, we then have 
\begin{align*}
    & \quad \ \Pr\left(\bigcup_{j=0}^{J} \Big\{ A_j \leq c, B_j < \min_{k \neq j} B_{k}\Big\}\right)\\
    & \ge 
    \E\left[\Pr\left( \tilde{A}_0 \leq c, \tilde{B}_0 < \min_{k>0} B_k \mid A_{1:J}, B_{1:J}\right)\right]
    + 
    \E\left[
    \I(\cE_{0}')
    \Pr\left( \tilde{B}_0 > \min_{k >0} B_k\mid A_{1:J}, B_{1:J}\right)\right]-2\Delta_n\\
    & = \Pr\left(\bigcup_{j=0}^{J} \Big\{ \tilde{A}_j \leq c, \tilde{B}_j < \min_{k \neq j} \tilde{B}_{k}\Big\}\right)-2\Delta_n,  
\end{align*}
where the last equality holds by the same logic as \eqref{eq:prob_union_strict_cond_exp}.  
\end{proof}

\begin{proof}[\bf Proof of Theorem \ref{thm:asymp}]
Recall the definition of $\bs{Z}_{[t]}$'s, $M_{[t]}$'s and $\hat{\tau}_{(1)}$ in Section \ref{sec:asym_bcr}. 
Define further $\hat{\tau}_{[t]}$ as the difference-in-means estimator under the treatment assignment $\bs{Z}_{[t]}$, for $1\le t\le T$. 
By the construction of the best-choice rerandomization, for any $c\in \mathbb{R}$, 
\begin{align}\label{eq:bound_bc_union}
    \Pr\left(\bigcup \limits_{t=1}^{T} \big\{ \hat{\tau}_{[t]} - \tau \leq c, M_{[t]} < \min_{j \neq t} M_{[j]}\big\}\right)
    \le 
    \Pr(\hat{\tau}_{(1)} - \tau \leq c) 
    \leq \Pr\left(\bigcup \limits_{t=1}^{T} \big\{ \hat{\tau}_{[t]} - \tau \leq c, M_{[t]} \leq \min_{j \neq t} M_{[j]}\big\}\right), 
\end{align}
where the inequality in \eqref{eq:bound_bc_union} comes mainly from the fact that there may be multiple treatment assignments achieving the minimum covariate imbalance. 
Recall also the definition of $\tilde{\tau}_{[t]}$'s and $\tilde{M}_{[t]}$'s in Section \ref{sec:dim_bcr}. 
Furthermore, without loss of generality, we assume that $\tilde{\tau}_{[t]}$'s and $\tilde{M}_{[t]}$'s are independent of $\hat{\tau}_{[t]}$'s and $M_{[t]}$'s. 
Obviously, for all $t$, $\tilde{M}_{[t]}$ follows the chi-squared distribution with degrees of freedom $K$. 

First, we consider the upper bound of $\Pr(\hat{\tau}_{(1)} - \tau \leq c)$. 
For $0 \le j\le T$, define  
\begin{align*}
    ( A_t^{(j)}, B_t^{(j)} ) & = 
    \begin{cases}
    (\hat{\tau}_{[t]} - \tau, M_{[t]}), & \textup{if } t>j, \\
    (\tilde{\tau}_{[t]}, M_{[t]}), & \textup{if } t\le j, 
    \end{cases}
    \quad (1\le t \le T).
\end{align*}
Obviously, $\{(A_t^{(j)}, B_t^{(j)}): 1\le t\le T\}$ and $\{(A_t^{(j+1)}, B_t^{(j+1)}): 1\le t\le T\}$ differ only in the $(j+1)$th element, for $0\le j \le T-1$. This allows us to apply Lemma~\ref{lemma:upper_bound} to get that, for any $0\le j \le T-1$, 
\[
\Pr\left(\bigcup \limits_{t=1}^{T} \big\{ A_t^{(j)} \leq c, B_t^{(j)} \leq \min_{j \neq t} B_{j}^{(j)}\big\}\right) \le \Pr\left(\bigcup \limits_{t=1}^{T} \big\{ A_t^{(j + 1)} \leq c, B_t^{(j + 1)} \leq \min_{j \neq t} B_{j}^{(j + 1)}\big\}\right) + 2 \Delta_n,
\]
where in the above inequality we take $(A_{j + 1}^{(j)}, B_{j + 1}^{(j)})$ and $(A_{j + 1}^{(j + 1)}, B_{j + 1}^{(j + 1)})$ as $(A_0, B_0)$ and $(\tilde{A}_0, \tilde{B}_0)$; and take $\{(A_t^{(j)}, B_t^{(j)}), 0 \le t \le T \;\&\; t \neq j + 1\}$ and $\{(A_t^{(j + 1)}, B_t^{(j + 1)}), 0 \le t \le T \;\&\; t \neq j + 1\}$ as $\{(A_j, B_j), 1 \le j \le J\}$ and $\{(\tilde{A}_j, \tilde{B}_j), 1 \le j \le J\}$.
Armed with the above inequality, we have 
\begin{align*}
    \Pr(\hat{\tau}_{(1)} - \tau \leq c) 
    & \leq \Pr\left(\bigcup \limits_{t=1}^{T} \big\{ \hat{\tau}_{[t]} - \tau \leq c, M_{[t]} \leq \min_{j \neq t} M_{[j]}\big\}\right)
    = \Pr\left(\bigcup \limits_{t=1}^{T} \big\{ A_t^{(0)} \leq c, B_t^{(0)} \leq \min_{j \neq t} B_{j}^{(0)}\big\}\right)
    \\
    & \le \Pr\left(\bigcup \limits_{t=1}^{T} \big\{ A_t^{(1)} \leq c, B_t^{(1)} \leq \min_{j \neq t} B_{j}^{(1)}\big\}\right) + 2\Delta_n\\
    & \le \Pr\left(\bigcup \limits_{t=1}^{T} \big\{ A_t^{(2)} \leq c, B_t^{(2)} \leq \min_{j \neq t} B_{j}^{(2)}\big\}\right) + 2\times 2\Delta_n\\
    & \le \ldots 
    \le \Pr\left(\bigcup \limits_{t=1}^{T} \big\{ A_t^{(T)} \leq c, B_t^{(T)} \leq \min_{j \neq t} B_{j}^{(T)}\big\}\right) + T\times 2\Delta_n\\
    & = \Pr\left(\bigcup \limits_{t=1}^{T} \big\{ \tilde{\tau}_{[t]} \leq c, \tilde{M}_{[t]} \leq \min_{j \neq t} \tilde{M}_{[j]} \big\}\right) + 2T\Delta_n. 
\end{align*}

Second, we consider the lower bound of $\Pr(\hat{\tau}_{(1)} - \tau \leq c)$. By the same logic as the proof of the upper bound and applying Lemma \ref{lemma:lower_bound}, we have 
\begin{align*}
    \Pr(\hat{\tau}_{(1)} - \tau \leq c)  
    \ge 
    \Pr\left(\bigcup \limits_{t=1}^{T} \big\{ \tilde{\tau}_{[t]} \leq c, \tilde{M}_{[t]} < \min_{j \neq t} \tilde{M}_{[j]} \big\}\right)
    - 2T \Delta_n. 
\end{align*}

Third, we prove that, for any $c\in \mathbb{R}$, 
\begin{align}\label{eq:bound_bc_union_tilde}
    \Pr\left(\bigcup \limits_{t=1}^{T} \big\{ \tilde{\tau}_{[t]} \leq c, \tilde{M}_{[t]} < \min_{j \neq t} \tilde{M}_{[j]}\big\}\right)
    = 
    \Pr(\tilde{\tau}_{(1)} \leq c) 
    = \Pr\left(\bigcup \limits_{t=1}^{T} \big\{ \tilde{\tau}_{[t]} \leq c, \tilde{M}_{[t]} \leq \min_{j \neq t} \tilde{M}_{[j]}\big\}\right). 
\end{align}
By the same logic as \eqref{eq:bound_bc_union}, the left-hand side of \eqref{eq:bound_bc_union_tilde} is bounded from above by $\Pr(\tilde{\tau}_{(1)} \leq c) $, which is further bounded from above by the right-hand side of \eqref{eq:bound_bc_union_tilde}. 
Furthermore,  the right-hand side of \eqref{eq:bound_bc_union_tilde} is also bounded from above by the left-hand side: 
\begin{align*}
    & \quad \ \Pr\left(\bigcup \limits_{t=1}^{T} \big\{ \tilde{\tau}_{[t]} \leq c, \tilde{M}_{[t]} \leq \min_{j \neq t} \tilde{M}_{[j]}\big\}\right)
    - 
    \Pr\left(\bigcup \limits_{t=1}^{T} \big\{ \tilde{\tau}_{[t]} \leq c, \tilde{M}_{[t]} < \min_{j \neq t} \tilde{M}_{[j]}\big\}\right)\\
    & \le 
    \Pr\left(\bigcup \limits_{t=1}^{T} \bigcup \limits_{j\ne t} \big\{ \tilde{M}_{[t]} = \tilde{M}_{[j]} \big\}\right) 
    = 0, 
\end{align*}
where the last equality holds because $\tilde{M}_{[j]}$'s are mutually independent continuous random variables. 
These facts then imply that \eqref{eq:bound_bc_union_tilde} must hold. 

From the above, we then have that
$| \Pr(\hat{\tau}_{(1)} - \tau \leq c)   - 
    \Pr(\tilde{\tau}_{(1)} \leq c) |
    \le 2 T \Delta_n$ for any $c\in \mathbb{R}$. 
Equivalently, the inequality in \eqref{eq:bound_diff_tilde} holds. 
If Conditions \ref{cond:gamma} and \ref{cond:iterations} hold, then $T\Delta_n = o(1)$, and consequently 
the supremum in \eqref{eq:bound_diff_tilde} converges to zero as $n\rightarrow \infty$. 
Therefore, Theorem \ref{thm:asymp} holds. 
\end{proof}

To prove Theorem \ref{thm:asym_equiv}, we need the following lemma. 

\begin{lemma}\label{lemma:unit_vec_inv}
    Let $\bs{D}_1, \ldots, \bs{D}_T\in \mathbb{R}^K$ be $T$ mutually independent $K$-dimensional standard Gaussian random vectors. 
    Then, for any constant unit vector $\bs{c}\in \mathbb{R}^K$, 
    \begin{align*}
        \bs{c}^\top \bs{D}_1 \mid \|\bs{D}_1\|_2^2 \le \min_{1\le j \le T} \|\bs{D}_j\|_2^2
        \ \sim \  
        D_{11} \mid \|\bs{D}_1\|_2^2 \le \min_{1\le j \le T} \|\bs{D}_j\|_2^2,
    \end{align*}
    where $D_{11}$ is the first coordinate of $\bs{D}_1$. 
\end{lemma}

\begin{proof}[Proof of Lemma \ref{lemma:unit_vec_inv}]
    For any given unit vector $\bs{c}\in \mathbb{R}^K$, we can always construct an orthogonal matrix $\bs{C}$ whose first row is $\bs{c}^\top$. 
    Then $\bs{c}^\top \bs{D}_1$ will be the first coordinate of $\bs{C}\bs{D}_1$, 
    and $\|\bs{C}\bs{D}_j\|_2^2 = \|\bs{D}_j\|_2^2$ for all $1\le j \le T$.
    By the property of standard Gaussian distributions, 
    $(\bs{C}\bs{D}_1, \ldots, \bs{C}\bs{D}_T)$ follows the same distribution as $(\bs{D}_1, \ldots, \bs{D}_T)$. 
    This then implies Lemma \ref{lemma:unit_vec_inv}. 
\end{proof}

\begin{proof}[\bf Proof of Theorem \ref{thm:asym_equiv}]
    From equation \eqref{eq:bound_bc_union_tilde} in the proof of Theorem \ref{thm:asymp}, for any $c\in \mathbb{R}$, 
    \begin{align*}
        \Pr(\tilde{\tau}_{(1)} \leq c) 
        & = \Pr\left(\bigcup \limits_{t=1}^{T} \big\{ \tilde{\tau}_{[t]} \leq c, \tilde{M}_{[t]} \le \min_{j \neq t} \tilde{M}_{[j]}\big\}\right)
        = 
        \sum_{t=1}^T \Pr\Big( \tilde{\tau}_{[t]} \leq c, \tilde{M}_{[t]} \le \min_{j \neq t} \tilde{M}_{[j]} \Big),
    \end{align*}
    where the last equality holds because $\tilde{M}_{[t]}$'s are mutually independent continuous random variables. 
    Because $(\tilde{\tau}_{[t]}, \tilde{M}_{[t]})$'s are i.i.d.\  across all $t$, and $\tilde{M}_{[t]}$'s are continuous random variables,  
    we then have $\Pr( \tilde{M}_{[1]} \le \min_{1\le j \le T} \tilde{M}_{[j]} ) = 1/T$, and 
\begin{align*}
    & \quad \ \Pr(\tilde{\tau}_{(1)} \leq c)
    \\
    & = 
    T  \cdot \Pr\Big( \tilde{\tau}_{[1]} \leq c, \tilde{M}_{[1]} \le \min_{1\le j \le T} \tilde{M}_{[j]} \Big)
    = 
    T \cdot 
    \Pr\Big( \tilde{M}_{[1]} \le \min_{1\le j \le T} \tilde{M}_{[j]} \Big)
    \cdot 
    \Pr\Big( \tilde{\tau}_{[1]} \leq c \mid \tilde{M}_{[1]} \le \min_{1\le j \le T} \tilde{M}_{[j]} \Big)
    \\
    & = \Pr\Big( \tilde{\tau}_{[1]} \leq c \mid \tilde{M}_{[1]} \le \min_{1\le j \le T} \tilde{M}_{[j]} \Big). 
\end{align*}
Let 
$\tilde{\tau}_{[1]}^\perp = \tilde{\tau}_{[1]} - V_{\tau \bs{x}} V_{\bs{xx}}^{-1} \tilde{\bs{\tau}}_{\bs{x}[1]}$. 
We can verify that $\tilde{\tau}_{[1]}^\perp\sim \mathcal{N}(0, V_{\tau\tau}(1-R^2))$ and it is independent from all the  $\tilde{\bs{\tau}}_{\bs{x}[t]}$'s. 
Let $\varepsilon_0 \sim \mathcal{N}(0,1)$ be a standard Gaussian random variable independent of all the  $\tilde{\bs{\tau}}_{\bs{x}[t]}$'s, 
and $\bs{D}_{t} = V_{\bs{xx}}^{-1/2} \tilde{\bs{\tau}}_{\bs{x}[t]}$ for $1\le t\le T$. 
We then have, for any $c\in \mathbb{R}$, 
\begin{align*}
    \Pr(\tilde{\tau}_{(1)} \leq c)
    & = 
    \Pr\Big( \tilde{\tau}_{[1]} \leq c \mid \tilde{M}_{[1]} \le \min_{1\le j \le T} \tilde{M}_{[j]} \Big)
    = 
    \Pr\Big( \tilde{\tau}_{[1]}^\perp + V_{\tau \bs{x}} V_{\bs{xx}}^{-1} \tilde{\bs{\tau}}_{\bs{x}[1]} \leq c \mid \tilde{M}_{[1]} \le \min_{1\le j \le T} \tilde{M}_{[j]} \Big)\\
    & = 
    \Pr\Big( \sqrt{V_{\tau\tau}(1-R^2)} \ \varepsilon_0 + V_{\tau \bs{x}} V_{\bs{xx}}^{-1/2} \bs{D}_1 \leq c \mid 
    \|\bs{D}_1\|_2^2 \le \min_{1\le j \le T} \|\bs{D}_j\|_2^2
    \Big)\\
    & = 
    \Pr\Big( \sqrt{V_{\tau\tau}(1-R^2)} \ \varepsilon_0 + 
    \sqrt{V_{\tau\tau}R^2}\  \bs{h}^\top \bs{D}_1 \leq c \mid 
    \|\bs{D}_1\|_2^2 \le \min_{1\le j \le T} \|\bs{D}_j\|_2^2
    \Big), 
\end{align*}
where $\bs{h} = (V_{\tau\tau}R^2)^{-1/2} V_{\bs{xx}}^{-1/2} V_{\bs{x}\tau}$ is a unit vector of length $1$ by the definition of $R^2$ in \eqref{eq:R2}. 
From Lemma \ref{lemma:unit_vec_inv}, this further implies that  
$
    \Pr(\tilde{\tau}_{(1)} \leq c)
    = 
    \Pr( \sqrt{V_{\tau\tau}(1-R^2)} \ \varepsilon_0 + 
    \sqrt{V_{\tau\tau}R^2}\  L_{K,T} \leq c
    )
$
for all $c\in \mathbb{R}$. 
We can then immediately derive Theorem \ref{thm:asym_equiv}. 
\end{proof}

\section{Proof for the properties of the asymptotic distribution}

\begin{proof}[\bf Proof of Proposition \ref{prop:L_KT} and the equivalence in \eqref{eq:chi2_KT}]
    We first prove Proposition \ref{prop:L_KT}. 
    From Lemma \ref{lemma:unit_vec_inv}, 
    $L_{K, T} \sim \bs{c}^\top \bs{D}_1 \mid \|\bs{D}_1\|_2^2 \le \min_{1 \le t \le T} \|\bs{D}_t\|_2^2$ for any constant unit vector $c\in \mathbb{R}^K$.
    Moreover, from \citet[Lemma A2]{LDR18}, 
    it suffices to prove that $L_{K,T} \sim \chi_{K,T} U_K$. 
    For all $1\le t\le T$, 
    define $\xi_t = \|\bs{D}_t\|_2$. 
    By the property of the multivariate standard Gaussian distribution, 
    $\xi_t^2$ follows chi-squared distribution with degrees of freedom $K$, 
    $\bs{D}_t/\xi_t$ follows the uniform distribution on the $K-1$ dimensional unit sphere, and they are mutually independent.
    These imply that, with $D_{11}$ and $[\bs{D}_1/\xi_1]_1$ being the first coordinates of $\bs{D}_1$ and $\bs{D}_1/\xi_1$, respectively,  
    \begin{equation}\label{eq:L_KT_equiv_proof}
        L_{K,T} 
        \ \sim \ D_{11} \mid \|\bs{D}_1\|_2^2 \le \min_{1 \le t \le T} \|\bs{D}_j\|_2^2
        \ \sim \  [\bs{D}_t/\xi_t]_{1} \xi_1 \mid \xi_1^2 \le \min_{1 \le t \le T} \xi_t^2
        \ \sim \  U_K \xi_1 \mid \xi_1 \le \min_{1 \le t \le T} \xi_t. 
    \end{equation}
    Consequently, Proposition \ref{prop:L_KT} holds.

    We then prove the equivalence in \eqref{eq:chi2_KT}. 
    The fact that $\chi^2_{K(1)} \sim F_K^{-1}(\text{Beta}(1,T))$ follows from the property of order statistics. 
    It then suffices to prove that $\xi_1 \mid \xi_1 \le \min_{1 \le t \le T} \xi_t \sim \xi_{(1)}$, where $\xi_{(1)} = \min_{1\le t \le T} \xi_{(t)}$. 
    This is true because, for any $c\in \mathbb{R}$, 
    \begin{align*}
        \pr(\xi_1 \le c \mid \xi_1 \le \min_{1 \le t \le T} \xi_t)
        & = 
        \frac{\pr(\xi_1 \le c, \xi_1 \le \min_{1 \le t \le T} \xi_t)}{\pr(\xi_1 \le \min_{1 \le t \le T} \xi_t)}
        = T \cdot \pr(\xi_1 \le c, \xi_1 \le \min_{1 \le t \le T} \xi_t)
        \\
        & =
        \sum_{j=1}^T \pr(\xi_j \le c, \xi_j \le \min_{1 \le t \le T} \xi_t)
        = 
        \pr(\xi_{(1)} \le c),
    \end{align*}
    where the equalities hold by symmetry and the fact that $\xi_t$'s are continuous random variables. 
\end{proof}

To prove Corollary \ref{cor:sum}, we need the following two lemmas. 

\begin{lemma}\label{lemma:lktunimodal}
    $L_{K, T}$ is a continuous random variable, and is symmetric and unimodal around zero.
\end{lemma}

\begin{proof}[Proof of Lemma~\ref{lemma:lktunimodal}]
    When $T=1$, $L_{K, T}\sim \mathcal{N}(0,1)$, and Lemma \ref{lemma:lktunimodal} holds obviously. Below we consider the case where $T\ge 2$. 
    By the same logic as in the proof of Lemma \ref{lemma:prob_R2}, 
    \begin{align*}
        \pr( L_{K,T} \le c)= 
        T \cdot \int_0^{\infty} \pr( L_{K,a}' \le c) 
        \pr(\chi_{K}^2 \le a )
        g(a) \deri a,
    \end{align*}
    where $L_{K,a}'$ is defined as in \citet[Proposition 2]{LDR18}, $\chi_{K}^2$ follows the chi-squared distribution with degrees of freedom $K$, and $g(\cdot)$ denote the density of the minimum of $T-1$ i.i.d.\ chi-squared random variables with degrees of freedom $K$. 
    For any $a>0$, 
    let $f_{K,a}'(x)$ denote the density of $L_{K, a}'$ as derived in \citet[][Proof of Proposition 2]{LDR18}. 
    We then have 
    \begin{align*}
        \pr(L_{K, T} \leq c) & = 
        T\cdot \int_0^{\infty}
        \pr( L_{K,a}' \le c )
        \pr( \chi^2_K \leq a ) g(a) \deri a
        = 
        \int_0^{\infty}
        \pr( L_{K,a}' \le c )
        g'(a) \deri a\\
        & = 
        \int_0^{\infty} \int_{-\infty}^{c}
        f_{K,a}'(x) \deri x 
        g'(a) \deri a
        = 
        \int_{-\infty}^{c} \int_0^{\infty} 
        f_{K,a}'(x) g'(a) \deri a \deri x, 
    \end{align*}
    where $g'(a) = T\cdot \pr( \chi^2_K \leq a ) g(a)$. 
    This then implies that $L_{K, a}$ is a continuous random variable, and its density has the following form:
    \begin{align*}
        f_{K,T}(x) = \int_0^{\infty} 
        f_{K,a}'(x) g'(a) \deri a. 
    \end{align*}
    Because $L_{K,a}'$ is symmetric and unimodal around zero \citep[][Proposition 2]{LDR18}, 
    we must have, for any $a>0$,  $f_{K,a}'(x) = f_{K,a}'(-x)$ and 
    $f_{K,a}'(x_1) \ge f_{K,a}'(x_2)$ for any $x_2 \ge x_1 \ge 0$. 
    These then imply that 
    \begin{align*}
        f_{K,T}(-x) = \int_0^{\infty} 
        f_{K,a}'(-x) g'(a) \deri a = \int_0^{\infty} 
        f_{K,a}'(x) g'(a) \deri a = f_{K,T}(x), 
    \end{align*}
    and, for any any $x_2 \ge x_1 \ge 0$,
    \begin{align*}
        f_{K,T}(x_1) = \int_0^{\infty} 
        f_{K,a}'(x_1) g'(a) \deri a 
        \ge \int_0^{\infty} 
        f_{K,a}'(x_2) g'(a) \deri a = f_{K,T}(x_2). 
    \end{align*}
    Thus, $L_{K,T}$ is also symmetric and unimodal around zero. 
    From the above, Lemma~\ref{lemma:lktunimodal} holds. 
\end{proof}

\begin{lemma}\label{lemma:sum_symmetric_unimodal}
    If both $\xi_1$ and $\xi_2$ are symmetric and unimodal around zero, and they are mutually independent, 
    then $\xi_1+\xi_2$ are also symmetric and unimodal around zero. 
\end{lemma}

\begin{proof}[Proof of Lemma \ref{lemma:sum_symmetric_unimodal}]
    Lemma \ref{lemma:sum_symmetric_unimodal} follows directly from \citet{wintner1936class}. 
\end{proof}

\begin{proof}[\bf Proof of Corollary \ref{cor:sum}]
    It is not difficult to see that 
    the standard Gaussian random variable $\varepsilon_0$ is symmetric and unimodal around zero. 
    From Lemma \ref{lemma:lktunimodal}, the constrained Gaussian random variable $L_{K,T}$ is also symmetric and unimodal around zero. 
    Consequently, 
    from Theorems \ref{thm:asymp} and \ref{thm:asym_equiv}, 
    we can derive that the asymptotic distribution of the difference-in-means estimator under the best-choice rerandomization is symmetric and unimodal around zero, i.e., Corollary \ref{cor:sum} holds. 
\end{proof}

To prove Corollary \ref{cor:reduce_var}, we need the following lemma.

\begin{lemma}\label{lemma:nonincreasing}
    For any fixed $K\ge 1$, $\Var(L_{K, T})$ is nonincreasing in $T$, 
    and 
    $\Var(L_{K, T}) < 1$ for for any $T \geq 2$. 
\end{lemma}

\begin{proof}[Proof of Lemma \ref{lemma:nonincreasing}]  
    For any $K,T\ge 1$, let $U_K$ be the first coordinate of a $K$-dimensional random vector uniformly distributed on the $(K-1)$-dimensional unit sphere and
    $\chi^2_{K[1]}, \chi^2_{K[2]}, \ldots, \chi^2_{K[T+1]}$ be i.i.d.\ chi-squared random variables with degrees of freedom $K$, 
    and assume that they are all mutually independent. 
    Let $\chi_{K,T}^2 = \min_{1\le t \le T} \chi^2_{K[t]}$ and $\chi_{K,T+1}^2 = \min_{1\le t \le T+1} \chi^2_{K[t]}$. 
    From Proposition \ref{prop:L_KT} and Lemma \ref{lemma:lktunimodal}, we can know that
    \begin{align*}
        \Var(L_{K, T}) = \E(\chi_{K,T}^2) \cdot \E(U_K^2)
        \ge 
        \E(\chi_{K,T+1}^2) \cdot \E(U_K^2) = \Var(L_{K, T+1}).
    \end{align*}    
    Thus, $\Var(L_{K, T})$ is nonincreasing in $T$. 

    We then prove that $\Var(L_{K, T}) < 1$ for $T\ge 2$. Because $\Var(L_{K, T})$ is nonincreasing in $T$,  it suffices to prove that $\Var(L_{K, 2}) < 1$. 
    Define $U_K, \chi^2_{K[1]}, \chi^2_{K[2]}$ the same as before, and assume that they are mutually independent. 
    From the proof of Proposition \ref{prop:L_KT}, we can know that $U_K \sim D_{11}/\|\bs{D}_1\|_2$, where $\bs{D}_1$ follows $K$-dimensional standard Gaussian distribution and $D_{11}$ is the first coordinate of $\bs{D}_1$. By symmetry, we then have $\E(U_k^2) = 1/K$. 
    Consequently, from Proposition \ref{prop:L_KT} and Lemma \ref{lemma:lktunimodal}, we have 
    \begin{align*}%
        1 - \Var(L_{K, 2})  
        & = 1 - \frac{1}{K} \E\big(\min_{t = 1,2} \chi^2_{K[t]}\big) = 
        \frac{1}{K} \Big\{ \E(\chi^2_{K[1]}) - \E\big(\min_{t = 1,2} \chi^2_{K[t]}\big) \Big\}
        \nonumber
        \\
        & 
        = \frac{1}{K} \Big\{ \E(\chi^2_{K[1]}) - \E\big(\min_{t = 1,2} \chi^2_{K[t]}\big) \Big\}
        = \frac{1}{K} \E\big(\chi^2_{K,1} - \min_{t = 1,2} \chi^2_{K[t]}\big). 
    \end{align*}
    Thus, to prove that $\Var(L_{K, 2}) < 1$, it suffices to show that $\E(\chi^2_{K[1]} - \min_{t = 1,2} \chi^2_{K[t]}) > 0$. 
    We prove this by contradiction. 
    Suppose that $\E(\chi^2_{K[1]} - \min_{t = 1,2} \chi^2_{K[t]}) = 0$. 
    Because $\chi^2_{K[1]} - \min_{t = 1,2} \chi^2_{K[t]}$ is a nonnegative random variable, the zero mean then implies that 
    $\chi^2_{K[1]} - \min_{t = 1,2} \chi^2_{K[t]}=0$ almost surely, or equivalently $\chi^2_{K[1]} \le \chi^2_{K[2]}$ almost surely. 
    However, $\Pr(\chi^2_{K[1]} \le \chi^2_{K[2]}) = 1/2$ by symmetry, leading to a contradiction. 

    From the above, Lemma \ref{lemma:nonincreasing} holds. 
\end{proof}

\begin{proof}[\bf Proof of Corollary \ref{cor:reduce_var}]
    From Theorems \ref{thm:asymp} and \ref{thm:asym_equiv}, 
    the asymptotic variance of the difference-in-means estimator scaled by $V_{\tau\tau}^{-1/2}$ under the best-choice rerandomization is 
    $(1-R^2) + R^2 v_{K,T} = 1 - (1-v_{K,T})R^2$. 
    The asymptotic variance of the difference-in-means estimator scaled by $V_{\tau\tau}^{-1/2}$ under the CRE, which can also be viewed as a special case of the best-choice rerandomization with $T=1$, is $1$. 
    Consequently, the percentage reduction in asymptotic variance is $(1-v_{K,T})R^2$. 
    From Lemma \ref{lemma:nonincreasing}, 
    the percentage reduction is nonnegative and is nondecreasing in $T$. 
    By its expression, the percentage reduction is obviously nondecreasing in $R^2$. 
    From the above, Corollary \ref{cor:reduce_var} holds. 
\end{proof}

To prove Corollary \ref{cor:reduce_var}, we need the following four lemmas.

\begin{lemma}\label{lemma:prob_T}
    For any given $K\ge 1$ and $c\ge 0$, the probability $\pr(L_{K,T} \ge c)$ is nonincreasing in $T\ge 1$. 
\end{lemma}

\begin{proof}[Proof of Lemma \ref{lemma:prob_T}]
    For any $K,T\ge 1$, let $U_K$ be the first coordinate of a $K$-dimensional random vector uniformly distributed on the $(K-1)$-dimensional unit sphere and
    $\chi^2_{K[1]}, \chi^2_{K[2]}, \ldots, \chi^2_{K[T+1]}$ be i.i.d.\ chi-squared random variables with degrees of freedom $K$, 
    and assume that they are all mutually independent. 
    Let $\chi_{K,T}^2 = \min_{1\le t \le T} \chi^2_{K[t]}$ and $\chi_{K,T+1}^2 = \min_{1\le t \le T+1} \chi^2_{K[t]}$. 
    From Lemma \ref{lemma:sum_symmetric_unimodal} and Proposition \ref{prop:L_KT}, for any $c\ge 0$, 
    \begin{align*}
        2\pr(L_{K,T} \ge c) 
        & = \pr(|L_{K,T}| \ge c) 
        = \pr(|\chi_{K,T}| |U_K| \ge c) 
        \ge \pr(|\chi_{K,T+1}| |U_K| \ge c) 
        = 
        \pr(|\chi_{K,T+1}| |U_K| \ge c)
        \\
        & = 2\pr(L_{K,T+1} \ge c). 
    \end{align*}
    Therefore, for any $c\ge 0$, $\pr(L_{K,T} \ge c)$ is nonincreasing in $T$, i.e., Lemma \ref{lemma:prob_T} holds. 
\end{proof}

\begin{lemma}\label{lemma:sum}
    Let $\zeta_0, \zeta_1$ and $\zeta_2$ be three random variables, where $\zeta_0 \ind \zeta_1$ and $\zeta_0 \ind \zeta_2$. If
    \begin{enumerate}
        \item[(1)] $\zeta_0$ is continuous and symmetric and unimodal around zero, or $\zeta_0 = 0$,
        \item[(2)] $\zeta_1$ and $\zeta_2$ are symmetric and unimodal around zero,
        \item[(3)] $\pr(\zeta_1 \geq c) \leq \pr(\zeta_2 \geq c)$ for any $c > 0$,
    \end{enumerate}
    then $\pr(\zeta_0 + \zeta_1 \geq c) \leq \pr(\zeta_0 + \zeta_2 \geq c)$ for any $c > 0$.
\end{lemma}

\begin{proof}[Proof of Lemma \ref{lemma:sum}]
    Lemma \ref{lemma:sum} follows directly from \citet[lemma~A7]{LDR18}; 
    see also \citet[Theorem 7.5]{DJ88}. 
\end{proof}

\begin{lemma}\label{lemma:prob_sum_T}
    For any given $K\ge 1$, $R^2\in [0,1]$ and $c\ge 0$, the probability $\pr(\sqrt{1-R^2} \varepsilon_0 + \sqrt{R^2} L_{K,T} \ge c)$ is nonincreasing in $T\ge 1$, 
    where $\varepsilon_0\sim \mathcal{N}(0,1)$ and is independent of $L_{K,T}$. 
\end{lemma}

\begin{proof}[Proof of Lemma \ref{lemma:prob_sum_T}]
    Lemma \ref{lemma:prob_sum_T} holds obivously when $c=0$, because $\sqrt{1-R^2} \varepsilon_0 + \sqrt{R^2} L_{K,T}$ is symmetric around zero. 
    When $c>0$, 
    Lemma \ref{lemma:prob_sum_T} follows immediately from Lemmas \ref{lemma:sum_symmetric_unimodal}, \ref{lemma:prob_T} and \ref{lemma:sum}. 
\end{proof}

\begin{lemma}\label{lemma:prob_R2}
    For any given $K, T\ge 1$ and $c\ge 0$, the probability $\pr(\sqrt{1-R^2} \varepsilon_0 + \sqrt{R^2} L_{K,T} \ge c)$ is nonincreasing in $R^2\in [0,1]$, 
    where $\varepsilon_0\sim \mathcal{N}(0,1)$ and is independent of $L_{K,T}$. 
\end{lemma}

\begin{proof}[Proof of Lemma \ref{lemma:prob_R2}]
    Because $L_{K,1} \sim \mathcal{N}(0,1)$, Lemma \ref{lemma:prob_R2} holds obviously when $T=1$. 
    For any $K\ge 1$ and $T \ge 2$, let $U_K$ be the first coordinate of a $K$-dimensional random vector unformly distributed on the $(K-1)$-dimensional unit sphere and
    $\chi^2_{K[1]}, \chi^2_{K[2]}, \ldots, \chi^2_{K[T]}$ be i.i.d.\ chi-squared random variables with degrees of freedom $K$, 
    and assume that they are all mutually independent and are independent of $\varepsilon_0$. 
    From \eqref{eq:L_KT_equiv_proof}, for any $R^2\in [0,1]$ and $c\ge 0$, 
    \begin{align*}
        \pr(\sqrt{1-R^2} \varepsilon_0 + \sqrt{R^2} L_{K,T} \ge c)
        & = 
        \pr\big( \sqrt{1-R^2} \varepsilon_0 + \sqrt{R^2} U_K \chi_{K[1]} \ge c \mid \chi_{K[1]}^2 \le \min_{2\le t\le T} \chi^2_{K[t]} \big)\\
        & = 
        \frac{\pr\big( \sqrt{1-R^2} \varepsilon_0 + \sqrt{R^2} U_K \chi_{K[1]} \ge c, \chi_{K[1]}^2 \le \min_{2\le t\le T} \chi^2_{K[t]} \big)}{
        \pr\big( \chi_{K[1]}^2 \le \min_{2\le t\le T} \chi^2_{K[t]} \big)
        }\\
        & = 
        T \cdot \pr\big( \sqrt{1-R^2} \varepsilon_0 + \sqrt{R^2} U_K \chi_{K[1]} \ge c, \chi_{K[1]}^2 \le \min_{2\le t\le T} \chi^2_{K[t]} \big).
    \end{align*}
    Let $g(a)$ denote the density of $\min_{2\le t\le T} \chi^2_{K[t]}$. We then have 
    \begin{align*}
        & \quad \ \pr(\sqrt{1-R^2} \varepsilon_0 + \sqrt{R^2} L_{K,T} \ge c)\\
        & = 
        T \cdot \int_0^{\infty} \pr\big( \sqrt{1-R^2} \varepsilon_0 + \sqrt{R^2} U_K \chi_{K[1]} \ge c, \chi_{K[1]}^2 \le a  \big) g(a) \deri a
        \\
        & = 
        T \cdot \int_0^{\infty} \pr\big( \sqrt{1-R^2} \varepsilon_0 + \sqrt{R^2} U_K \chi_{K[1]} \ge c \mid  \chi_{K[1]}^2 \le a  \big) 
        \pr(\chi_{K[1]}^2 \le a )
        g(a) \deri a\\
        & = 
        T \cdot \int_0^{\infty} \pr( \sqrt{1-R^2} \varepsilon_0 + \sqrt{R^2} L_{K,a}' \ge c) 
        \pr(\chi_{K[1]}^2 \le a )
        g(a) \deri a,
    \end{align*}
    where $L_{K,a}'$ is defined as in \citet[Proposition 2]{LDR18}. 
    From \citet[Lemma A4]{LDR18}, for any $0\le R_1^2 \le R_2^2 \le 1$, 
    $\pr( \sqrt{1-R_1^2} \varepsilon_0 + \sqrt{R_1^2} L_{K,a}' \ge c) \ge \pr( \sqrt{1-R_2^2} \varepsilon_0 + \sqrt{R_2^2} L_{K,a}' \ge c)$ for any $a>0$, and thus 
    \begin{align*}
        \pr(\sqrt{1-R_1^2} \varepsilon_0 + \sqrt{R_1^2} L_{K,T} \ge c)
        & = 
        T \cdot \int_0^{\infty} \pr( \sqrt{1-R_1^2} \varepsilon_0 + \sqrt{R_1^2} L_{K,a}' \ge c) 
        \pr(\chi_{K[1]}^2 \le a )
        g(a) \deri a
        \\
        & \ge 
        T \cdot \int_0^{\infty} \pr( \sqrt{1-R_2^2} \varepsilon_0 + \sqrt{R_2^2} L_{K,a}' \ge c) 
        \pr(\chi_{K[1]}^2 \le a )
        g(a) \deri a\\
        & = \pr(\sqrt{1-R_2^2} \varepsilon_0 + \sqrt{R_2^2} L_{K,T} \ge c). 
    \end{align*}
    From the above, Lemma \ref{lemma:prob_R2} holds. 
\end{proof}

\begin{proof}[\bf Proof of Corollary \ref{cor:reduce_qr}]
    Note that when $R^2=0$ or $T=1$, the asymptotic distribution of $V_{\tau\tau}^{-1/2} (\hat{\tau}_{(1)} - \tau)$ under the best-choice rerandomization reduces to that under the CRE, i.e., a standard Gaussian distribution. 
    From Lemmas \ref{lemma:prob_sum_T} and \ref{lemma:prob_R2},
    the asymptotic symmetric quantile ranges under the best-choice rerandomization will be shorter than that under the CRE, 
    and, moreover, the percentage reduction is nondecreasing in $R^2$ and $T$. 
    Therefore, Corollary \ref{cor:reduce_qr} holds. 
\end{proof}

\section{Proof for the asymptotic behavior of the constrained Gaussian random variable}

Below we first show that, for any sequence of positive integers $\{K_n: n\ge 1\}$ and $\{T_n: n\ge 1\}$, 
$L_{K_n, T_n} = o_{\Pr}(1)$ if and only if $\Var(L_{K_n, T_n}) = o(1)$. 
By the same logic as \citet[][Proposition A2]{wang2022rerandomization}, 
it suffices to show that $\{L_{K_n, T_n}^2: n\ge 1\}$ is uniformly integrable. 
From Lemma \ref{lemma:prob_T}, 
for any $K,T\ge 1$, $L_{K,T}$ is stochastically smaller than a standard Gaussian random variable $\varepsilon_0^2 \in \mathcal{N}(0,1)$. 
Similar to the proof of \citet[][Proposition A2]{wang2022rerandomization}, 
This then implies that, for any $c > 0$, 
$\sup_{n\ge 1} \E\{ L_{K_n, T_n}^2 \I(L_{K_n, T_n}^2> c) \} \le \E\{ \varepsilon_0^2 \I(\varepsilon_0^2> c) \}$. 
Letting $c\rightarrow \infty$ and applying the dominated convergence theorem, we can know that $\{L_{K_n, T_n}^2: n\ge 1\}$ must be uniformly integrable. 

In the remaining of this section, we will focus on the asymptotic behavior of the variance $v_{K,T}$ of the constrained Gaussian random variable $L_{K,T}$.

\subsection{Technical lemmas and their proofs}

\begin{lemma}\label{lemma:varlbnd}
	For any $a > 0$, we have that
        \begin{align*}
            \Var(L_{K, T}) \ge \pr(\chi_K^2 > a)^T \cdot \frac{a}{K}, 
        \end{align*}
        and 
	\begin{align*}
	& \quad \ \Var(L_{K, T}) \\
	& \leq \min\left\{\frac{a}{K} + \frac{1}{K} \int_{a}^{\infty} \pr(\chi_K^2 > b)^T \deri b, \ 
        1 - \frac{1 - \pr(\chi_K^2 > a)^{T-1}}{K} \int_{a}^{\infty} \pr(\chi_K^2 > b) \deri b \right\}. 
	\end{align*}
\end{lemma}

\begin{proof}[Proof of Lemma \ref{lemma:varlbnd}]
    Let $U_K$ be the first coordinate of a $K$-dimensional random vector uniformly distributed on the $(K-1)$-dimensional unit sphere and
    $\chi^2_{K[1]}, \chi^2_{K[2]}, \ldots, \chi^2_{K[T+1]}$ be i.i.d.\ chi-squared random variables with degrees of freedom $K$, 
    and assume that they are all mutually independent. 
    Let $\chi_{K,T}^2 = \min_{1\le t \le T} \chi^2_{K[t]}$. 
    From the proof of Lemma \ref{lemma:nonincreasing}, 
    \begin{align*}
        \Var(L_{K,T}) & = \E(\chi_{K,T}^2) \E(U_K^2) = \frac{\E(\chi_{K,T}^2)}{K} 
        = 
        \frac{1}{K} \int_0^\infty \pr\big(\chi_{K,T}^2 > b \big) \deri b =  
        \frac{1}{K} \int_0^\infty \pr\big(\min_{1\le t \le T} \chi^2_{K[t]} > b \big) \deri b 
        \\
        & = 
        \frac{1}{K} \int_0^\infty \pr(\chi_K^2 > b)^T \deri b, 
    \end{align*}
where $\chi_K^2$ denotes a chi-squared random variable with degrees of freedom $K$. 
Consequently, 
for any fixed $a>0$, we have
\begin{align*}
  \Var(L_{K, T}) = \frac{1}{K} \int_0^a \pr(\chi_K^2 > b)^T \deri b + \frac{1}{K} \int_a^\infty \pr(\chi_K^2 > b)^T \deri b \leq \frac{a}{K} + \frac{1}{K} \int_a^\infty \pr(\chi_K^2 > b)^T \deri b  
\end{align*}
and that
\begin{align*}
    \Var(L_{K, T}) &  \leq \frac{1}{K} \int_0^a \pr(\chi_K^2 > b) \deri b + \frac{\pr(\chi_K^2 > a)^{T-1}}{K} \int_a^\infty \pr(\chi_K^2 > b) \deri b \\
    & = \frac{1}{K} \int_0^\infty \pr(\chi_K^2 > b) \deri b - \frac{1 - \pr(\chi_K^2 > a)^{T-1}}{K} \int_{a}^{\infty} \pr(\chi_K^2 > b) \deri b \\
	& = 1 - \frac{1 - \pr(\chi_K^2 > a)^{T-1}}{K} \int_{a}^{\infty} \pr(\chi_K^2 > b) \deri b, 
\end{align*}
where the last equality holds because $\int_0^\infty \pr(\chi_K^2 > b) \deri b=\E(\chi_K^2) = K$. 
These then imply the upper bound of $\Var(L_{K,T})$ in Lemma \ref{lemma:varlbnd}. 
We then derive the lower bound of $\Var(L_{K,T})$ in Lemma \ref{lemma:varlbnd}: 
\[
\Var(L_{K, T}) \geq \frac{1}{K} \int_0^a \pr(\chi_K^2 > b)^T \deri b \geq \frac{\pr(\chi_K^2 > a)^T}{K} \int_0^a \deri b = \pr(\chi_K^2 > a)^T \cdot \frac{a}{K}.
\]
From the above, Lemma \ref{lemma:varlbnd} holds. 
\end{proof}

\begin{lemma}\label{lemma:expconv}
	There exists a constant $c > 0$ such that, for any integer $T \geq 2$,
	$
	(1 - 1/T)^T \geq c.
	$
\end{lemma}

\begin{proof}[Proof of Lemma \ref{lemma:expconv}]
    Note that 
    $
    \lim_{T \to \infty} (1 - 1/T )^T = e^{-1}.
    $
    Thus, there must exist an integer $T_0\ge 2$ such that $
    (1 - 1/T )^T \ge e^{-1}/2
    $ for all $T \geq T_0$.
    We can the derive Lemma \ref{lemma:expconv} 
    by 
    letting 
    $
    c = \min\big\{\min_{2 \leq T \leq T_0} (1 - 1/T)^T, \ e^{-1} / 2\big\} > 0.
    $
\end{proof}

Equipped with these lemmas, we now give a proof of Theorem~\ref{thm:v_kt} by proving the cases (i) -- (iv) separately. In the rest of the proof of Theorem~\ref{thm:v_kt} we keep the subscript $n$ (e.g., writing ``$K$'' as ``$K_n$'') to emphasize their dependence on sample size $n$ more explicitly.

\subsection{Limiting behaviour when $\lim_{n\rightarrow \infty} \log (T_n) / K_n = \infty$}

\begin{lemma}\label{lem:dimlka}
	As $n \to \infty$, if $\log (T_n) / K_n \to \infty$, then there exists a positive sequence $\{a_n\}$ such that $a_n / K_n \to 0$ and $K_n^{-1} \int_{a_n}^\infty \pr(\chi_{K_n}^2 > b)^{T_n} \deri b \to 0$, which implies that $\Var(L_{K_n, T_n}) \to 0$.
\end{lemma}

\begin{proof}[Proof of Lemma \ref{lem:dimlka}]
    For all $n$, define $p_n = T_n^{-1/2}$, and $a_n$ as the $p_n$-th quantile of the chi-squared distribution with degree of freedom $K_n$, 
    i.e., $p_n = \pr(\chi_{K_n}^2 \leq a_n)$. 
    As $n\to \infty$, 
    because $\log (T_n) / K_n \to \infty$, we must have $T_n \to \infty$. 
    We can then verify that 
    \[
	\lim_{n\rightarrow \infty} (T_n - 1) p_n = \infty, \quad \lim_{n\rightarrow \infty} \log (p_n^{-1}) / K_n = \infty, \quad  \lim_{n \to \infty} p_n = 0. 
    \]
    From \citet[Lemma A17]{wang2022rerandomization},  $a_n / K_n \to 0$. 
    In addition, 
    \begin{align*}
        \frac{1}{K_n} \int_{a_n}^\infty \pr(\chi_{K_n}^2 > b)^{T_n} \deri b 
        & \leq \pr(\chi_{K_n}^2 > a_n)^{T_n - 1} \cdot \frac{1}{K_n}  \int_{a_n}^\infty \pr(\chi_{K_n}^2 > b) \deri b 
        \leq (1 - p_n)^{T_n - 1} \cdot \frac{\E \chi_{K_n}^2}{K_n} \\
        & 
        = (1 - p_n)^{T_n - 1}
        = 
        \{(1 - p_n)^{1/p_n}\}^{(T_n - 1)p_n}
        \to 0.
    \end{align*}
    From Lemma~\ref{lemma:varlbnd}, we then have $\Var(L_{K_n, T_n}) \to 0$ as $n\rightarrow \infty$. 
    Therefore, Lemma \ref{lem:dimlka} holds. 
    \end{proof}

\subsection{Limiting behaviour when $\limsup_{n\rightarrow \infty} \log (T_n) / K_n < \infty$}

\begin{lemma}\label{lem:nodim}
	If $\limsup \log(T_n) / K_n < \infty$, then $\liminf_{n \to \infty} \Var( L_{K_n, T_n} ) > 0$.
\end{lemma}

\begin{proof}[Proof of Lemma \ref{lem:nodim}]
    First, 
    for all $n$, 
    let $p_n = (2T_n)^{-1}$ and $a_n$ be the $p_n$th quantile of the chi-squared distribution with degrees of freedom $K_n$, 
    i.e., 
    $\pr(\chi_{K_n}^2 \leq a_n) = p_n =  (2T_n)^{-1}$. 
    We then have 
    $\limsup \log(p_n^{-1}) / K_n < \infty$. From~\citet[Lemma A22]{wang2022rerandomization}, this implies that 
    \begin{equation}\label{eq:liminfan}
    \liminf_{n \to \infty} a_n / K_n > 0.
    \end{equation}
    From Lemmas~\ref{lemma:varlbnd} and \ref{lemma:expconv}, this further implies that 
    \begin{align*}
        \liminf_{n \to \infty} \Var(L_{K_n, T_n}) 
        & \geq 
        \liminf_{n \to \infty} \Big\{ \pr(\chi_{K_n}^2 > a_n)^{T_n} \cdot \frac{a_n}{K_n} \Big\}
        \ge 
        \liminf_{n \to \infty} \Big( \big[ \{1 - 1/(2T_n)\}^{2T_n} \big]^{1/2} \cdot \frac{a_n}{K_n} \Big)
        \\
        & \ge c^{1/2} \cdot \liminf_{n \to \infty} a_n / K_n > 0, 
    \end{align*}
    where $c>0$ is the constant from Lemma \ref{lemma:expconv}. 
    Therefore, Lemma \ref{lem:nodim} holds. 
\end{proof}

\subsection{Limiting behaviour when $\liminf_{n\rightarrow \infty} \log (T_n) / K_n > 0$}

\begin{lemma}\label{lem:varub}
	If $\liminf_{n \to \infty} \log(T_n) / K_n > 0$, then $\limsup_{n \to \infty} \Var(L_{K_n, T_n}) < 1$.
\end{lemma}

\begin{proof}[Proof of Lemma \ref{lem:varub}]
    We prove Lemma \ref{lem:varub} by contradiction. Suppose that $\liminf_{n \to \infty} \log(T_n) / K_n > 0$, 
    and there exists a subsequence $\{n_j, j = 1, 2, \cdots\}$ such that $\Var(L_{K_{n_j}, T_{n_j}}) \to 1$ as $j \to \infty$. 
    Below we consider two cases, depending on whether $\limsup_{j \to \infty} K_{n_j}$ is finite. 

    We first consider the case in which $\limsup_{j \to \infty} K_{n_j} = \infty$. 
    Thus, there must exist a further ubsequence $\{m_j, j = 1, 2, \cdots\} \subset \{n_j, j = 1, 2, \cdots\}$ such that $K_{m_j} \to \infty$ 
    as $j \to \infty$. 
    For any $j\ge 1$, define $p_j = T_{m_j}^{-1}$ and $a_j$ as the $p_j$th quantile of the chi-squared random variable with degrees of freedom $T_{m_j}$, 
    i.e., $\pr(\chi_{K_{m_j}}^2  \leq a_j) = p_j = T_{m_j}^{-1}$. 
    Then we must have that $\liminf_{j \to \infty} \log(p_j^{-1}) / K_{m_j} > 0$. 
    From \citet[Lemma A23(i)]{wang2022rerandomization},  
    we must have $\limsup_{j \to \infty} a_j/K_{m_j} < 1$. 
    Using Lemma~\ref{lemma:varlbnd} with $T=1$, 
    we can know that the limit inferior of 
    \begin{align*}
        \frac{1}{K_{m_j}} \int_{a_j}^\infty \pr(\chi_{K_{m_j}}^2 > b ) \deri b \geq \Var(L_{K_{m_j}, 1}) - \frac{a_j}{K_{m_j}} 
        = 
        1 - \frac{a_j}{K_{m_j}}
    \end{align*}
    must be positive, 
    where the last equality holds because $L_{K,T}\sim \mathcal{N}(0,1)$ when $T=1$. 
    In addition, 
    because $\liminf_{n \to \infty} \log(T_n) / K_n > 0$, 
    we must have $T_{m_j}\to \infty$ as $j\to \infty$, and consequently 
    \[
    \lim_{j \to \infty} \pr(\chi_{K_{m_j}}^2 > a_j)^{T_{m_j} - 1} = \lim_{j \to \infty} \big(1 - T_{m_j}^{-1} \big)^{T_{m_j} - 1}= e^{-1}. 
    \]
    From Lemma~\ref{lemma:varlbnd}, these imply that 
    \begin{align*}
        \limsup_{j \to \infty} \Var(L_{K_{m_j}, T_{m_j}}) & \leq 1 - \liminf_{j \to \infty} \Big[ \big\{1 - \pr(\chi_{K_{m_j}}^2 > a_j)^{T_{m_j} - 1}\big\} \cdot \frac{1}{K_{m_j}} \int_{a_j}^\infty \pr(\chi_{K_{m_j}}^2 > b) \deri b \Big] \\
    	& = 1 - (1-e^{-1}) \cdot \liminf_{j \to \infty} \frac{1}{K_{m_j}} \int_{a_j}^\infty \pr(\chi_{K_{m_j}}^2 > b) \deri b < 1.
    \end{align*}
    However, this contradicts with that $\lim_{j \to \infty}  \Var(L_{K_{m_j}, T_{m_j}}) = 1$.

    We then consider the case in which $\limsup_{j \to \infty} K_{n_j} < \infty$.  Then there exists a $\bar{K}$ such that $K_{n_j} \leq \bar{K}$ for all $j$. Note that $\liminf_{n\rightarrow \infty} \log (T_n) / K_n > 0$. This immediately implies that there exists a positive constant $c$ and a further subsequence $\{m_j, j =1, 2, \cdots\} \subset \{n_j, j =1, 2, \cdots\}$ such that $ \log (T_{m_j}) / K_{m_j} > c$ for all $j$. 
    Consequently, there must exist a constant $\bar{T} \geq 2$ such that $T_{m_j} \geq \bar{T}$ for all $j$. From Lemma~\ref{lemma:nonincreasing}, we then have
    \[
    1 = \lim_{j \to \infty} \Var(L_{K_{m_j}, T_{m_j}}) \leq 
    \sup_{1 \leq K \leq \bar{K}} \nu_{K, \bar{T}} < 1,
    \]
    which leads to a contradiction.

    From the above, Lemma \ref{lem:varub} holds. 
\end{proof}

\subsection{Limiting behaviour when $\lim_{n\rightarrow \infty} \log (T_n) / K_n = 0$}

\begin{lemma}\label{lem:infpn2}
    If $\lim_{n\rightarrow \infty} \log(T_n) / K_n = 0$, then $\lim_{n\rightarrow \infty} \Var(L_{K_n, T_n}) = 1$.
\end{lemma}

\begin{proof}[Proof of Lemma \ref{lem:infpn2}]
    Note that $\Var(L_{K_n, T_n})\le 1$ as implied by Lemma \ref{lemma:nonincreasing}. 
    It suffices to prove that, when $\lim_{n\rightarrow \infty} \log(T_n) / K_n = 0$, $\liminf_{n\rightarrow \infty} \Var(L_{K_n, T_n}) = 1$. 
    We prove this by contradiction. 
    Suppose that $\lim_{n\rightarrow \infty} \log(T_n) / K_n = 0$ and $\liminf_{n\rightarrow \infty} \Var(L_{K_n, T_n}) < 1$. 
    Then there exists a subsequence $\{n_j, j = 1, 2, \cdots\}$ such that 
    $\Var(L_{K_{n_j}, T_{n_j}})<1$ for all $j$ and 
    $\lim_{j \to \infty} \Var(L_{K_{n_j}, T_{n_j}}) < 1$. 
    From Lemma \ref{lemma:nonincreasing},  we must have $T_{n_j} \geq 2$ for all $j$. 
    Because $\log(T_{n_j}) / K_{n_j} \to 0$ as $j\to \infty$, this then implies that $K_{n_j} \to \infty$ as $j\to \infty$. 
    
    Define $p_j = (K_{n_j} T_{n_j})^{-1}$ and $a_j$ as the $p_j$th quantile of the chi-squared distribution with degrees of freedom $K_{n_j}$, 
    i.e., $\pr(\chi^2_{K_{n_j}}) = p_j = (K_{n_j} T_{n_j})^{-1}$. 
    We can verify that, as $j\to \infty$, 
    $\log(p_j^{-1}) / K_{n_j} \to 0, p_j \to 0$ and $p_j T_{n_j} \to 0$.
    From~\citet[Lemma A24]{wang2022rerandomization}, 
    these imply that 
    $\liminf_{j \to \infty} a_j / K_{n_j} \geq 1$. 
    In addition, 
    \[
    \lim_{j \to \infty} \pr(\chi_{K_{n_j}}^2 > a_j)^{T_{n_j}} = \lim_{j \to \infty} (1 - p_j)^{T_{n_j}} = \lim_{j \to \infty} \big\{ (1 - p_j)^{p_j^{-1}}\big\}^{p_j T_{n_j}} = 1.
    \]
    From Lemma~\ref{lemma:varlbnd}, we then have 
    \begin{align*}
        \liminf_{j \to \infty} \Var(L_{K_{n_j}, T_{n_j}}) & \geq \liminf_{j \to \infty} \big[ \pr(\chi_{K_{n_j}}^2 > a_j)^{T_{n_j}} \cdot a_j / K_{n_j} \big] 
        = \lim_{j \to \infty} \pr(\chi_{K_{n_j}}^2 > a_j)^{T_{n_j}} \cdot \liminf_{j \to \infty} \frac{a_j}{K_{n_j}}\\
        \geq 1,
    \end{align*}
    which contradicts with the assumption that $\lim_{j \to \infty} \Var(L_{K_{n_j}, T_{n_j}}) < 1$. 
    
    From the above, Lemma \ref{lem:infpn2} holds. 
\end{proof}

\subsection{Proof of Theorem~\ref{thm:v_kt}}

\begin{proof}[\bf Proof of Theorem~\ref{thm:v_kt}]
    Theorem~\ref{thm:v_kt}(i)--(iv) are direct consequences of Lemmas \ref{lem:dimlka}--\ref{lem:infpn2}.    
\end{proof}

\section{Proof for the optimal best-choice rerandomization}

\begin{proof}[\bf Proof of Theorem \ref{thm:opt_bcr}]
    Let $\psi_n = \sqrt{1-R^2} \ \varepsilon_0$ and $\psi_n' = \sqrt{1-R^2} \ \varepsilon_0 + \sqrt{R^2} \ L_{K, T}$. 
    From Theorem \ref{thm:v_kt}, 
    under Condition \ref{cond:T_opt} and by Chebyshev's inequality, 
    $L_{K,T} = O_{\pr}(\sqrt{v_{K,T}}) = o_{\pr}(1)$. 
    Because $\limsup_{n\rightarrow\infty} R^2 < 1$, this then implies that 
    $
        \psi_n' - \psi_n = \sqrt{R^2} \ L_{K, T} 
        = O(\sqrt{1-R^2}) \cdot o_{\pr}(1) = o_{\pr}(\sqrt{1-R^2}).
    $
    From \citet[Lemma A27]{wang2022rerandomization}, 
    this further implies that, as $n\to \infty$, 
    $\sup_{c\in \mathbb{R}}|\pr(\psi_n \le c) - \pr(\psi_n' \le c)| \to 0$. 
    From Theorems \ref{thm:asymp} and \ref{thm:asym_equiv}, we then have, as $n\to \infty$, 
    \begin{align*}
        & \quad \ \sup_{c\in \mathbb{R}} \Big| \Pr\big\{ V_{\tau\tau}^{-1/2} (\hat{\tau}_{(1)} - \tau) \leq c \big\}   - 
        \Pr
        \big( \psi_n \leq c\big) \Big|
        \\
        & \le 
        \sup_{c\in \mathbb{R}} \Big| \Pr\big\{ V_{\tau\tau}^{-1/2} (\hat{\tau}_{(1)} - \tau) \leq c \big\}   - 
        \Pr
        \big( \psi_n' \leq c\big) \Big|
        + 
        \sup_{c\in \mathbb{R}} \Big| \Pr\big( \psi_n' \leq c \big)   - 
        \Pr
        \big( \psi_n \leq c\big) \Big|
        \\
        & \to 0. 
    \end{align*}
    Therefore, Theorem \ref{thm:opt_bcr} holds. 
\end{proof}

\begin{proof}[\bf Proof of Theorem \ref{thm:opt_rerand}]
    Using Theorem \ref{thm:v_kt} and 
    following the same analysis as in \citet[Proof of Theorem 6]{wang2022rerandomization} but with $p_n^{-1}$ replaced by $T$, 
    we can immediately derive Theorem \ref{thm:opt_rerand}. 
\end{proof}

\section{Proof for the large-sample inference under the best-choice rerandomization}

\subsection{Technical lemmas}

\begin{lemma}\label{lemma:s_uw_vector}
    Let $\{ (u_i, \bs{w}_i^\top)\in \mathbb{R}^{1+K}: i = 1, 2, \ldots, N \}$ be a finite population of $N\ge 2$ units, 
    with 
    $\bs{w}_i = (w_{1i}, w_{2i}, \ldots w_{Ki})^\top$ and 
    finite population averages and covariance 
    $\bar{u} \equiv N^{-1} \sum_{i=1}^N u_i$, 
    $\bar{\bs{w}} = (\bar{w}_1, \ldots, \bar{w}_K)^\top = N^{-1} \sum_{i=1}^N \bs{w}_i$ and 
    $\bs{S}_{u\bs{w}} = (S_{uw_1}, \ldots, S_{uw_K})^\top =  (N-1)^{-1} \sum_{i=1}^N (u_i - \bar{u}) (\bs{w}_i - \bar{\bs{w}})$. 
    Let $(Z_1, \cdots, Z_N)$ denote a sampling indicator vector for a simple random sample of size $m\ge 2$, 
    with corresponding sample averages and covariance
    $\hat{u} = m^{-1} \sum_{i=1}^N Z_i u_i$, 
    $\hat{\bs{w}} = m^{-1} \sum_{i=1}^N Z_i \bs{w}_i$ and 
    $\bs{s}_{u\bs{w}} = (s_{uw_1}, \ldots, s_{uw_K})^\top = (m-1)^{-1} \sum_{i=1}^N Z_i (u_i - \hat{u}) (\bs{w}_i - \hat{\bs{w}})$. 
    Let $f=m/N$, and for $1\le k\le K$, define 
    \begin{align*}
        \Delta_{u} = \hat{u} - \bar{u}, 
        \quad
        \Delta_{w_k} = \hat{w}_k - \bar{w}_k,
        \quad
        \Delta_{uw_k} = 
        \frac{1}{m}\sum_{i=1}^N Z_i (u_i - \bar{u}) (w_{ki} - \bar{w}_k)
        - 
        \frac{N-1}{N} S_{uw_k}, 
    \end{align*}
    and 
    \begin{align*}
        \sigma^2_u = \frac{1}{N} \sum_{i=1}^N(u_i-\bar{u})^2, \ \ 
        \sigma^2_{w_k} = \frac{1}{N} \sum_{i=1}^N(w_{ki}-\bar{w}_k)^2,
        \ \ 
        \sigma^2_{u\times w_k} = \frac{1}{N} \sum_{i=1}^N\left\{ 
        (u_i - \bar{u}) (w_{ki} - \bar{w}_k)- \frac{N-1}{N}S_{uw_k}
        \right\}^2. 
    \end{align*}
    Then 
    \begin{align*}
        \left\| \bs{s}_{u\bs{w}} - \bs{S}_{u\bs{w}}  \right\|^2_2
        \le 
        12 \sum_{k=1}^K \Delta_{u\times w_k}^2 + 12 \Delta_u^2 \sum_{k=1}^K \Delta_{w_k}^2 
        + 
        \frac{12(1-f)^2}{m^2} \sum_{k=1}^K  S_{uw_k}^2, 
    \end{align*}
    and for any $t > 0$, 
    \begin{align*}
        \pr\left( \Delta_u^2 \ge t \right)
        & \le 
        2 
        \exp\left(
        - \frac{70^2}{71^2} \frac{N f^2 t}{\sigma^2_u}
        \right), 
        \qquad \quad
        \pr\left( \sum_{k=1}^K \Delta_{w_k}^2 \ge t \right)
        \le 
        2 K 
        \exp\left(
        - \frac{70^2}{71^2} \frac{N f^2  t}{\sum_{k=1}^K \sigma^2_{w_k}}
        \right), \\
        \pr\left( \sum_{k=1}^K \Delta_{u\times w_k}^2 \ge t \right)
        & 
        \le 
        2 K 
        \exp\left(
        - \frac{70^2}{71^2} \frac{N f^2  t}{\sum_{k=1}^K \sigma^2_{u\times w_k}}
        \right). 
    \end{align*}
\end{lemma}
\begin{proof}[Proof of Lemma \ref{lemma:s_uw_vector}]
    Lemma \ref{lemma:s_uw_vector} follows directly from \citet[][Lemma A26]{wang2022rerandomization}. 
\end{proof}

\begin{lemma}\label{lemma:s_uw_re}
    Consider the same setting as in Lemma \ref{lemma:s_uw_vector}. 
    For any integer $T\ge 1$, let $\bs{Z}_{[1]}, \ldots, \bs{Z}_{[T]}$ be $T$ mutually independent vectors of sampling indicators for a simple random of size $m$ from the finite population of $N$ units, 
    and define $\bs{s}_{[t]u\bs{w}}$ analogously as in Lemma \ref{lemma:s_uw_vector} for each sampling indicator vector $\bs{Z}_{[t]}$, for $1\le t\le T$. 
    Define further 
    \begin{align*}
        \xi
        & = 
        \frac{\max\{1, \log K, \log T\}}{Nf^2} \sum_{k=1}^K \sigma^2_{u\times w_k} 
        +
        \frac{ \max\{1, \log T\} \cdot \max\{1, \log K, \log T\} }{N^2f^4} \sigma^2_u \sum_{k=1}^K \sigma^2_{w_k}
        \\
        & \quad \ + 
        \frac{(1-f)^2}{N^2 f^2} \sum_{k=1}^K  S_{uw_k}^2. 
    \end{align*}
    Then 
    for any $t \ge 3 \cdot 71^2/70^2$, 
    \begin{align*}
        \pr
        \big( 
        \max_{1\le t\le T}\left\| \bs{s}_{[t]u\bs{w}} - \bs{S}_{u\bs{w}}  \right\|^2_2 > 36 t^2 \xi
        \big)
        & \le 
        6
        \exp\left(
        - \frac{1}{3} \frac{70^2}{71^2} t
        \right). 
    \end{align*}
\end{lemma}

\begin{proof}[Proof of Lemma~\ref{lemma:s_uw_re}]
    For any $t > 0$, by the union bound,
    \begin{align*}
        \pr
        \big( 
        \max_{1\le t\le T}\left\| \bs{s}_{[t]u\bs{w}} - \bs{S}_{u\bs{w}}  \right\|^2_2 > 36 t^2 \xi
        \big)
        \le 
        T \cdot \pr
        \big( 
        \left\| \bs{s}_{[1]u\bs{w}} - \bs{S}_{u\bs{w}}  \right\|^2_2 > 36 t^2 \xi
        \big),
    \end{align*}
    where we use the fact that $\bs{s}_{[t]u\bs{w}}$'s follows the same distribution as $\bs{s}_{u\bs{w}}$ defined in Lemma \ref{lemma:s_uw_vector}. 
    Following the same analysis as in \citet[Proof of Lemma~A31]{wang2022rerandomization} but with $p^{-1}$ replaced by $T$, 
    we can directly obtain Lemma~\ref{lemma:s_uw_re}. 
\end{proof}

\begin{lemma}\label{lemma:V_R2_hat_bound}
	Under the best-choice rerandomization, 
	along the sequence of finite populations with increasing sample size $n$, 
	if $\min\{n_1, n_0\} \ge 2$ when $n$ is sufficiently large, then 
	the estimators $\hat{V}_{\tau\tau}$ and $\hat{R}^2$ satisfy that 
	\begin{align*}
		\hat{V}_{\tau\tau} - V_{\tau\tau} - n^{-1} S_{\tau\setminus \bs{X}}^2
		& = 
		O_{\pr}\left(
		\frac{\xi_{11}^{1/2}}{n_1} + 
		\frac{\xi_{00}^{1/2}}{n_0} + 
		\frac{\xi_{1\bs{w}} + \xi_{0\bs{w}}}{n} + \left\| S_{1\bs{w}} - S_{0\bs{w}} \right\|_2\frac{\xi_{1\bs{w}}^{1/2} + \xi_{0\bs{w}}^{1/2}}{n} 
		\right), 
	\end{align*}
	and 
	\begin{align*}
		\hat{V}_{\tau\tau} \hat{R}^2_{n} - V_{\tau\tau} R^2_n
		& = 
		O_{\pr}\left(
		\frac{\xi_{1\bs{w}}}{n_1} 
		+ 
		\frac{\xi_{0\bs{w}}}{n_0}
		+ \left\| S_{1\bs{w}} \right\|_2 \frac{\xi_{1\bs{w}}^{1/2}}{n_1}
		+ \left\| S_{0\bs{w}} \right\|_2 \frac{\xi_{0\bs{w}}^{1/2}}{n_1} 
		+ 
		\left\| S_{1\bs{w}} - S_{0\bs{w}} \right\|_2 \frac{\xi_{1\bs{w}}^{1/2} + \xi_{0\bs{w}}^{1/2}}{n}
		\right), 
	\end{align*}
	where 
	$\bs{w}_i = (w_{1i}, \ldots, w_{K_n i})^\top = \bs{S}_{\bs{X}}^{-1}(\bs{X}_i - \bar{\bs{X}})$ is the standardized covariates, 
	$S_{z\bs{w}} = (S_{zw_1}, \ldots, S_{zw_{K}})$ is the finite population covariance between $Y(z)$ and $\bs{w}$, 
	\begin{align*}
		\xi_{zz}
		& = 
		\frac{\max\{1, \log T \}}{n r_z^2} \sigma^2_{z\times z} 
		+
		\frac{ \max\{1, (\log T)^2 \} }{n^2 r_z^4} \sigma^4_z 
		+ 
		\frac{(1-r_z)^2}{n^2 r_z^2} S_{z}^4, 
		\\
		\xi_{z \bs{w}}
		& = 
		\frac{\max\{1, \log K_n, \log T\}}{n r_z^2} \sum_{k=1}^K \sigma^2_{z\times w_k} 
		+
		\frac{ \max\{1, \log T \} \cdot \max\{1, \log K, \log T \} }{n^2 r_z^4} \sigma^2_u \sum_{k=1}^{K} \sigma^2_{w_k}
		\\
		& \quad \ + 
		\frac{(1-r_z)^2}{n^2 r_z^2} \sum_{k=1}^{K} S_{z w_k}^2,  
	\end{align*}
	and 
	\begin{align*}
		\sigma^2_z & = \frac{1}{n} \sum_{i=1}^n \{Y_i(z) - \bar{Y}(z)\}^2 = \frac{n-1}{n} S_z^2, 
		\qquad
		\sigma_{w_k}^2 
		= \frac{1}{n} \sum_{i=1}^n ( w_{ki} - \bar{w}_k )^2  = \frac{n-1}{n}, 
		\\
		\sigma^2_{z\times z} 
		& = 
		\frac{1}{n} \sum_{i=1}^n\Big[ 
		\{Y_i(z) - \bar{Y}(z) \}^2 - \sigma^2_z
		\Big]^2,
		\ \ 
		\sigma^2_{z\times w_k} 
		= 
		\frac{1}{n} \sum_{i=1}^n\Big[ 
		\{Y_i(z) - \bar{Y}(z)\} (w_{ki} - \bar{w}_k)- \frac{n-1}{n} S_{z w_k}
		\Big]^2. 
	\end{align*}
\end{lemma}

\begin{proof}[Proof of Lemma~\ref{lemma:V_R2_hat_bound}]
    Define analogously $s_{[t]z}^2$ and $s_{[t]z\bs{w}}$ for the $t$-th completely randomized treatment assignment under the best-choice rerandomization, for $1\le t\le T$. 
    From Lemma~\ref{lemma:s_uw_re} and by the Markov inequality, we can know that,  under the best-choice rerandomization, 
    \begin{align*}
        | s_{z}^2 - S_z^2 | \le \max_{1\le t\le T} | s_{[t]z}^2 - S_z^2 | = O_{\pr}\left( \xi_{zz}^{1/2}\right), 
        \quad
        \| s_{z\bs{w}} - S_{z\bs{w}} \|_2 \le \max_{1\le t\le T} \| s_{[t]z\bs{w}} - S_{z\bs{w}} \|_2 = O_{\pr}\left( \xi_{z\bs{w}}^{1/2}\right).  
    \end{align*}
    Then 
    following the same analysis as in \citet[Proof of Lemma~A32]{wang2022rerandomization}, 
    we can directly obtain Lemma~\ref{lemma:V_R2_hat_bound}.
\end{proof}

\begin{lemma}\label{lemma:V_R2_hat_bound_simp}
	Under the same setting as Lemma \ref{lemma:V_R2_hat_bound}, 
	if $\max\{1, \log K, \log T \} = O(nr_1^2 r_0^2)$, 
	then 
	\begin{align*}
		& \quad \ \max\left\{ \big| \hat{V}_{\tau\tau} - V_{\tau\tau} - n^{-1} S_{\tau\setminus \bs{X}}^2 \big|, \ \ \big| \hat{V}_{\tau\tau} \hat{R}^2 - V_{\tau\tau} R^2 \big| \right\}
		\\
		& 
		= 
		\max_{z\in \{0,1\}}\max_{1\le i \le n}\{Y_i(z) - \bar{Y}(z)\}^2 \cdot 
		O_{\pr}\left( 
		\max\{K, 1\} \cdot \frac{\sqrt{ \max\{1, \log K, \log T\} }}{n^{3/2} r_1^2r_0^2} \right). 
	\end{align*}
\end{lemma}

\begin{proof}[Proof of Lemma~\ref{lemma:V_R2_hat_bound_simp}]
    This follows from the same analysis as in \citet[Proof of Lemma A33]{wang2022rerandomization} but with $b_n$ and $c_n$ there redefined as $b_n = \max\{1, \log T\}$ and $c_n = \max\{1, \log K, \log T\}$.
\end{proof}

\begin{lemma}\label{lemma:cond_infer}
	Under the same setting as Lemmas \ref{lemma:V_R2_hat_bound} and \ref{lemma:V_R2_hat_bound_simp}, 
	\begin{itemize}
		\item[(i)]
		$\max_{z\in \{0,1\}}\max_{1\le i \le n}\{Y_i(z) - \bar{Y}(z)\}^2/
		(r_0 S^2_{1\setminus \bs{x}} + r_1 S^2_{0\setminus \bs{x}})
		\ge 1/2$; 
		\item[(ii)] 
		if Conditions \ref{cond:iterations} and \ref{cond:infer} hold, 
		then,  $\max\{1, \log K, \log T \} = o(nr_1^2 r_0^2)$. 
	\end{itemize}
\end{lemma}

\begin{proof}[Proof of Lemma~\ref{lemma:cond_infer}]
    (i) follows directly from \citet[Lemma~A34(ii)]{wang2022rerandomization}, 
    and (ii) follows from the same analysis as in the \citet[Proof of Lemma~A34(iii)]{wang2022rerandomization} with $-\log \tilde{p}_n$ there replaced by $\log T$. 
\end{proof}

\begin{lemma}\label{lemma:non_Gaussian_quantile}
Let $\varepsilon_0 \sim \mathcal{N}(0,1)$, and define $L_{K_n, T_n}$ as in \eqref{eq:L_KT} for all $n$,
where $\{K_n\}$ and $\{T_n\}$ are sequences of positive integers,
and $\varepsilon_0$ is independent of $L_{K_n, T_n}$ for all $n$. 
Let $\{A_n\}$,  $\{B_n\}$, $\{\tilde{A}_n\}$ and $\{\tilde{B}_n\}$  be sequences of nonnegative constants, 
and for each $n$, 
define $\psi_n = A_n^{1/2} \cdot \varepsilon_0 + B_n^{1/2} \cdot L_{K_n, T_n}$ and $\tilde{\psi}_n = \tilde{A}_n^{1/2} \cdot \varepsilon_0 + \tilde{B}_n^{1/2} \cdot L_{K_n, T_n}$. 
For each $n$ and $\alpha \in (0,1)$, let $q_{n}(\alpha)$ and $\tilde{q}_{n}(\alpha)$ be the $\alpha$th quantile of $\psi_n$ and $\tilde{\psi}_n$, respectively. 
If $\max\{ | \tilde{A}_n - A_n|, |\tilde{B}_n - B_n|\} = o(A_n)$, 
then for any $0 < \alpha < \beta < 1$, as $n \rightarrow \infty$, 
$\I \{ \tilde{q}_n(\beta) \le q_n(\alpha) \} \rightarrow 0$ 
and 
$\I \{ q_n(\beta) \le \tilde{q}_n(\alpha) \} \rightarrow 0$. 
\end{lemma}

\begin{proof}[Proof of Lemma~\ref{lemma:non_Gaussian_quantile}]
    From Lemma \ref{lemma:nonincreasing}, $\Var(L_{K_n, T_n}) \le 1$ for all $n$, 
    and thus $L_{K_n, T_n}=O_{\pr}(1)$. 
    Lemma~\ref{lemma:non_Gaussian_quantile} then follows from the same analysis as in \citet[Lemma A37]{wang2022rerandomization} with $L_{K_n, a_n}$ there replaced by $L_{K_n, T_n}$.
\end{proof}

\subsection{Proof of Theorem~\ref{thm:inf}}

\begin{proof}[\bf Proof of Theorem~\ref{thm:inf}(i)]
    Theorem~\ref{thm:inf}(i) follows from the same analysis as in \citet[Proof of Theorem 7(i)]{wang2022rerandomization}, but with Lemmas A33 and A34 there replaced by Lemmas~\ref{lemma:V_R2_hat_bound_simp} and~\ref{lemma:cond_infer}.
\end{proof}

\begin{proof}[\bf Proof of Theorem~\ref{thm:inf}(ii)]
    Let $\varepsilon_0\sim \mathcal{N}(0,1)$ and $L_{K, T}$ be the constrained Gaussian random variable as in Proposition \ref{prop:L_KT}, 
    and assume that they are mutually independent and also independent of the $T$ completely randomized treatment assignments for the best-choice rerandomization. 
    Define 
    \begin{align*}
        \theta_n & = \sqrt{V_{\tau\tau} (1-R^2)} \cdot \varepsilon_0 + \sqrt{V_{\tau\tau} R^2} \cdot L_{K, T}
        \equiv A_n^{1/2} \cdot \varepsilon_0 + B_n^{1/2} \cdot L_{K, T}, 
        \\ 
        \tilde{\theta}_n & = \sqrt{V_{\tau\tau} (1-R^2)+n^{-1} S^2_{\tau \setminus \bs{x}}} \cdot \varepsilon_0 + \sqrt{V_{\tau\tau} R^2} \cdot L_{K, T}
        \equiv \tilde{A}_n^{1/2} \cdot \varepsilon_0 + \tilde{B}_n^{1/2} \cdot L_{K, T}, 
        \\
        \hat{\theta}_n & = \sqrt{\hat{V}_{\tau\tau} (1-\hat{R}^2)} \cdot \varepsilon_0 + \sqrt{\hat{V}_{\tau\tau} \hat{R}^2} \cdot L_{K, T}
        \equiv \hat{A}_n^{1/2} \cdot \varepsilon_0 + \hat{B}_n^{1/2} \cdot L_{K, T}.
    \end{align*}
    Define further $q_{\alpha}(A, B, K, T)$ as the $\alpha$th quantile of $A^{1/2} \varepsilon_0 + B^{1/2} L_{K,T}$, and let $q_{n,\alpha} = q_{\alpha}(A_n, B_n, K, T)$, 
    $\tilde{q}_{n,\alpha} = q_{\alpha}(\tilde{A}_n, \tilde{B}_n, K, T)$, 
    and 
    $\hat{q}_{n,\alpha} = q_{\alpha}(\hat{A}_n, \hat{B}_n, K, T)$. 
    From Theorem \ref{thm:inf}(i), 
    under the best-choice rerandomization, 
    $\max\{|\hat{A}_n - \tilde{A}_n|, |\hat{B}_n - \tilde{B}_n|\} = o_{\pr}(\tilde{A}_n)$. 
    Applying Lemma~\ref{lemma:non_Gaussian_quantile} and following the same analysis as in \citet[proof of Theorem 7(ii)]{wang2022rerandomization}, 
    we can derive that, under the best-choice rerandomization, 
    for any $0<\alpha < \beta < 1$, 
    $\pr( \hat{q}_{n, \beta} \leq \tilde{q}_{n, \alpha} ) \to 0$ as $n \to \infty$. 

    For any $\alpha \in (0,1)$ and $\eta \in (0, (1 - \alpha) / 2)$, 
    following the same analysis as in \citet[Proof of Theorem 7(ii)]{wang2022rerandomization}, 
    the coverage probability of the confidence interval $\hat{\mathcal{C}}_{\alpha}$ can be bounded by 
    \begin{align*}
        \pr(\tau \in \hat{\mathcal{C}}_{\alpha}) 
        & = 
        \pr(| \hat{\tau}_{(1)} - \tau|  \le \hat{q}_{n, 1-\alpha/2}) 
        \ge 
        \pr\{ | \hat{\tau}_{(1)} - \tau|  \le  \tilde{q}_{n, 1-\alpha/2-\eta}\}
        - 
        \pr\{  \hat{q}_{n, 1-\alpha/2} < \tilde{q}_{n, 1-\alpha/2-\eta}\}.
    \end{align*}
    From Theorems \ref{thm:asymp} and \ref{thm:asym_equiv}, 
    $\pr\{ | \hat{\tau}_{(1)} - \tau|  \le  \tilde{q}_{n, 1-\alpha/2-\eta}\} = \pr\{ | \theta_n |  \le  \tilde{q}_{n, 1-\alpha/2-\eta}\} + o(1)$, and 
    from the discussion before, $\pr\{  \hat{q}_{n, 1-\alpha/2} < \tilde{q}_{n, 1-\alpha/2-\eta}\} = o(1)$. 
    These then imply that 
    \begin{align*}
        \pr(\tau \in \hat{\mathcal{C}}_{\alpha}) 
        & \ge 
        \pr\{ | \theta_n |  \le  \tilde{q}_{n, 1-\alpha/2-\eta}\}  + o(1).
    \end{align*}
    From Lemma \ref{lemma:lktunimodal}, $L_{K,T}$ is continuous, and is also symmetric and unimodal around zero. 
    Applying Lemma \ref{lemma:sum} with $\zeta_0 = B_n^{1/2} L_{K,T}$, $\zeta_1 = A_{n}^{1/2} \varepsilon_0$ and $\zeta_2 = \tilde{A}_{n}^{1/2} \varepsilon_0$, we then have 
    $\pr\{ | \theta_n |  \le  \tilde{q}_{n, 1-\alpha/2-\eta}\}\ge \pr\{ | \tilde{\theta}_n |  \le  \tilde{q}_{n, 1-\alpha/2-\eta}\} = 1 - \alpha - 2\eta$. 
    Consequently, 
    $\pr(\tau \in \hat{\mathcal{C}}_{\alpha}) \ge 1 - \alpha - 2\eta + o(1)$, and thus 
    $\liminf_{n\to \infty} \pr(\tau \in \hat{\mathcal{C}}_{\alpha}) \ge 1 - \alpha - 2\eta$. 
    Because this inequality holds for any $\eta \in (0, (1 - \alpha) / 2)$, 
    we must have 
    $
    \liminf_{n\to \infty}\pr(\tau \in \hat{\mathcal{C}}_{\alpha}) \ge 1 - \alpha. 
    $

    From the above, Theorem \ref{thm:inf}(ii) holds. 
\end{proof}

\begin{proof}[\bf Proof of Theorem \ref{thm:inf}(iii)]
    We adopt the notation 
    Following the same analysis as in \citet[Proof of Theorem 7(ii)]{wang2022rerandomization}, 
    for any $\alpha\in (0,1)$ and $\eta \in (0,\alpha/2)$, 
    \begin{align*}
        \pr(\tau \in \hat{\mathcal{C}}_{\alpha})
        & \le 
        \pr( |\hat{\tau}_{(1)} - \tau | \le \tilde{q}_{n, 1-\alpha/2+\eta} ) 
        + 
        \pr(\hat{q}_{n, 1-\alpha/2}> \tilde{q}_{n, 1-\alpha/2+\eta})\\
        & = \pr( |\theta_n | \le \tilde{q}_{n, 1-\alpha/2+\eta} ) 
        + o(1) + 
        \pr(\hat{q}_{n, 1-\alpha/2}> \tilde{q}_{n, 1-\alpha/2+\eta})\\
        & = \pr( |\theta_n | \le \tilde{q}_{n, 1-\alpha/2+\eta} ) 
        + o(1),
    \end{align*}
    where the second last equality follows from Theorems \ref{thm:asymp} and \ref{thm:asym_equiv}, 
    and 
    the last equality follows from the same logic as in the proof of Theorem \ref{thm:inf}(ii) and Lemma \ref{lemma:non_Gaussian_quantile}.

    Under the condition in Theorem \ref{thm:inf}(iii) and following the same analysis as in \citet[Proof of Theorem 7(ii)]{wang2022rerandomization}, we can derive that $\tilde{\theta}_n - \theta_n = \sqrt{V_{\tau\tau}(1-R^2)} \cdot o_{\pr}(1)$, which, by \citet[Lemma A27]{wang2022rerandomization}, implies that 
    $\sup_{c\in \mathbb{R}} |\pr( \theta_n  \le c )-\pr( \tilde{\theta}_n  \le c )|\to 0$ as $n\to \infty$. 
    Consequently, we have 
    \begin{align*}
        \pr(\tau \in \hat{\mathcal{C}}_{\alpha})
        & \le 
        \pr( |\tilde{\theta}_n | \le \tilde{q}_{n, 1-\alpha/2+\eta} ) 
        + o(1)
        = 
        1 - \alpha + 2 \eta + o(1). 
    \end{align*}
    Because this inequality holds for any $\eta \in (0,\alpha/2)$, 
    we can derive that $\limsup_{n\to \infty} \pr(\tau \in \hat{\mathcal{C}}_{\alpha}) \le  1 - \alpha$. 
    From Theorem \ref{thm:inf}(ii), 
    we then have $\lim_{n\to \infty} \pr(\tau \in \hat{\mathcal{C}}_{\alpha}) =  1 - \alpha$.
    Therefore, Theorem \ref{thm:inf}(iii) holds. 
\end{proof}

\begin{lemma}\label{lemma:non_Gaussian_quantile_optimal}
    Let $\varepsilon_0 \sim \mathcal{N}(0,1)$, and define $L_{K_n, T_n}$ as in \eqref{eq:L_KT} for all $n$,
    where $\{K_n\}$ and $\{T_n\}$ are sequences of positive integers,
    and $\varepsilon_0$ is independent of $L_{K_n, T_n}$ for all $n$. 
    Let $\{A_n\}$,  $\{B_n\}$, $\{\tilde{A}_n\}$ and $\{\tilde{B}_n\}$  be sequences of nonnegative constants, 
    and for each $n$, 
    define $\psi_n = A_n^{1/2} \cdot \varepsilon_0 + B_n^{1/2} \cdot L_{K_n, T_n}$ and $\tilde{\psi}_n = \tilde{A}_n^{1/2} \cdot \varepsilon_0 + \tilde{B}_n^{1/2} \cdot L_{K_n, T_n}$. 
    For each $n$ and $\alpha \in (0,1)$, let $q_{n}(\alpha)$ and $\tilde{q}_{n}(\alpha)$ be the $\alpha$th quantile of $\psi_n$ and $\tilde{\psi}_n$, respectively. 
    If $L_{K_n, T_n}=o_{\pr}(1)$, $\tilde{A}_n - A_n = o(A_n)$ and $|\tilde{B}_n - B_n| = O(A_n)$, 
    then for any $0 < \alpha < \beta < 1$, as $n \rightarrow \infty$, 
    $\I \{ \tilde{q}_n(\beta) \le q_n(\alpha) \} \rightarrow 0$ 
    and 
    $\I \{ q_n(\beta) \le \tilde{q}_n(\alpha) \} \rightarrow 0$. 
\end{lemma}

\begin{proof}[Proof of Lemma \ref{lemma:non_Gaussian_quantile_optimal}]
    Lemma \ref{lemma:non_Gaussian_quantile_optimal} follows by the same logic as \citet[Lemma A37]{wang2022rerandomization}.
\end{proof}

\begin{proof}[\bf Proof Theorem~\ref{thm:inf_gaussian}]
    Adopting the notation from the proof of Theorem \ref{thm:inf}, define further 
    \begin{align*}
        \check{\theta}_n = \sqrt{\hat{V}_{\tau\tau}(1-\hat{R}^2)} \cdot \varepsilon_0 + 0 \cdot L_{K,T} 
        \equiv 
        \check{A}_n^{1/2} \cdot \varepsilon_0 + \check{B}_n^{1/2} \cdot L_{K,T},
    \end{align*}
    and $\check{q}_{n, \alpha} = q_{\alpha}(\check{A}_n, \check{B}_n, K, T)$. 
    We can then prove Theorem~\ref{thm:inf_gaussian} by the same logic as Theorem \ref{thm:inf}, by replacing $\hat{q}_{n, \alpha}$ with  $\check{q}_{n, \alpha}$ and applying Lemma \ref{lemma:non_Gaussian_quantile_optimal}. 
    We omit the detailed proof for conciseness. 
\end{proof}

{\rev 
\section{Proof for rerandomization with regression adjustment}

Throughout this section, we define $\hat{\tau}_{(1)}(\tilde{\bs{\beta}}_1, \tilde{\bs{\beta}}_0)$ in the same way as $\hat{\tau}_{(1)}(\hat{\bs{\beta}}_1, \hat{\bs{\beta}}_0)$ in~\eqref{eq:taureg}, but with $\hat{\bs{\beta}}_z$ replaced by $\tilde{\bs{\beta}}_z$. Additionally, we let $\hat{\bs{\tau}}_{[t]\bs{w}}$ denote the difference-in-means of the $\bs{w}_i$'s under treatment assignment $\bs{Z}_{[t]}$, and define $\hat{\bs{\tau}}_{(1)\bs{w}}$ analogously. Armed with $\hat{\bs{\tau}}_{(1)\bs{w}}$, we can rewrite~\eqref{eq:taureg} as
\[
\hat{\tau}_{(1)}(\hat{\bs{\beta}}_1, \hat{\bs{\beta}}_0) = \hat{\tau}_{(1)} - (r_0 \hat{\bs\beta}_1 + r_1 \hat{\bs\beta}_0)^\top \hat{\bs{\tau}}_{(1)\bs{w}}.
\]
We now invoke the following lemmas for the proof of Theorem~\ref{thm:regrem}.

\begin{lemma}\label{lem:regrem}
	Under the best choice rerandomization with Mahalanobis distance, 
    if Condition \ref{cond:regrem_Delta} holds, then
    as $n \to \infty$,
	\begin{align*}%
		\sup_{c\in \mathbb{R}} & \bigg| \Pr \big\{
		V_{\tau\tau}^{-1/2} (1-\rho^2)^{-1/2}
		\{ \hat{\tau}_{(1)}(\tilde{\bs{\beta}}_1, \tilde{\bs{\beta}}_0) - \tau\} \le c \big\} \\
		&  - \Pr\left( \sqrt{1-\tilde{R}^2}\ \varepsilon_0  + \sqrt{\tilde{R}^2} \ L_{K, T} \le c \right)
		\bigg| 
		\rightarrow 0.
	\end{align*}
\end{lemma}

\begin{proof}[Proof of Lemma \ref{lem:regrem}]
This follows from exactly the same proof as~\citet[Lemma~A42]{wang2022rerandomization}, except that we replace \citet[Theorem~3]{wang2022rerandomization} by Theorems \ref{thm:asymp} and \ref{thm:asym_equiv}.
\end{proof}

\begin{lemma}\label{lemma:beta_tauw}
	Under the best-choice rerandomization, 
	if $\min\{n_1,n_0\}\ge 2$ when $n$ is sufficiently large, and 
	$\max\{1, \log J, \log T \} = O(n r_1^2 r_0^2)$, 
	then
	\begin{align*}
	& \quad \ \big\{ r_0 (\hat{\bs{\beta}}_1 - \tilde{\bs{\beta}}_1) + r_1 (\hat{\bs{\beta}}_0 - \tilde{\bs{\beta}}_0)\big\}^\top \hat{\bs{\tau}}_{(1)\bs{w}} \\
	& = 
    O_{\Pr}\left( \max_{z \in \{0,1\}} \max_{1 \leq i \leq n} |Y_i(z) - \bar{Y}(z)| \cdot J \frac{\max\{1, \log J, \log T\}}{nr_1^2 r_0^2} \right). 
	\end{align*}
\end{lemma}

\begin{proof}[Proof of Lemma \ref{lemma:beta_tauw}]
Following the same logic as the proof of~\citet[Lemma~A44]{wang2022rerandomization}, it remains to prove that for $z = 0, 1$,
\begin{align*}
        \| \bs{s}_{z\bs{w}} - \bs{S}_{z\bs{w}} \|_2^2 
        & = 
        \max_{z \in \{0,1\}} \max_{1 \leq i \leq n} \{Y_i(z) - \bar{Y}(z)\}^2 \cdot J \frac{\max\{1, \log J, \log T\}}{nr_z^2} \cdot O_{\Pr}(1)
\end{align*}
and that
\begin{align*}
        \| \bar{\bs{w}}_{z} - \bar{\bs{w}} \|_2^2 
        & = 
        J \frac{\max\{1, \log J, \log T\}}{n r_z^2} \cdot O_{\Pr}(1),
    \end{align*}
where, without the loss of generality, the $\bs{w}_i$'s are assumed to have an identity finite population covariance matrix, as implied by the proof of~\citet[Lemma~A44]{wang2022rerandomization}.
For the first result, from Lemma~\ref{lemma:s_uw_re} and by the same logic as the proof of \citet[Lemma~A33]{wang2022rerandomization}, we have
\begin{align*}
 \| \bs{s}_{z\bs{w}} - \bs{S}_{z\bs{w}} \|_2^2 & \le \max_{1 \le t \le T} \| \bs{s}_{[t]z\bs{w}} - \bs{S}_{z\bs{w}} \|_2^2 \\
 & = \max_{z \in \{0,1\}} \max_{1 \leq i \leq n} \{Y_i(z) - \bar{Y}(z)\}^2 \cdot J \frac{\max\{1, \log J, \log T\}}{n r_z^2} \cdot O_{\Pr}(1).
\end{align*}
For the second result, noting again that 
\[
 \| \bar{\bs{w}}_{z} - \bar{\bs{w}} \|_2^2 \le \max_{1 \le t \le T} \| \bar{\bs{w}}_{[t]z} - \bar{\bs{w}} \|_2^2,
\]
and the fact that for any $c > 0$, 
\[
\pr\left(\max_{1 \le t \le T} \| \bar{\bs{w}}_{[t]z} - \bar{\bs{w}} \|_2^2 > c\right) \le T \cdot \pr\left(\| \bar{\bs{w}}_{[1]z} - \bar{\bs{w}} \|_2^2 > c\right),
\]
the desired result then follows from the same analysis as in \citet[Lemma~A43]{wang2022rerandomization}, but with $p^{-1}$ replaced by $T$ (note that the same trick has also been used in the proof of Lemma~\ref{lemma:s_uw_re}).
\end{proof}

\begin{lemma}\label{lemma:beta_tauw_cond}
	Under the best-choice rerandomization, if Conditions \ref{cond:regrem_Delta} and \ref{cond:regrem} hold, then 
    \[
    \max\{1, \log K, \log T \} = o(nr_1^2 r_0^2).
    \]
\end{lemma}

\begin{proof}
    This follows from exactly the same proof as \citet[Lemma~A45]{wang2022rerandomization}, but with $\tilde{p}_n^{-1}$ replaced by $T$.
\end{proof}

\begin{proof}[\bf Proof of Theorem~\ref{thm:regrem}(i)]
    Similar to 
    the proof of \citet[Theorem~A1]{wang2022rerandomization}, we define 
\begin{align*}
    \tilde{\psi}_n & = V_{\tau\tau}^{-1/2} (1-\rho^2)^{-1/2}
	\{ \hat{\tau}_{(1)}(\tilde{\bs{\beta}}_1, \tilde{\bs{\beta}}_0) - \tau\}, 
	\qquad 
	\hat{\psi}_n = V_{\tau\tau}^{-1/2} (1-\rho^2)^{-1/2}
	\{ \hat{\tau}_{(1)}(\hat{\bs{\beta}}_1, \hat{\bs{\beta}}_0) - \tau\}, \\
	\psi_n & = \sqrt{1-\tilde{R}^2}\ \varepsilon_0  + \sqrt{\tilde{R}^2} \ L_{K, T}.
\end{align*}
Following the same logic as the proof of~\citet[Theorem~A1(i)]{wang2022rerandomization}, but with \citet[Lemmas~A42, A44-A45]{wang2022rerandomization} replaced by Lemmas~\ref{lem:regrem}--\ref{lemma:beta_tauw_cond}, and \citet[Conditions~A1 and~A2]{wang2022rerandomization} replaced by Conditions~\ref{cond:regrem_Delta} and~\ref{cond:regrem}, it remains to prove that for any $\eta > 0$, with $\delta_n \equiv  \sqrt{1 - \tilde{R}^2} \ \eta$,
\begin{align*} 
	\limsup_{n\rightarrow \infty}
	\sup_{b\in \mathbb{R}}\Pr(b < \psi_n \le b+\delta_n) \le \eta / \sqrt{2\pi}.
\end{align*}
For any $b$, by the same analysis as in \citet[Theorem~A1(i)]{wang2022rerandomization}, but with $L_{K, a}$ replaced by $L_{K,T}$, we have 
\begin{align*}
    \Pr(b < \psi_n \le b+\delta_n) \le \eta / \sqrt{2\pi},
\end{align*}
thereby proving the desired result.
\end{proof}

\begin{proof}[\bf Proof of Theorem~\ref{thm:regrem}(ii)] In light of Theorem~\ref{thm:opt_bcr}, the desired result follows by the same logic as the proof of~\citet[Theorem~A1(ii)]{wang2022rerandomization}.
\end{proof}}

\end{document}